\documentclass[letterpaper,11pt]{article}

\usepackage[margin=1in]{geometry}
\usepackage[utf8]{inputenc}

\usepackage{amsmath}
\usepackage{amssymb}
\usepackage{amsthm}
\usepackage{tikz}
\usepackage{subcaption}
\usepackage{enumitem}

\usetikzlibrary{matrix}
\usetikzlibrary{decorations.pathreplacing}

\newtheorem{theorem}{Theorem}
\newtheorem{lemma}{Lemma}

\newtheorem{proposition}[theorem]{Proposition}

\newtheorem{corollary}{Corollary}
\theoremstyle{definition}
\newtheorem{definition}{Definition}

\newtheorem{problem}{Problem}
  
  \usepackage{amsthm}
\usepackage{thmtools}
\usepackage{thm-restate}  
  
\usepackage{algorithm}
\usepackage[noend]{algorithmic}

\usepackage[colorlinks=true,linkcolor=blue,urlcolor=blue,citecolor=blue]{hyperref}

\usepackage{xspace}
\usepackage{bbm}

\usepackage{accents}
\usepackage{todonotes}
\newcommand{\var}{\mathit{var}}
\newcommand{\OrEstimateSample}{\mathsf{estimateSample}_+}
\newcommand{\AndEstimateSample}{\mathsf{estimateSample}_\times}

\newcommand{\approxMCDNNFcore}{\textup{\textsf{countCore}}}
\newcommand{\approxMCDNNF}{\textup{\textsf{counter}}}

\newcommand{\nSamples}{m}
\newcommand{\tree}{\mathit{tree}}
\newcommand{\desc}{\mathit{desc}}
\newcommand{\prev}{\mathit{children}}
\newcommand{\children}{\mathit{children}}
\newcommand{\round}{\mathsf{roundDown}}
\newcommand{\lcsn}[2]{#1 \sqcap #2}

\newcommand{\lcsnvar}[3]{#1 \sqcap_{#3} #2}
\newcommand{\Ex}{\textsf{E}}

\newcommand{\Va}{\textsf{Var}}
\newcommand{\calE}{\mathcal{E}}
\newcommand{\calP}{\mathcal{P}}

\newcommand{\calH}{\mathcal{H}}

\newcommand{\calA}{\mathcal{A}}
\newcommand{\calL}{\mathcal{L}}
\newcommand{\calZ}{\mathcal{Z}}
\newcommand{\supp}{\mathit{supp}}
\newcommand{\nsnt}{{n_sn_t}}
\newcommand{\mods}{\mathit{models}}
\newcommand{\threshold}{\theta}
\newcommand{\median}{\mathsf{median}}
\newcommand{\reduce}{\mathsf{reduce}}
\newcommand{\premises}{\mathit{premises}}
\newcommand{\descendants}{\mathit{desc}}
\newcommand{\union}{\mathsf{union}}
\newcommand{\vY}[4]{\mathfrak{S}^{#1}_{#2,#3}(#4)}
\newcommand{\vZ}[3]{\hat{\mathfrak{S}}^{#1}_{#2}(#3)}

\usepackage[most]{tcolorbox}
\usepackage[customcolors]{hf-tikz}
\usetikzlibrary{calc,patterns}

\newtheorem{fact}{Fact}
\newtheorem{example}{Example}
\usepackage{thm-restate}
\usepackage{subcaption}
\usepackage{bbm}

\begin{document}

\newcommand\relatedversion{}
\renewcommand\relatedversion{\thanks{The authors decided to forgo the old convention of alphabetical ordering of authors in favor of a randomized ordering, denoted by \textcircled{r}. The publicly verifiable record of the randomization is available at
	\protect \url{https://www.aeaweb.org/journals/policies/random-author-order/search}
	} \thanks{The current version expands the exposition of our techniques and fixes minor errors in Algorithms 3 and 4 and their corresponding proofs. The current version expands the exposition of our techniques and fixes minor errors in Algorithms 3 and 4 and their corresponding proofs. These fixes were incorporated in the camera-ready version published at SODA 2026. }} %

\title{\Large \#CFG and \#DNNF admit FPRAS\relatedversion}

\author{Kuldeep S. Meel$^{1}$ \and \textcircled{r} \and Alexis de Colnet$^2$}
\date{\small $^1$University of Toronto, Canada and Georgia Institute of Technology, USA  \\ $^2$TU Wien, Austria}

\maketitle

\begin{abstract}
	We provide the first fully polynomial-time randomized approximation scheme for the following two counting problems: 
	\begin{enumerate}
		\item Given a Context-Free Grammar $G$ over alphabet $\Sigma$, count the number of words of length exactly $n$ generated by $G$
		\item Given a circuit $\varphi$ in Decomposable Negation Normal Form (DNNF) over the set of Boolean variables $X$, compute the number of assignments to $X$ such that $\varphi$ evaluates to 1. 
	\end{enumerate}
	
Finding polynomial time algorithms for the aforementioned problems has been a longstanding open problem. 
Prior work could either only obtain a quasi-polynomial runtime 
or a polynomial-time randomized approximation scheme for non-deterministic finite automata 
and non-deterministic tree automata.  
	
\end{abstract}

\section{Introduction.}\label{sec:introduction}  

In this paper, we focus on the following computational problems:

\begin{problem}[\#CFG]
	Given a Context-Free Grammar (CFG) $G$ over alphabet $\Sigma$ and an integer $n$ in unary, compute the number of words of length exactly $n$ generated by $G$. Notationally, we are interested in computing $|L_{n}(G)|$, where $L_n(G) =  L(G) \cap \Sigma^n$ and $L(G)$ is the language generated by $G$.
\end{problem}

\begin{problem}[\#DNNF]
	Given a circuit $\varphi$ in Decomposable Negation Normal Form (DNNF) over $n$ variables, count the number of assignments over inputs that $\varphi$ evaluates to $1$. Notationally, we are interested in computing $|\mods(\varphi)|$.
\end{problem}

Since counting and sampling are inter-reducible for self-reducible problems, we are equally interested in the design of uniform samplers for CFG and DNNF circuits.

Context-free grammars and decomposable negation normal form circuits are fundamental objects in computer science. Context-free grammars are central in programming languages,  and form the basis for almost every modern programming language~\cite{AMJ07}. Decomposable Negation Normal Form (DNNF) circuits 
were introduced by ~\cite{Darwiche01},  as a representation for Boolean functions that support many interesting problems in polynomial-time, such as satisfiability, clausal entailment or existential variable quantification. Accordingly, DNNF circuits have found applications in planning~\cite{PBDG05}, diagnosis~\cite{Darwiche01}, databases~\cite{AS11}, design automation~\cite{SGRM18}, and the like.

Both \#CFG and \#DNNF are \#P-complete via simple reductions from counting models of a formula in disjunctive normal form, and consequently, there has been an interest in determining whether there exists a polynomial-time randomization scheme to approximately compute the desired quantity within specified tolerance and confidence. 
The best-known randomized approximation scheme~\cite{GoreJKSM97,KSM95} for \#CFG and \#DNNF has a {\em quasi-polynomial} runtime, i.e., running time $\exp(\mathrm{polylog}(|G|,n))$ for \#CFG and $\exp(\mathrm{polylog}(|\varphi|))$ for \#DNNF. Recent work~\cite{ArenasCJR21,ArenasCJR21a,MeelCM24,MeeldC25} has demonstrated the existence of a fully polynomial-runtime approximation scheme (FPRAS) for subsets of \#CFG and \#DNNF, but the problem of whether there exists a fully polynomial runtime approximation scheme for \#CFG and \#DNNF remained open. 

As highlighted by~\cite{GoreJKSM97}, these problems have proven resistant to numerous approaches:
\begin{quote}
	{\em
		``The main and obvious open question is whether a truly polynomial-time
		algorithm exists for sampling from the n-slice of a context-free language, or for
		estimating its size $\ldots$ A superficially appealing approach is  $\ldots$ using dynamic
		programming.  We have not been able to make this approach work nor have we been successful with other approaches such as proving rapid mixing of a suitably defined Markov chain, or finding unbiased estimators with small standard deviations"
	}
\end{quote}

The primary contribution of our work is to present the first FPRAS for \#CFG and \#DNNF, formalized in the form of the following theorems.

\begin{theorem}\label{thm:cfg}
	There is an algorithm $\calA$ that takes a CFG $G$ and $n$ (in unary) as input and returns an estimate $\mathsf{est}$ such that $\Pr\left[ \mathsf{est} \in (1 \pm \varepsilon)|L_n(G)|\right] \geq 1 - \delta.$ Furthermore, $\calA$ runs in time $\mathrm{poly}(\varepsilon^{-1}, \log \delta^{-1}, |G|,n)$.
\end{theorem}

\begin{theorem}\label{thm:dnnf}
	There is an algorithm $\calA$ that takes a DNNF circuit $\varphi$ as input and returns an estimate $\mathsf{est}$ such that $\Pr\left[ \mathsf{est} \in (1 \pm \varepsilon)|\mods(\varphi)|\right] \geq 1 - \delta.$ Furthermore, $\calA$ runs in time $\mathrm{poly}(\varepsilon^{-1}, \log \delta^{-1}, |\varphi|)$.
\end{theorem}

Since the problem of approximate counting and almost uniform sampling are inter-reducible for CFG and DNNF, the following corollaries immediately follow. 
\begin{corollary}\label{corr:cfg}
There is an almost uniform sampler $\calA$ that takes a CFG $G$ and $n$ in unary as input and parameter $\varepsilon$, and returns $\mathsf{ret}  \in L_n(G) \cup \{\bot\}$ such that $\Pr[\mathsf{ret} = \sigma \mid \mathsf{ret}  \neq  \bot ] \in  \frac{1 \pm \varepsilon}{|L_n(G|)}$,  where $\bot$ is a special symbol allowing the algorithm to {\em fail} with very small probability, which can be fixed to an arbitrary constant (e.g., 0.01).	Furthermore, $\calA$ runs in $\mathrm{poly}(\varepsilon^{-1},  |G|,n)$.
\end{corollary}

\begin{corollary}\label{corr:dnnf}
	There is an almost uniform sampler $\calA$ that takes a DNNF circuit $\varphi$ as input and parameter $\varepsilon$, and returns $\mathsf{ret}  \in \mods(\varphi) \cup \{\bot\}$ such that $\Pr[\mathsf{ret} = \sigma \mid \mathsf{ret}  \neq  \bot ] \in  \frac{1 \pm \varepsilon}{|\mods(\varphi)|}$, where $\bot$ is a special symbol allowing the algorithm to {\em fail} with very small probability, which can be bounded by an arbitrary constant (e.g., 0.01).	Furthermore, $\calA$ runs in $\mathrm{poly}(\varepsilon^{-1},  |\varphi|)$.
\end{corollary}

\subsection{Notations and Definitions.}\label{sec:definitions}

\subsubsection{$\Delta$ Programs.}

We will use notations introduced in~\cite{GoreJKSM97}. Let $\Delta = (\Omega, I, O)$ be an algebra on the underlying set $\Omega$, with a distinguished subset $I \subset \Omega$ of primitive elements (also referred to as inputs/variables) and operators $O$. A $\Delta$ program is a sequence $\calP = (q_i \in \Omega: 0 \leq i \leq |\calP|-1)$ ($|\calP|$ represents the size of $\calP$) such that for all $0 \leq i \leq |\calP|-1$, either $q_i \in I$ or $q_i = q_j \circ q_k$ where $i < j$ and $i < k$ and $\circ \in O$. We refer to each $q_i$ as a node, and when $q_i = q_j \circ q_k$, we call $q_j$ and $q_k$ the children of $q_i$, denoted by the notation $\prev(q_i)$. The \emph{height} of a node is its distance to an input node. Input nodes have height 0; if $q_i = q_j \circ q_k$, then the height of $q_i$ is one plus the maximum of the height of $q_j$ and $q_k$. The depth of $\calP$ is the maximum height of any $q_i$. We say that $\calP$ computes $f \in \Omega$ if $f = q_{k}$ where $k = |\calP|-1$ and $q_{k}$ is also referred to as the root node of $\calP$. 

\subsubsection{Context-Free Grammars.}

A \emph{context-free grammar} (CFG) is a tuple $G = (V, \Sigma, R, S)$ where:
\begin{itemize}
	\item $V$ is a finite set of nonterminal symbols,
	\item $\Sigma$ is a finite set of terminal symbols (the alphabet), with $V \cap \Sigma = \emptyset$,
	\item $R \subseteq V \times (V \cup \Sigma)^*$ is a finite set of production rules,
	\item $S \in V$ is the start symbol.
\end{itemize}
We write production rules $(A, \alpha) \in R$ as $A \to \alpha$. A grammar is in \emph{Chomsky Normal Form}  if every production rule is of the form $A \to BC$ or $A \to a$, where $A, B, C \in V$ and $a \in \Sigma$.

For a grammar $G$ and integer $n \geq 0$, we define $L_n(G)$ as the language of all strings derivable from the start symbol $S$ in exactly $n$ derivation steps. The following proposition, due to~\cite{GoreJKSM97},  shows that we can represent $L_n(G)$ using a $(\cup, \cdot)$ program over the algebra $(2^{\Sigma^*}, \Sigma, \{\cup, \cdot\})$, where $2^{\Sigma^*}$ is the set of all languages over $\Sigma$, and $\cup$ and $\cdot$ denote union and concatenation of languages. In addition, the $(\cup, \cdot)$ program is \emph{homogeneous}, that is, every node represents a language of fixed-length words. We state the proposition and provide a formal proof in Appendix for completneess.

\begin{restatable}{proposition}{cnfunion}
	\label{prop:cnf-to-union-concat}
	For any context-free grammar $G=(V, \Sigma, R, S)$ in Chomsky Normal Form, there exists a $(\cup, \cdot)$ homogeneous program $\mathcal{P}$ over the algebra $(2^{\Sigma^*}, \Sigma, \{\cup, \cdot\})$ such that $\mathcal{P}$ represents $L_n(G)$ for any $n \geq 0$.
\end{restatable}

\subsubsection{DNNF Circuits.}
We now turn our attention to DNNF circuits, which are refinements of circuits in Negation Normal Form (NNF). Given a set of Boolean variables $X = \{x_1, x_2, \ldots, x_n\}$, let $\hat{X}$ be the set of literals (variables and their negations) of $X$, i.e., $\hat{X} = \{x_1, \neg x_1, x_2, \neg x_2, \ldots, x_n, \neg x_n\}$. An NNF circuit can be viewed as a $(\vee,\wedge)$ program defined over the algebra $(2^{\hat{X}}, \hat{X}, (\vee,\wedge))$. Let $\var(q)$ be the set of variables appearing {\em below} $q$ (defined recursively), wherein $\var(x_i) = \var(\neg x_i) = \{x_i\}$. An NNF circuit is decomposable (DNNF) if it holds that for every $\wedge$ node $q_i = q_j \wedge q_k$, we have $\var(q_j) \cap \var(q_k) = \emptyset$. A DNNF circuit is smooth if it holds that for every $\vee$ node $q_i = q_j \vee q_k$, we have $\var(q_j) = \var(q_k)$. Every DNNF circuit can be transformed into a smooth DNNF circuit in polynomial time, and therefore, we will assume DNNF circuits are smooth. Given a DNNF circuit, let $\mods(\varphi)$ represent the set of satisfying assignments of $\varphi$, i.e., assignments to variables for which $\varphi$ evaluates to $1$.

\subsubsection{$(+,\times)$ Programs.}

Our core technical analysis relies on  $(+,\times)$ programs, which are defined over the algebra $\Delta = ([X], X, $ $ \{+,\times\})$ where $X$ is a finite set of variables and $[X]$ is the semiring of polynomials with integer non-negative coefficients. Therefore, a $(+,\times)$ program $\calP$ computes a (multivariate) polynomial with non-negative integer coefficients. We will refer to every node of the program by the corresponding operation: input nodes have no children and are labeled by variables, $\times$ nodes have two children and are labeled by $\times$, and $+$ nodes have an unbounded number of children and are labeled by $+$. The children of a node are ordered, so $\children(q)$ is not a set but a sequence of nodes. We will assume a total ordering $\prec$ on the nodes of the program and  that for every node, the ordering of its children is consistent with $\prec$.

The degree of $q$, denoted as $\deg(q)$, is defined inductively: if $q$ is an input node, then $\deg(q) = 1$; if $q_i = q_j + q_k$, then $\deg(q_i) = \max(\deg(q_j), \deg(q_k))$; and if $q_i = q_j \times q_k$, then $\deg(q_i) = \deg(q_j) + \deg(q_k)$. For a $(+,\times)$ program and $q$ a node in $P$, the \emph{support} of $q$, denoted by $\supp(q)$, is the set of monomials of the multilinear polynomial represented by $q$.
A $(+,\times)$ program is \emph{homogeneous} when every $+$ node $q_i = q_j + q_k$ satisfies $\deg(q_j) = \deg(q_k)$. A \emph{multilinear $(+,\times)$ program} is a $(+,\times)$ program whose nodes all compute a multilinear polynomial. Figure~\ref{fig:running_example} represents a multilinear homogeneous $(+,\times)$ program and will be used as a running example. We shall assume the node ordering $q_0 \prec q_1 \prec q_2 \prec \dots \prec q_{19}$.

\begin{fact}\label{fact:product_are_decomposable}
	Every $\times$ node $q = q_1 \times q_2$ in a multilinear $(+,\times)$ program is such that $\var(q_1) \cap \var(q_2) = \emptyset$ and $\supp(q) = \supp(q_1) \otimes \supp(q_2)$ where $\otimes$ is the cross product.
\end{fact}

While the standard definitions of $\Delta$ programs restrict the {\em fan-in} of $+$ nodes to two, we will work with relaxations of $(+,\times)$ programs wherein the $+$ nodes are allowed to have unbounded children, and the $\times$ nodes have a fan-in of two. Accordingly, we will assume that there is no $+$ node that has another $+$ node has a child, since we allow $+$ nodes to have unbounded children.

Note that the number of edges in $\calP$ is $O(|\calP|^2)$. For $q$ a node in a $(+,\times)$ program $\calP$, we denote by $\var(q)$ the set of variables appearing below $q$, that is, labeling nodes that are descendants of $q$.

For the convenience of our analysis, we want to focus only on the nodes $q$ for which $|\supp(q)| > 16n|\calP|^2$, and therefore, we define the notion of {\em effective height}. We say that the effective height of all nodes $q$ such that $|\supp(q)| \leq 16n|\calP|^2$ is zero. Furthermore, the effective height of every other node is one plus the maximum of the effective height of its children. We denote by $Q^i$ (resp. $Q^{\leq i}$, resp. $Q^{> i}$) the set of nodes of $\calP$ of effective height $i$ (resp. $\leq i$, resp. $> i$). For obvious reasons, the running example of Figure~\ref{fig:running_example} is small and has only nodes of effective height $0$.

\begin{figure}
\centering
\begin{tikzpicture}[yscale=1.9,xscale=0.7]
	\def\sep{1};
	\def\op{40};
	\def\dis{-5pt}
	\def\labelcolor{cyan}
	\node[label={[label distance=\dis]-90:\footnotesize{$\textcolor{\labelcolor}{q_{11}}$}}] (v) at (-1,0) {$x_1$};
	\node[label={[label distance=\dis]-90:\footnotesize{$\textcolor{\labelcolor}{q_{12}}$}}] (y) at (0,0) {$x_2$};
	\node[label={[label distance=\dis]-90:\footnotesize{$\textcolor{\labelcolor}{q_{13}}$}}] (x) at (1,0) {$x_3$};
	\node[label={[label distance=\dis]-90:\footnotesize{$\textcolor{\labelcolor}{q_{14}}$}}] (nx) at (2,0) {$x_4$};
	\node[label={[label distance=\dis]-90:\footnotesize{$\textcolor{\labelcolor}{q_{15}}$}}] (ny) at (3,0) {$x_5$};
	\node[label={[label distance=\dis]-90:\footnotesize{$\textcolor{\labelcolor}{q_{16}}$}}] (nz) at (4,0) {$x_6$};
	\node[label={[label distance=\dis]-90:\footnotesize{$\textcolor{\labelcolor}{q_{17}}$}}] (nw) at (5,0) {$x_7$};
	\node[label={[label distance=\dis]-90:\footnotesize{$\textcolor{\labelcolor}{q_{18}}$}}] (w) at (6,0) {$x_8$};
	\node[label={[label distance=\dis]-90:\footnotesize{$\textcolor{\labelcolor}{q_{19}}$}}] (z) at (7,0) {$x_9$};

	\node[label={[label distance=\dis]180:\footnotesize{$\textcolor{\labelcolor}{q_7}$}}] (o0) at (-0.5,0.5) {$+$};
	\node[label={[label distance=\dis]180:\footnotesize{$\textcolor{\labelcolor}{q_8}$}}] (o1) at (1.5,0.5) {$+$};

	\node[label={[label distance=\dis]180:\footnotesize{$\textcolor{\labelcolor}{q_4}$}}] (a1) at (0,1) {$\times$};
	\node[label={[label distance=\dis]180:\footnotesize{$\textcolor{\labelcolor}{q_5}$}}] (a2) at (2,1) {$\times$};
	\node[label={[label distance=\dis]180:\footnotesize{$\textcolor{\labelcolor}{q_9}$}}] (a4) at (4.5,0.5) {$\times$};
	\node[label={[label distance=\dis]180:\footnotesize{$\textcolor{\labelcolor}{q_{10}}$}}] (a5) at (6.5,0.5) {$\times$};

	\node[label={[label distance=\dis]180:\footnotesize{$\textcolor{\labelcolor}{q_3}$}}] (o2) at (0.5,1.5) {$+$};
	\node[label={[label distance=\dis]180:\footnotesize{$\textcolor{\labelcolor}{q_6}$}}] (o3) at (5.55,1) {$+$};

	\node[label={[label distance=\dis]180:\footnotesize{$\textcolor{\labelcolor}{q_1}$}}] (a6) at (1,2) {$\times$};
	\node[label={[label distance=\dis]180:\footnotesize{$\textcolor{\labelcolor}{q_2}$}}] (a7) at (5.5,2) {$\times$};

	\node[label={[label distance=\dis]180:\footnotesize{$\textcolor{\labelcolor}{q_0}$}}]  (o4) at (3.5,2.5) {$+$};
	
	\node (f) at (-1,2.9) {$ $};

	\draw (x) -- (o1) -- (nx);
	\draw (x) -- (a1) -- (o0);
	\draw (v) -- (o0);
	\draw (a1) -- (o0) -- (y);
	\draw (o1) -- (a2) -- (ny);
	\draw (nz) -- (a4) -- (nw);
	\draw (z) -- (a5) -- (w);
	\draw (a4) -- (o3) -- (a5);
	\draw (a1) -- (o2) -- (a2);
	\draw (a2) -- (a7) -- (a5);
	\draw (o2) -- (a6) -- (o3);
	\draw (a6) -- (o4) -- (a7);
	
\end{tikzpicture}
\caption{A $(+,\times)$ program}\label{fig:running_example}
\end{figure}

\subsection{Results.}

We now justify our focus on $(+,\times)$ programs by pointing out the reduction from DNNF circuits and $(\cup, \cdot)$ programs to multilinear $(+,\times)$ programs. We begin with the following lemma, due to~\cite{GoreJKSM97}, concerning $(+,\times)$ programs and $(\cup, \cdot)$ programs. 

\begin{lemma}[\cite{GoreJKSM97}]\label{lem:cfgtoplustimes}
	Suppose there is a $(\cup, \cdot)$ program $\Pi$ computing language $\calL  \subseteq \Sigma^n$. Then $\Pi$ can be transformed in polynomial-time into a homogeneous multilinear $(+, \times)$ program $\calP$ of size $n|\Pi|$ and degree $n$ such that $\supp(\calP)$ has a one-to-one mapping with words in $\calL$, and therefore, $|\supp(\calP)| = |\calL|$.
\end{lemma}

The transformation by ~\cite{GoreJKSM97} also ensures that the mapping between ${\calL}$ and $\supp(\calP)$ is polynomial-time computable, which enables design of almost uniform sampler (Corollary~\ref{corr:cfg}). 

Now, it is easy to see smooth DNNF circuits can be transformed into $(+,\times)$ programs by having a variable assigned for every literal, i.e., $x_i$ and $\neg x_i$ are mapped to different variables, and by replacing $\land$ nodes by $\times$ nodes, and $\lor$ nodes by $+$ nodes. We formalize this observation in the form of the following proposition.

\begin{lemma}\label{lemma:dnnftoplustimes}
	Every smooth DNNF circuit $\varphi$ can be transformed in linear-time into a homogeneous multilinear $(+,\times)$ program $\calP$ such that there is a one-to-one mapping between $\mods(\varphi)$ and $\supp(\calP)$, and therefore, $|\mods(\varphi)| = |\supp(\calP)|$. 
\end{lemma}

Furthermore, the mapping between $\mods(\varphi)$ and $\supp(\calP)$ is linear-time computable, which enables design of almost uniform sampler (Corollary~\ref{corr:dnnf}).

Our main result requires $(+,\times)$ programs to have depth $d \leq 3 \lceil \log n \rceil$, with $n$ the number of variables. Fortunately, there is a polynomial-time algorithm for the depth reduction of $(+,\times)$ programs which is due to~\cite{ValiantSBR83,GoreJKSM97}\footnote{The
	result in~\cite{ValiantSBR83} was stated for binary $+$ operations, which was
	subsequently improved by~\cite{GoreJKSM97} to handle the case when $+$ nodes have
	an unbounded number of children.}.

\begin{lemma}[\cite{ValiantSBR83,GoreJKSM97}]\label{lemma:depth_reduction}
	There is an algorithm $\calA$ that transforms every $n$-variables homogeneous multilinear $(+,\times)$ program $\calP$ into a equivalent homogeneous multilinear $(+,\times)$ program 
	of size $O(|\calP|^2)$ and depth at most $3\lceil \log(n) \rceil$. Furthermore, $\calA$ has runtime $O(|\calP|^3)$. 
\end{lemma}

In addition, we may assume that the $(+,\times)$ programs after depth reduction are such that no $+$ node has a $+$ child. This is not a problem since no bound is imposed on the number of children of $+$ nodes. We are now ready to state the main technical result.

\begin{restatable}{theorem}{mainResult}\label{thm:main}
	There is an algorithm that takes as input a multilinear and homogeneous $(+,\times)$ program $\calP$ of depth at most $3 \lceil \log n \rceil$, $n$ the number of variables, and returns an estimate $\mathsf{est}$ such that $\Pr\left[ \mathsf{est} \in (1 \pm \varepsilon)|\supp(\calP)|\right] \geq 1 - \delta$, and runs in time $O(n^{15}|\calP|^{14}\varepsilon^{-4}\log(\delta^{-1}))$.
\end{restatable}

Combining Theorem~\ref{thm:main} with Lemma~\ref{lem:cfgtoplustimes} immediately implies Theorem~\ref{thm:cfg}. Similarly, combining Theorem~\ref{thm:main} with Lemma~\ref{lemma:dnnftoplustimes} implies Theorem~\ref{thm:dnnf}.

\subsection{Technical Overview. } 

The high-level idea of our algorithm is conceptually simple. Let $n = \deg(\calP)$. The nodes of $\calP$ are visited in bottom-up order. For every node $q$, we compute or construct several objects:

\begin{itemize}
	\item A rational number $p(q)$ that estimates $\frac{16n}{|\supp(q)|}$. 
	\item Sets of samples $S^1(q), \dots, S^{\nSamples}(q)$, where $\nSamples$ is fixed. The number $\nSamples$ of sample sets per node is a product $\nSamples = n_s \cdot n_t$, with $n_s$ and $n_t$ polynomial in the input size (see Section~\ref{sec:algorithm} for exact values). Each $S^r(q)$ is a subset of $\supp(q)$. Ideally, the $\nSamples$ sets would be stochastically independent, but that is not the case here for technical reasons. The number $\nSamples$ of sample sets per node is predefined, but the size of each sample set is not. One of the main challenges is to keep these sets small to ensure polynomial running time.
\end{itemize}

Nodes of $\calP$ with effective height $0$ have a small $\supp(q)$, and, therefore, for these nodes, the algorithm finds $\supp(q)$, sets $p(q)$ to its target value (i.e., $\frac{16n}{|\supp(q)|}$), and independently constructs the sets $\{ S^r(q) \}_{r=1}^{\nSamples}$
by sampling directly from $\supp(q)$. For nodes of higher effective height, $p(q)$ and $S^r(q)$ are computed in a randomized procedure that uses the values of $p(\cdot)$ and $S^r(\cdot)$ already obtained for the children of $q$. In particular, the samples for $q$ are constructed directly from its children's samples. Once all nodes of $\calP$ have been processed, the algorithm outputs $\frac{16n}{p(o)}$, where $o$ is the output node of $\calP$. Informally, if $p(o)$ is a good estimate of $\frac{16n}{|\supp(o)|}$, then the algorithm outputs a good estimate of $|\supp(o)|$, which equals $|\supp(\calP)|$. 

Unsurprisingly, the primary challenge lies in analyzing this approach.  For analysis, we would ideally want $S^r(q)$ to be such that every element of $\supp(q)$ is in $S^r(q)$ independently with identical probability. However, the desideratum of i.i.d. samples is rather too expensive, which led to quasi-polynomial runtime in prior work~\cite{GoreJKSM97}\footnote{As observed by ~\cite{GoreJKSM97}, their techniques lends itself to FPRAS for constant depth circuits or for the case when $(+,\times)$-program was a formula instead of circuit}.   In recent works on FPRAS for Non-deterministic Finite Automaton (NFA)~\cite{ArenasCJR21} and Non-deterministic Tree Automaton (NTA)~\cite{ArenasCJR21a}, independence could be achieved by exploiting {\em self-reducibility union property}\footnote{The term ``self-reducibility union property" was coined and defined in~\cite{MeelCM24} and is different from self-reduciblity property.}, but these techniques do not seem to extend to general $(+,\times)$ programs. The techniques behind the recent FPRAS for non-deterministic read-once branching programs~\cite{MeeldC25} (and the improved FPRAS for  NFA~\cite{MdC25b}) were the first to go beyond self-reducibility and independent samples. 

A core insight of our approach, as well as the approach of~\cite{MeeldC25,MdC25b}, is to abandon the desire for independence altogether. In particular, we do not  ensure even pairwise independence. That is,  for $\alpha, \alpha' \in \supp(q)$, it may not be the case that $\Pr[\alpha \in S^r(q) \mid \alpha' \in S^r(q)] = \Pr[\alpha \in S^r(q)]$. Of course, we need to quantify the dependence. Before discussing how $p(q)$ and $S^r(q)$ are computed, we introduce the procedure $\reduce$ that takes a set ${\cal Z}$ and a probability $t$ and returns a copy of ${\cal Z}$ in which every element is kept with probability $t$.

\begin{algorithm}[H]
	\begin{algorithmic}[1]
		\STATE ${\cal Z'} \leftarrow \emptyset$
		\FOR{$z \in {\cal Z}$}
		\STATE add $z$ to ${\cal Z'}$ with probability $t$
		\ENDFOR
		\STATE return ${\cal Z'}$
	\end{algorithmic}
	\caption{$\reduce({\cal Z},t)$}
\end{algorithm}

\noindent

We assume $q$ has effective height at least $1$ and that the values $p(\cdot)$ and $S^r(\cdot)$ are known for its children. The computation differs depending on the type of $q$:
\begin{enumerate}
	\item[•] If $q$ is a $\times$ node $q = q_1 \times q_2$ then
	\begin{align*}
		\hat{S}^{r}(q) = S^{r}(q_1) \otimes S^{r}(q_2)		\quad\qquad
		p(q) = \frac{p(q_1) p(q_2)}{16n} 						
		\quad\qquad
		S^{r}(q) = \reduce\left( \hat{S}^{r}(q) , \frac{p(q)}{p(q_1) p(q_2)}\right)
	\end{align*}
	where $\otimes$ denotes  the cross product over the sets, and $r$ ranges in $[\nSamples]$.
	
	\item[•] If $q$ is a $+$ node $q = q_1 + q_2$ (assuming fan-in two for simplicity) with $q_1 \prec q_2$. The technical difficulty for $+$ nodes arises from the fact that it is possible that for a given $\alpha \in \supp(q)$, we have  $\alpha \in \supp(q_1)$ as well as $\alpha \in \supp(q_2)$. Therefore, in order to ensure no $\alpha$ has undue advantage even if it lies in the support of multiple children, we compute $S^r(q)$ as follows: 
	\begin{align*}
		\rho(q) &= \min(p(q_1),p(q_2)) \\
		\hat{S}^r(q) &= \reduce\left(S^r(q_1), \frac{\rho(q)}{p(q_1)}\right) \bigcup \left(\reduce\left(S^r(q_2), \frac{\rho(q)}{p(q_2)}\right) \setminus \supp(q_1)\right)  \\
		p(q) &= 16n\rho(q)  \cdot \left(\underset{0 \leq j < n_t}{\median}\frac{1}{n_s}\sum\limits_{r = j\cdot n_s+1}^{(j+1)n_s}|\hat{S}^r(q)|\right)^{-1} \\
		S^r(q) &= \reduce\left(\hat{S}^r(q), \frac{p(q)}{\rho(q)}\right)
	\end{align*} 
	where $r$ ranges in $[\nSamples]$.
\end{enumerate}

There are a few things to unpack here: First, in both cases, the sets $S^r(q)$ are reduced from intermediate sets $\hat{S}^r(q)$ in such a way that 
\begin{equation}\label{equation:uniform}
	\Pr[\alpha \in S^r(q)] = p(q)
\end{equation}
 holds for every $\alpha$ in $\supp(q)$. Because the sets $S^r(q)$ must remain small for polynomial running time, $p(q)$ is computed to be in the interval $\frac{16n(1 \pm \kappa)}{|\supp(q)|}$ with high probability, where $\kappa$ is a fixed constant (so the expected size of $S^r(q)$ is linear in $n$). We call this interval $\Delta(q)$.

The $\times$ nodes are generally not a problem as the guarantee that $p(q) \in \Delta(q)$ propagates from the children to $q$ by construction. When dealing with a $+$ node $q = q_1 + q_2$, we are concerned about the fact that it is possible for some $\alpha$ to belong to both $\supp(q_1)$ and $\supp(q_2)$. Therefore, the usage of $S^r(q_2) \setminus \supp(q_1)$ ensures that for every $\alpha \in \supp(q)$, there is exactly one child $q'$ of $q$ such that, if $\alpha \in \hat{S}^r(q)$, then $\alpha \in S^r(q')$, and therefore, no $\alpha$ has *undue advantage*. Finally, $p(q)$ is computed via the standard method of  median of means.

\begin{figure}
	\begin{subfigure}{0.5\textwidth}
		\centering
		\begin{tikzpicture}[yscale=1.8,xscale=0.58]
			\def\sep{1};
			\def\op{40};
			\def\treecolor{red};
			
			\draw[\treecolor!\op,line width=0.2cm] (1,0) -- (1.5,0.5) -- (2,1);
			\draw[\treecolor!\op,line width=0.2cm] (2,1) -- (3,0);
			\draw[\treecolor!\op,line width=0.2cm] (7,0) -- (6.5,0.5) -- (6,0);
			\draw[\treecolor!\op,line width=0.2cm] (5.5,1) -- (6.5,0.5) ;
			\draw[\treecolor!\op,line width=0.2cm] (0.5,1.5) -- (2,1);
			\draw[\treecolor!\op,line width=0.2cm] (0.5,1.5) -- (1,2) -- (5.5,1);
			\draw[\treecolor!\op,line width=0.2cm] (1,2) -- (3.25,2.5);

			\node (v) at (-1,0) {$x_1$};
			\node (y) at (0,0) {$x_2$};
			\node[fill=\treecolor!\op,circle,inner sep=\sep] (x) at (1,0) {$x_3$};
			\node (nx) at (2,0) {$x_4$};
			\node[fill=\treecolor!\op,circle,inner sep=\sep] (ny) at (3,0) {$x_5$};
			\node (nz) at (4,0) {$x_6$};
			\node (nw) at (5,0) {$x_7$};
			\node[fill=\treecolor!\op,circle,inner sep=\sep] (w) at (6,0) {$x_8$};
			\node[fill=\treecolor!\op,circle,inner sep=\sep] (z) at (7,0) {$x_9$};

			\node (o0) at (-0.5,0.5) {$+$};
			\node[fill=\treecolor!\op,circle,inner sep=\sep] (o1) at (1.5,0.5) {$+$};

			\node (a1) at (0,1) {$\times$};
			\node[fill=\treecolor!\op,circle,inner sep=\sep] (a2) at (2,1) {$\times$};
			\node (a4) at (4.5,0.5) {$\times$};
			\node[fill=\treecolor!\op,circle,inner sep=\sep] (a5) at (6.5,0.5) {$\times$};

			\node[fill=\treecolor!\op,circle,inner sep=\sep] (o2) at (0.5,1.5) {$+$};
			\node[fill=\treecolor!\op,circle,inner sep=\sep] (o3) at (5.55,1) {$+$};

			\node[fill=\treecolor!\op,circle,inner sep=\sep] (a6) at (1,2) {$\times$};
			\node (a7) at (5.5,2) {$\times$};

			\node[fill=\treecolor!\op,circle,inner sep=\sep] (o4) at (3.25,2.5)  {$+$};

			\draw (x) -- (o1) -- (nx);
			\draw (x) -- (a1) -- (o0);
			\draw (v) -- (o0);
			\draw (a1) -- (o0) -- (y);
			\draw (o1) -- (a2) -- (ny);
			\draw (nz) -- (a4) -- (nw);
			\draw (z) -- (a5) -- (w);
			\draw (a4) -- (o3) -- (a5);
			\draw (a1) -- (o2) -- (a2);
			\draw (a2) -- (a7) -- (a5);
			\draw (o2) -- (a6) -- (o3);
			\draw (a6) -- (o4) -- (a7);
		\end{tikzpicture}
		\caption{}\label{fig:red_derivation_tree}
	\end{subfigure}
	\begin{subfigure}{0.5\textwidth}
		\centering
		\begin{tikzpicture}[yscale=1.8,xscale=0.58]
			\def\sep{1};
			\def\op{30};
			\def\treecolor{blue};
			
			\draw[\treecolor!\op,line width=0.2cm] (1,0) -- (1.5,0.5) -- (2,1);
			\draw[\treecolor!\op,line width=0.2cm] (2,1) -- (3,0);
			\draw[\treecolor!\op,line width=0.2cm] (7,0) -- (6.5,0.5) -- (6,0);
			\draw[\treecolor!\op,line width=0.2cm] (6.5,0.5) -- (5.5,2) -- (2,1);
			\draw[\treecolor!\op,line width=0.2cm] (5.5,2) -- (3.25,2.5);	

			\node (v) at (-1,0) {$x_1$};
			\node (y) at (0,0) {$x_2$};
			\node[fill=\treecolor!\op,circle,inner sep=\sep] (x) at (1,0) {$x_3$};
			\node (nx) at (2,0) {$x_4$};
			\node[fill=\treecolor!\op,circle,inner sep=\sep] (ny) at (3,0) {$x_5$};
			\node (nz) at (4,0) {$x_6$};
			\node (nw) at (5,0) {$x_7$};
			\node[fill=\treecolor!\op,circle,inner sep=\sep] (w) at (6,0) {$x_8$};
			\node[fill=\treecolor!\op,circle,inner sep=\sep] (z) at (7,0) {$x_9$};

			\node (o0) at (-0.5,0.5) {$+$};
			\node[fill=\treecolor!\op,circle,inner sep=\sep] (o1) at (1.5,0.5) {$+$};

			\node (a1) at (0,1) {$\times$};
			\node[fill=\treecolor!\op,circle,inner sep=\sep] (a2) at (2,1) {$\times$};
			\node (a4) at (4.5,0.5) {$\times$};
			\node[fill=\treecolor!\op,circle,inner sep=\sep] (a5) at (6.5,0.5) {$\times$};

			\node (o2) at (0.5,1.5) {$+$};
			\node  (o3) at (5.55,1) {$+$};

			\node (a6) at (1,2) {$\times$};
			\node[fill=\treecolor!\op,circle,inner sep=\sep] (a7) at (5.5,2) {$\times$};

			\node[fill=\treecolor!\op,circle,inner sep=\sep] (o4) at (3.25,2.5)  {$+$};

			\draw (x) -- (o1) -- (nx);
			\draw (x) -- (a1) -- (o0);
			\draw (v) -- (o0);
			\draw (a1) -- (o0) -- (y);
			\draw (o1) -- (a2) -- (ny);
			\draw (nz) -- (a4) -- (nw);
			\draw (z) -- (a5) -- (w);
			\draw (a4) -- (o3) -- (a5);
			\draw (a1) -- (o2) -- (a2);
			\draw (a2) -- (a7) -- (a5);
			\draw (o2) -- (a6) -- (o3);
			\draw (a6) -- (o4) -- (a7);
		\end{tikzpicture}
		\caption{}\label{fig:blue_derivation_tree}
	\end{subfigure}
	\caption{Two ways to derive $x_3x_5x_8x_9$.}\label{fig:red_blue_trees}
\end{figure}

The sample selection from a canonical child leads to a key feature of our algorithm: for every node $q$ and $\alpha \in \supp(q)$, $\alpha$ may be sampled in $S^r(q)$ only if $\alpha$ has been constructed following a unique {\em derivation tree}, denoted by $\tree(\alpha, q)$. Returning to the running example,  Figures~\ref{fig:red_derivation_tree} and~\ref{fig:blue_derivation_tree} illustrate the only two derivations for the monomial $x_3x_5x_8x_9$. Recall that the nodes are ordered as $q_0 \prec q_1 \prec q_2 \prec \dots \prec q_{19}.$
$x_3x_5x_8x_9$ is in the support of $q_0$ (the root), of $q_1$ (the red child of $q_0$), and of $q_2$ (the blue child of $q_0$). Since $q_1 \prec q_2$, this monomial will appear in a sample set $S^r(q_0)$ only if it is in $S^r(q_1)$, so only if it is constructed as shown in Figure~\ref{fig:red_derivation_tree}.\footnote{Technically, all nodes in the depicted example have effective height $0$ and thus are handled in a special manner. For the sake of high-level intuition, we will assume (in this section) that only the leaves have effective height $0$.} The red-highlighted subcircuit in Figure~\ref{fig:red_derivation_tree} represents the derivation tree $\tree(x_3x_5x_8x_9, q_0)$.

\subsubsection{Tackling Dependence.}

As emphasized earlier, the crucial distinction in our work from earlier efforts is the embrace of dependence. For instance, consider $q = q_1 + q_2$ and $\hat{q} = q_1 + q_3$, then $S^r(q)$ and $S^r(\hat{q})$ will, of course, reuse samples $S^r(q_1)$, and therefore, do not even have pairwise independence. Now, of course, we need to bound dependence so as to retain any hope of computing $p(q)$ from $S^r(q)$. To this end, we focus on the following quantity of interest: $\Pr[\alpha \in S^r(q) \text{ and } \alpha' \in S^r(q)]$ for $\alpha, \alpha' \in \supp(q)$, which determines the variance for the estimator.

Our updates ensure that $\Pr[\alpha \in S^r(q) \text{ and } \alpha' \in S^r(q)]$ can be captured by the {\em overlap} $\tau$ (defined formally in Section~\ref{sec:derivation}) between their derivation trees, i.e., $\tree(\alpha,q)$ and $\tree(\alpha',q)$. Informally, the overlap $\tau$ is the set of nodes such that the subtrees rooted at these nodes are in both $\tree(\alpha,q)$ and $\tree(\alpha',q)$. 

Accordingly, we establish the following bound:
\begin{align}\label{eq:overlapeq}
	\Pr[\alpha \in S^r(q) \text{ and } \alpha' \in S^r(q)] \leq \frac{p(q)^2}{\prod_{\hat{q} \in \tau} p(\hat{q})}
\end{align}

\begin{figure}
\begin{subfigure}{0.33\textwidth}
\centering
\begin{tikzpicture}[yscale=1.8,xscale=0.58, every node/.style={font=\small}]
	\def\sep{1};
	\def\op{30};

	\node (v) at (-1,0) {$q_{11}$};
	\node (y) at (0,0) {$q_{12}$};
	\node (x) at (1,0) {$q_{13}$};
	\node (nx) at (2,0) {$q_{14}$};
	\node (ny) at (3,0) {$q_{15}$};
	\node (nz) at (4,0) {$q_{16}$};
	\node (nw) at (5,0) {$q_{17}$};
	\node (w) at (6,0) {$q_{18}$};
	\node (z) at (7,0) {$q_{19}$};

	\node (o0) at (-0.5,0.5) {$q_7$};
	\node (o1) at (1.5,0.5) {$q_8$};

	\node (a1) at (0,1) {$q_4$};
	\node (a2) at (2,1) {$q_5$};
	\node (a4) at (4.5,0.5) {$q_9$};
	\node (a5) at (6.5,0.5) {$q_{10}$};

	\node (o2) at (0.5,1.5) {$q_3$};
	\node (o3) at (5.55,1) {$q_6$};

	\node (a6) at (1,2) {$q_1$};
	\node (a7) at (5.5,2) {$q_2$};

	\node (o4) at (3.25,2.5) {$q_0$};

	\draw (x) -- (o1) -- (nx);
	\draw (x) -- (a1) -- (o0);
	\draw (v) -- (o0);
	\draw (a1) -- (o0) -- (y);
	\draw (o1) -- (a2) -- (ny);
	\draw (nz) -- (a4) -- (nw);
	\draw (z) -- (a5) -- (w);
	\draw (a4) -- (o3) -- (a5);
	\draw (a1) -- (o2) -- (a2);
	\draw (a2) -- (a7) -- (a5);
	\draw (o2) -- (a6) -- (o3);
	\draw (a6) -- (o4) -- (a7);
\end{tikzpicture}
\caption{}\label{fig:nodes_name}
\end{subfigure}\begin{subfigure}{0.33\textwidth}
\centering
\begin{tikzpicture}[yscale=1.8,xscale=0.58]
	\def\sep{1};
	\def\op{40};
	\def\treecolor{red};
	
	\draw[\treecolor!\op,line width=0.2cm] (1,0) -- (1.5,0.5);
	\draw[\treecolor!\op,line width=0.2cm] (1.5,0.5) -- (2,1);
	\draw[\treecolor!\op,line width=0.2cm] (2,1) -- (3,0);
	\draw[\treecolor!\op,line width=0.2cm] (7,0) -- (6.5,0.5);
	\draw[\treecolor!\op,line width=0.2cm] (6.5,0.5) -- (6,0);
	\draw[\treecolor!\op,line width=0.2cm] (5.5,1) -- (6.5,0.5);
	\draw[\treecolor!\op,line width=0.2cm] (0.5,1.5) -- (2,1);
	\draw[\treecolor!\op,line width=0.2cm] (0.5,1.5) -- (1,2);
	\draw[\treecolor!\op,line width=0.2cm] (1,2) -- (5.5,1);
	\draw[\treecolor!\op,line width=0.2cm] (1,2) -- (3.25,2.5);

	\node (v) at (-1,0) {$x_1$};
	\node (y) at (0,0) {$x_2$};
	\node[fill=\treecolor!\op,circle,inner sep=\sep] (x) at (1,0) {$x_3$};
	\node (nx) at (2,0) {$x_4$};
	\node[fill=\treecolor!\op,circle,inner sep=\sep] (ny) at (3,0) {$x_5$};
	\node (nz) at (4,0) {$x_6$};
	\node (nw) at (5,0) {$x_7$};
	\node[fill=\treecolor!\op,circle,inner sep=\sep] (w) at (6,0) {$x_8$};
	\node[fill=\treecolor!\op,circle,inner sep=\sep] (z) at (7,0) {$x_9$};

	\node (o0) at (-0.5,0.5) {$+$};
	\node[fill=\treecolor!\op,circle,inner sep=\sep] (o1) at (1.5,0.5) {$+$};

	\node (a1) at (0,1) {$\times$};
	\node[fill=\treecolor!\op,circle,inner sep=\sep] (a2) at (2,1) {$\times$};
	\node (a4) at (4.5,0.5) {$\times$};
	\node[fill=\treecolor!\op,circle,inner sep=\sep] (a5) at (6.5,0.5) {$\times$};

	\node[fill=\treecolor!\op,circle,inner sep=\sep] (o2) at (0.5,1.5) {$+$};
	\node[fill=\treecolor!\op,circle,inner sep=\sep] (o3) at (5.55,1) {$+$};

	\node[fill=\treecolor!\op,circle,inner sep=\sep] (a6) at (1,2) {$\times$};
	\node (a7) at (5.5,2) {$\times$};

	\node[fill=\treecolor!\op,circle,inner sep=\sep] (o4) at (3.25,2.5)  {$+$};

	\draw (x) -- (o1) -- (nx);
	\draw (x) -- (a1) -- (o0);
	\draw (v) -- (o0);
	\draw (a1) -- (o0) -- (y);
	\draw (o1) -- (a2) -- (ny);
	\draw (nz) -- (a4) -- (nw);
	\draw (z) -- (a5) -- (w);
	\draw (a4) -- (o3) -- (a5);
	\draw (a1) -- (o2) -- (a2);
	\draw (a2) -- (a7) -- (a5);
	\draw (o2) -- (a6) -- (o3);
	\draw (a6) -- (o4) -- (a7);
\end{tikzpicture}
\caption{}\label{fig:real_derivation_tree}
\end{subfigure}\begin{subfigure}{0.33\textwidth}
\centering
\begin{tikzpicture}[yscale=1.8,xscale=0.58]
	\def\sep{1};
	\def\op{30};
	\def\treecolor{blue};
	
	\draw[\treecolor!\op,line width=0.2cm] (0.5,1.5) -- (0,1);
	\draw[\treecolor!\op,line width=0.2cm] (0,1) -- (-1,0);
	\draw[\treecolor!\op,line width=0.2cm] (0,1) -- (1,0);
	\draw[\treecolor!\op,line width=0.2cm] (7,0) -- (6.5,0.5);
	\draw[\treecolor!\op,line width=0.2cm] (6.5,0.5) -- (6,0);
	\draw[\treecolor!\op,line width=0.2cm] (5.5,1) -- (6.5,0.5) ;
	\draw[\treecolor!\op,line width=0.2cm] (0.5,1.5) -- (1,2);
	\draw[\treecolor!\op,line width=0.2cm] (1,2) -- (5.5,1);
	\draw[\treecolor!\op,line width=0.2cm] (1,2) -- (3.25,2.5);

	\node[fill=\treecolor!\op,circle,inner sep=\sep] (v) at (-1,0) {$x_1$};
	\node (y) at (0,0) {$x_2$};
	\node[fill=\treecolor!\op,circle,inner sep=\sep] (x) at (1,0) {$x_3$};
	\node (nx) at (2,0) {$x_4$};
	\node (ny) at (3,0) {$x_5$};
	\node (nz) at (4,0) {$x_6$};
	\node (nw) at (5,0) {$x_7$};
	\node[fill=\treecolor!\op,circle,inner sep=\sep] (w) at (6,0) {$x_8$};
	\node[fill=\treecolor!\op,circle,inner sep=\sep] (z) at (7,0) {$x_9$};

	\node[fill=\treecolor!\op,circle,inner sep=\sep](o0) at (-0.5,0.5) {$+$};
	\node (o1) at (1.5,0.5) {$+$};

	\node[fill=\treecolor!\op,circle,inner sep=\sep] (a1) at (0,1) {$\times$};
	\node (a2) at (2,1) {$\times$};
	\node (a4) at (4.5,0.5) {$\times$};
	\node[fill=\treecolor!\op,circle,inner sep=\sep] (a5) at (6.5,0.5) {$\times$};

	\node[fill=\treecolor!\op,circle,inner sep=\sep] (o2) at (0.5,1.5) {$+$};
	\node[fill=\treecolor!\op,circle,inner sep=\sep] (o3) at (5.55,1) {$+$};

	\node[fill=\treecolor!\op,circle,inner sep=\sep] (a6) at (1,2) {$\times$};
	\node (a7) at (5.5,2) {$\times$};

	\node[fill=\treecolor!\op,circle,inner sep=\sep] (o4) at (3.25,2.5)  {$+$};

	\draw (x) -- (o1) -- (nx);
	\draw (x) -- (a1) -- (o0);
	\draw (v) -- (o0);
	\draw (a1) -- (o0) -- (y);
	\draw (o1) -- (a2) -- (ny);
	\draw (nz) -- (a4) -- (nw);
	\draw (z) -- (a5) -- (w);
	\draw (a4) -- (o3) -- (a5);
	\draw (a1) -- (o2) -- (a2);
	\draw (a2) -- (a7) -- (a5);
	\draw (o2) -- (a6) -- (o3);
	\draw (a6) -- (o4) -- (a7);
\end{tikzpicture}
\caption{}\label{fig:other_derivation_tree}
\end{subfigure}
\caption{The derivation trees of $x_3x_5x_8x_9$ and $x_1x_3x_8x_9$.}\label{fig:derivation_trees}
\end{figure}

We spend some time providing insights behind Equation ~\ref{eq:overlapeq}.  For that we use Figure~\ref{fig:derivation_trees}. First let us explain how to decompose the probability of a monomial being sampled at any node. Consider the monomial $x_3x_5x_8x_9$ in $\supp(q_0)$, with its derivation tree highlighted in red in Figure~\ref{fig:real_derivation_tree}. $x_3x_5x_8x_9$ is in $S^r(q_0)$ if and only if $x_3x_5x_8x_9$ is in $S^r(q_1)$ and is not reduced away when constructing $S^r(q_0)$. We shall say that $x_3x_5x_8x_9$ \emph{survives} the $\reduce$ at $q_1$. Now when does $x_3x_5x_8x_9 \in S^r(q_1)$ occur? This occurs when $x_3x_5$ is in $S^r(q_3)$ and $x_8x_9$ is in $S^r(q_6)$ and when $x_3x_5x_8x_9$ is not reduced away when computing $S^r(q_1) = \reduce(S^r(q_3) \otimes S^r(q_6),t)$ (for some $t$). As we keep going down in the circuit, we conclude that $x_3x_5x_8x_9$ is in $S^r(q)$ if and only if, for every node $q'$ of the derivation tree, the restriction of $x_3x_5x_8x_9$ to $var(q')$ was in $S^r(q')$ and has survived a $\reduce$ procedure at $q'$. 

The probability $\Pr[x_3x_5x_8x_9 \in S^r(q_0)]$ can be decomposed along the derivation tree as follows:
\begin{align*}
&\Pr[x_3x_5x_8x_9 \text{ survives the } \reduce \text{ at } q_0 \mid x_3x_5x_8x_9 \in S^r(q_1)]
\\
\cdot &\Pr[x_3x_5x_8x_9 \text{ survives the } \reduce \text{ at } q_1 \mid x_3x_5 \in S^r(q_3),\, x_8x_9 \in S^r(q_6)]
\\
\cdot &\Pr[x_3x_5 \text{ survives the } \reduce \text{ at } q_3 \mid x_3x_5 \in S^r(q_5)]
\\
\cdot &\Pr[x_3x_5 \text{ survives the } \reduce \text{ at } q_5 \mid x_2 \in S^r(q_8),\, x_5 \in S^r(q_{15})]
\\
\cdot &\Pr[x_3 \text{ survives the } \reduce \text{ at } q_8 \mid x_3 \in S^r(q_{13})] & 
\\
\cdot &\Pr[x_3 \in S^r(q_{13})]\Pr[x_5 \in S^r(q_{15})]&
\\
\tikzmarkin[color=lightgray!20,draw=gray]{boxQ6}(-3.2,-0.25)(-0.2,0.4)\cdot & \Pr[x_8x_9 \text{ survives the } \reduce \text{ at } q_6 \mid x_8x_9 \in S^r(q_{10})]
\\
\cdot &\Pr[x_8x_9 \text{ survives the } \reduce \text{ at } q_{10} \mid x_8 \in S^r(q_{18}),\, x_9 \in S^r(q_{19})] & \textcolor{gray}{\Pr[x_8x_9  \in S^r(q_6)]}
\\
\cdot &\Pr[x_8 \in S^r(q_{18})]
\Pr[x_9 \in S^r(q_{19})] &
\tikzmarkend{boxQ6}
\end{align*}
This chain-rule-like decomposition has one factor per node of the derivation tree. Of course we can stop the decomposition higher up in the tree. For instance we can replace the last four probabilities by $\Pr[x_8x_9  \in S^r(q_6)]$. 

Now we consider the probability of two sampling events occuring simultaneously. Our second event is $x_1x_3x_8x_9 \in S^r(q_0)$. One can decompose $\Pr[x_1x_3x_8x_9 \in S^r(q_0)]$ along its derivation tree -- shown Figure~\ref{fig:other_derivation_tree} -- and see that the resulting decomposition shares two factors with $\Pr[x_3x_5x_8x_9 \in S^r(q_0)]$, namely, $\Pr[x_3 \in S^r(q_{13})]$ and $\Pr[x_8x_9 \in S^r(q_6)]$. When we consider the two events at once we can decompose $\Pr[x_3x_5x_8x_9 \in S^r(q_0) \text{ and } x_1x_3x_8x_9 \in S^r(q_0)]$ along the union of the two derivation trees. The factors of this decomposition are exactly the factors of the two decompositions for $\Pr[x_3x_5x_8x_9 \in S^r(q_0)]$ and $\Pr[x_1x_3x_8x_9 \in S^r(q_0)]$ but no factor is duplicated, so  
$$
\Pr[x_3x_5x_8x_9 \in S^r(q_0) \text{ and } x_1x_3x_8x_9 \in S^r(q_0)] = \frac{\Pr[x_3x_5x_8x_9 \in S^r(q_0)]\Pr[x_1x_3x_8x_9 \in S^r(q_0)]}{\Pr[x_3 \in S^r(q_{13})]\Pr[x_8x_9 \in S^r(q_6)]}.
$$
Compared to Equation~(\ref{eq:overlapeq}), the denominator $\prod_{\hat q \in \tau} p(\hat q)$ corresponds to $\Pr[x_3 \in S^r(q_{13})]\Pr[x_8x_9 \in S^r(q_6)]$. The set of nodes where the two tree merge  is $\tau = \{q_6,q_{13}\}$, and by construction $\Pr[\hat \alpha \in S^r(\hat q)] = p(\hat q)$ for every $\hat{\alpha} \in \supp(\hat q)$. Note that the two trees overlap at other nodes (like $q_1$ or $q_3$), but they merge completely only below $q_6$ and $q_{13}$. In other words the construction of $x_3x_5x_8x_9$ and $x_1x_3x_8x_9$ are shared up to $q_6$ and $q_{13}$ and 
independent above that. Equation (2) captures this idea: by construction $\Pr[\alpha \in S^r(q)]\Pr[\alpha' \in S^r(q)] = p(q)^2$; but $p(q)^2$ is an underestimate of $\Pr[\alpha \in S^r(q) \text{ and } \alpha' \in S^r(q)]$ that counts twice the shared part, i.e., $\prod_{\hat q \in \tau} p(\hat q)$, hence the division by $\prod_{\hat q \in \tau} p(\hat q)$ to correct things.

We now return to the discussion of the variance of the estimator. The variance computation deals with the quantity\footnote{The technical analysis actually deals with the variance for $|\hat{S}^r(q)|$ and requires scaling the expression in (\ref{eq:sum_variance}), we present the ideas using $S^r(q)$ and omit discussion of the scaling factor for the sake of intuition.} 
\begin{equation}\label{eq:sum_variance}
\sum_{\substack{\alpha,\alpha' \in \supp(q) \\ \alpha \neq \alpha'}}\Pr[\alpha \in S^r(q) \text{ and } \alpha' \in S^r(q)]
\end{equation} 
In an ideal independence setting, $\Pr[\alpha \in S^r(q) \text{ and } \alpha' \in S^r(q)]$ would be  $p(q)^2$,  the quantity~(\ref{eq:sum_variance}) would be in $O(p(q)^2 |\supp(q)|^2)$, which is the square of the estimator's mean.
However, there exist some $\alpha$ and $\alpha'$ for which the overlap of the derivation trees $\tree(\alpha,q)$ and $\tree(\alpha',q)$ is significant.  For these pairs $(\alpha,\alpha')$ the probability $\Pr[\alpha \in S^r(q) \text{ and } \alpha' \in S^r(q)]$ may approach $p(q)$, which  is highly undesirable since the quantity $\sum_{\alpha,\alpha'} p(q)$ would be in $O(|\supp(q)|)$ (as  $p(q)$ estimates $16n|\supp(q)|^{-1}$) and thus the variance could be exponential in $n$. The real situation is somewhere in between and our only hope lies in bounding the number of pairs $(\alpha, \alpha')$ with {\em bad} overlap of the derivation trees.

To gain intuition, consider our running example. For $\alpha, \alpha' \in \supp(q_0)$, if $q_3 \in \tau$, then $\alpha$ and $\alpha'$ must agree on $\{x_1, x_2, x_3, x_4, x_5\}$. In other words, $x_i$ in $\alpha$ if and only if $x_i$ in $\alpha'$ for $i = 1, 2, 3, 4, 5$. Thus, $\alpha$ and $\alpha'$ can only differ on the remaining leaves, i.e., $\{x_6, x_7, x_8,x_9\}$ and of course, this bounds the number of such $(\alpha, \alpha')$ pairs. 
More generally, the closer a node $\hat{q} \in \tau$ is to $q_0$, the fewer leaves on which $\alpha$ and $\alpha'$ can differ, bounding the number of such pairs. We show in Lemma~\ref{lemma:antichain_lemma} that the number of pairs $(\alpha, \alpha')$ with overlap $\tau$ is bounded by 
\begin{equation}\label{eq:overlap_bound}
\frac{|\supp(q)|^2}{\prod_{\hat{q} \in \tau} |\supp(\hat{q})|}
\end{equation}

To understand the bound~(\ref{eq:overlap_bound}), we explain that for a fixed $\alpha$, the number of monomials $\alpha'$ for which $\tree(\alpha,q)$ and $\tree(\alpha',q)$ overlap exactly beyond $\tau$ is at most $\frac{|\supp(q)|}{\prod_{\hat{q} \in \tau} |\supp(\hat{q})|}$. For that we reuse the running example with $q = q_0$ and $\alpha = x_3x_5x_8x_9$ and $\tau = \{q_6,q_{13}\}$. The possible $\alpha'$ are $x_1x_3x_8x_9$ and $x_2x_3x_8x_9$. Their derivation trees are shown in Figure~\ref{fig:mutations}. By definition, all possible $\alpha'$ are consistent on the variables appearing below $\tau$. For a fixed $\alpha'$, there are $\prod_{\hat{q} \in \tau} |\supp(\hat{q})|$ different monomials that can be obtained by changing $\tree(\alpha',q)$ below $\tau$. Let us call these  \emph{mutations of $\alpha'$ below $\tau$}. For example the only mutation of $x_1x_3x_8x_9$ below $\{q_6,q_{13}\}$ is $x_1x_3x_6x_7$ and the only mutation of $x_2x_3x_8x_9$ below $\{q_6,q_{13}\}$ is $x_2x_3x_6x_7$. In Figure~\ref{fig:mutations}, the mutations are highlighted in gray. All $\alpha'$ monomials and their $\prod_{\hat{q} \in \tau} |\supp(\hat{q})|$ mutations are pairwise distinct and are in the support of $q$, so there are most $|\supp(q)|$ of them, hence the bound $\frac{|\supp(q)|}{\prod_{\hat{q} \in \tau} |\supp(\hat{q})|}$ on the number of possible $\alpha'$.

Returning to bound~(\ref{eq:overlap_bound}), observe that when $p(\hat{q}) \approx p(q)$, 
we have $|\supp(\hat{q})| \approx |\supp(q)|$, implying that there are not too many 
pairs $(\alpha, \alpha')$ with {\em bad} overlap. Therefore, we can upper bound the 
quantity of interest 	
$\sum_{\substack{\alpha,\alpha' \in \supp(q) \\ \alpha \neq \alpha'}} \Pr[\alpha \in S^r(q) \text{ and } \alpha' \in S^r(q)]$, 
which is dominating factor in the variance (Lemma~\ref{lem:variance-bound}). 
\begin{align*}
	&\sum_{\substack{\alpha,\alpha' \in \supp(q) \\ \alpha \neq \alpha'}} \Pr[\alpha \in S^r(q) \text{ and } \alpha' \in S^r(q)] \\
	&\leq \sum_{\tau} \frac{p(q)^2}{\prod_{\hat{q} \in \tau} p(\hat{q})} \cdot \frac{|\supp(q)|^2}{\prod_{\hat{q} \in \tau} |\supp(\hat{q})|} \\ 
	&= \sum_{\ell = 0}^{n} \left(\sum_{\{\tau \mid |\tau| = \ell\}} \frac{p(q)^2}{\prod_{\hat{q} \in \tau} p(\hat{q})} \cdot \frac{|\supp(q)|^2}{\prod_{\hat{q} \in \tau} |\supp(\hat{q})|} \right)\\
	&\leq  \sum_{\ell = 0}^{n} \left(\sum_{\{\tau \mid |\tau| = \ell\}} \frac{p(q)^2 |\supp(q)|^2}{(8n)^{\ell}}\right)
\end{align*}
The last inequality follows from quantifying the fact that  $p(\hat{q})$ is an 
approximation of $\frac{16n}{|\supp(\hat{q})|}$. Now observe that there are at most 
${4n \choose \ell}$ antichains $\tau$ of size $\ell$, so we have 
\begin{align*}
	&\sum_{\ell = 0}^{n} \left(\sum_{\{\tau \mid |\tau| = \ell\}} \frac{p(q)^2 |\supp(q)|^2}{(8n)^{\ell}}\right) \\
	&\leq \sum_{\ell = 0}^{n} {4n \choose \ell} \frac{p(q)^2 |\supp(q)|^2}{(8n)^{\ell}}\\
	&= p(q)^2 |\supp(q)|^2 \sum_{\ell = 0}^{n} \frac{1}{\ell!} \frac{(4n)!}{(4n-\ell)!} \frac{1}{(8n)^{\ell}}\\
	&\leq p(q)^2 |\supp(q)|^2 \sum_{\ell = 0}^{n} \frac{1}{\ell!} \frac{(4n)^{\ell}}{(8n)^{\ell}}\\
	&= p(q)^2 |\supp(q)|^2 \sum_{\ell = 0}^{n} \frac{1}{\ell!} \frac{1}{2^{\ell}}\\
	&\leq p(q)^2 |\supp(q)|^2 e^{1/2} \\
	&\leq 2 p(q)^2 |\supp(q)|^2
\end{align*}
Since $\mu = \Ex[|S^r(q)|] = p(q)|\supp(q)|$, we have 
\begin{align*}
	\sum_{\substack{\alpha,\alpha' \in \supp(q) \\ \alpha \neq \alpha'}} \Pr[\alpha \in S^r(q) \text{ and } \alpha' \in S^r(q)] &\leq 2 p(q)^2 |\supp(q)|^2 = 2\mu^2
\end{align*}
This allows us to upper bound the variance by a constant factor of the square of its 
mean, enabling the use of the median-of-means estimator.

\begin{figure}
\begin{subfigure}{0.33\textwidth}
\centering
\begin{tikzpicture}[yscale=1.8,xscale=0.58, every node/.style={font=\small}]
	\def\sep{1};
	\def\op{30};

	\node (v) at (-1,0) {$q_{11}$};
	\node (y) at (0,0) {$q_{12}$};
	\node (x) at (1,0) {$q_{13}$};
	\node (nx) at (2,0) {$q_{14}$};
	\node (ny) at (3,0) {$q_{15}$};
	\node (nz) at (4,0) {$q_{16}$};
	\node (nw) at (5,0) {$q_{17}$};
	\node (w) at (6,0) {$q_{18}$};
	\node (z) at (7,0) {$q_{19}$};

	\node (o0) at (-0.5,0.5) {$q_7$};
	\node (o1) at (1.5,0.5) {$q_8$};

	\node (a1) at (0,1) {$q_4$};
	\node (a2) at (2,1) {$q_5$};
	\node (a4) at (4.5,0.5) {$q_9$};
	\node (a5) at (6.5,0.5) {$q_{10}$};

	\node (o2) at (0.5,1.5) {$q_3$};
	\node (o3) at (5.55,1) {$q_6$};

	\node (a6) at (1,2) {$q_1$};
	\node (a7) at (5.5,2) {$q_2$};

	\node (o4) at (3.25,2.5) {$q_0$};

	\draw (x) -- (o1) -- (nx);
	\draw (x) -- (a1) -- (o0);
	\draw (v) -- (o0);
	\draw (a1) -- (o0) -- (y);
	\draw (o1) -- (a2) -- (ny);
	\draw (nz) -- (a4) -- (nw);
	\draw (z) -- (a5) -- (w);
	\draw (a4) -- (o3) -- (a5);
	\draw (a1) -- (o2) -- (a2);
	\draw (a2) -- (a7) -- (a5);
	\draw (o2) -- (a6) -- (o3);
	\draw (a6) -- (o4) -- (a7);
\end{tikzpicture}
\caption{}\label{fig:nodes_name}
\end{subfigure}\begin{subfigure}{0.33\textwidth}
\centering
\begin{tikzpicture}[yscale=1.8,xscale=0.58]
	\def\sep{1};
	\def\op{30};
	\def\treecolor{blue};
	
	\draw[\treecolor!\op,line width=0.2cm] (0.5,1.5) -- (0,1) -- (-1,0);
	\draw[\treecolor!\op,line width=0.2cm] (0,1) -- (1,0);
	\draw[\treecolor!\op,line width=0.2cm] (7,0) -- (6.5,0.5);
	\draw[\treecolor!\op,line width=0.2cm] (6.5,0.5) -- (6,0);
	\draw[\treecolor!\op,line width=0.2cm] (5.5,1) -- (6.5,0.5) ;
	\draw[\treecolor!\op,line width=0.2cm] (0.5,1.5) -- (1,2);
	\draw[\treecolor!\op,line width=0.2cm] (1,2) -- (5.5,1);
	\draw[\treecolor!\op,line width=0.2cm] (1,2) -- (3.25,2.5);
	
	\draw[\treecolor!\op,line width=0.2cm] (1,2) -- (3.25,2.5);
	\draw[gray!\op,line width=0.2cm] (4,0) -- (4.5,0.5);
	\draw[gray!\op,line width=0.2cm] (4.5,0.5) -- (5.5,1);
	\draw[gray!\op,line width=0.2cm] (5,0) -- (4.5,0.5);

	\node[fill=\treecolor!\op,circle,inner sep=\sep] (v) at (-1,0) {$x_1$};
	\node (y) at (0,0) {$x_2$};
	\node[fill=\treecolor!\op,circle,inner sep=\sep] (x) at (1,0) {$x_3$};
	\node (nx) at (2,0) {$x_4$};
	\node (ny) at (3,0) {$x_5$};
	\node[fill=gray!\op,circle,inner sep=\sep] (nz) at (4,0) {$x_6$};
	\node[fill=gray!\op,circle,inner sep=\sep] (nw) at (5,0) {$x_7$};
	\node[fill=\treecolor!\op,circle,inner sep=\sep] (w) at (6,0) {$x_8$};
	\node[fill=\treecolor!\op,circle,inner sep=\sep] (z) at (7,0) {$x_9$};

	\node[fill=\treecolor!\op,circle,inner sep=\sep](o0) at (-0.5,0.5) {$+$};
	\node (o1) at (1.5,0.5) {$+$};

	\node[fill=\treecolor!\op,circle,inner sep=\sep] (a1) at (0,1) {$\times$};
	\node (a2) at (2,1) {$\times$};
	\node[fill=gray!\op,circle,inner sep=\sep] (a4) at (4.5,0.5) {$\times$};
	\node[fill=\treecolor!\op,circle,inner sep=\sep] (a5) at (6.5,0.5) {$\times$};

	\node[fill=\treecolor!\op,circle,inner sep=\sep] (o2) at (0.5,1.5) {$+$};
	\node[fill=\treecolor!\op,circle,inner sep=\sep] (o3) at (5.55,1) {$+$};

	\node[fill=\treecolor!\op,circle,inner sep=\sep] (a6) at (1,2) {$\times$};
	\node (a7) at (5.5,2) {$\times$};

	\node[fill=\treecolor!\op,circle,inner sep=\sep] (o4) at (3.25,2.5)  {$+$};

	\draw (x) -- (o1) -- (nx);
	\draw (x) -- (a1) -- (o0);
	\draw (v) -- (o0);
	\draw (a1) -- (o0) -- (y);
	\draw (o1) -- (a2) -- (ny);
	\draw (nz) -- (a4) -- (nw);
	\draw (z) -- (a5) -- (w);
	\draw (a4) -- (o3) -- (a5);
	\draw (a1) -- (o2) -- (a2);
	\draw (a2) -- (a7) -- (a5);
	\draw (o2) -- (a6) -- (o3);
	\draw (a6) -- (o4) -- (a7);
\end{tikzpicture}
\caption{}\label{fig:mutations_one}
\end{subfigure}\begin{subfigure}{0.33\textwidth}
\centering
\begin{tikzpicture}[yscale=1.8,xscale=0.58]
	\def\sep{1};
	\def\op{30};
	\def\treecolor{cyan};
	
	\draw[\treecolor!\op,line width=0.2cm] (0.5,1.5) -- (0,1);
	\draw[\treecolor!\op,line width=0.2cm] (0,1) -- (-0.5,0.5) ;
	\draw[\treecolor!\op,line width=0.2cm] (-0.5,0.5) -- (0,0);
	\draw[\treecolor!\op,line width=0.2cm] (0,1) -- (1,0);
	\draw[\treecolor!\op,line width=0.2cm] (7,0) -- (6.5,0.5);
	\draw[\treecolor!\op,line width=0.2cm] (6.5,0.5) -- (6,0);
	\draw[\treecolor!\op,line width=0.2cm] (5.5,1) -- (6.5,0.5) ;
	\draw[\treecolor!\op,line width=0.2cm] (0.5,1.5) -- (1,2);
	\draw[\treecolor!\op,line width=0.2cm] (1,2) -- (5.5,1);
	\draw[\treecolor!\op,line width=0.2cm] (1,2) -- (3.25,2.5);
	
	\draw[\treecolor!\op,line width=0.2cm] (1,2) -- (3.25,2.5);
	\draw[gray!\op,line width=0.2cm] (4,0) -- (4.5,0.5);
	\draw[gray!\op,line width=0.2cm] (4.5,0.5) -- (5.5,1);
	\draw[gray!\op,line width=0.2cm] (5,0) -- (4.5,0.5);

	\node (v) at (-1,0) {$x_1$};
	\node[fill=\treecolor!\op,circle,inner sep=\sep] (y) at (0,0) {$x_2$};
	\node[fill=\treecolor!\op,circle,inner sep=\sep] (x) at (1,0) {$x_3$};
	\node (nx) at (2,0) {$x_4$};
	\node (ny) at (3,0) {$x_5$};
	\node[fill=gray!\op,circle,inner sep=\sep] (nz) at (4,0) {$x_6$};
	\node[fill=gray!\op,circle,inner sep=\sep] (nw) at (5,0) {$x_7$};
	\node[fill=\treecolor!\op,circle,inner sep=\sep] (w) at (6,0) {$x_8$};
	\node[fill=\treecolor!\op,circle,inner sep=\sep] (z) at (7,0) {$x_9$};

	\node[fill=\treecolor!\op,circle,inner sep=\sep](o0) at (-0.5,0.5) {$+$};
	\node (o1) at (1.5,0.5) {$+$};

	\node[fill=\treecolor!\op,circle,inner sep=\sep] (a1) at (0,1) {$\times$};
	\node (a2) at (2,1) {$\times$};
	\node[fill=gray!\op,circle,inner sep=\sep] (a4) at (4.5,0.5) {$\times$};
	\node[fill=\treecolor!\op,circle,inner sep=\sep] (a5) at (6.5,0.5) {$\times$};

	\node[fill=\treecolor!\op,circle,inner sep=\sep] (o2) at (0.5,1.5) {$+$};
	\node[fill=\treecolor!\op,circle,inner sep=\sep] (o3) at (5.55,1) {$+$};

	\node[fill=\treecolor!\op,circle,inner sep=\sep] (a6) at (1,2) {$\times$};
	\node (a7) at (5.5,2) {$\times$};

	\node[fill=\treecolor!\op,circle,inner sep=\sep] (o4) at (3.25,2.5)  {$+$};

	\draw (x) -- (o1) -- (nx);
	\draw (x) -- (a1) -- (o0);
	\draw (v) -- (o0);
	\draw (a1) -- (o0) -- (y);
	\draw (o1) -- (a2) -- (ny);
	\draw (nz) -- (a4) -- (nw);
	\draw (z) -- (a5) -- (w);
	\draw (a4) -- (o3) -- (a5);
	\draw (a1) -- (o2) -- (a2);
	\draw (a2) -- (a7) -- (a5);
	\draw (o2) -- (a6) -- (o3);
	\draw (a6) -- (o4) -- (a7);
\end{tikzpicture}
\caption{}\label{fig:mutation_two}
\end{subfigure}
\caption{The derivation trees of $x_1x_3x_8x_9$ and $x_2x_3x_8x_9$ and their mutations below $q_6$.}\label{fig:mutations}
\end{figure}

\subsubsection{Further Technical Difficulties.}

A careful reader will immediately notice the situation is more delicate than what has been portrayed above. In particular, two concerns warrant immediate attention: 

\begin{enumerate}
	\item The first concern is about expressions of the form $\Pr[\alpha \in S^r(q)$ and $ \alpha' \in S^r(q)] \leq \frac{(p(q))^2}{\prod_{\hat{q} \in \tau} p(\hat{q})}$ These state an inequality between a probability, so a fixed number, and $p(q)$ which itself is a random variable. %
	
	\item 
	The second concern is about the independence of the sample sets $S^r(\cdot)$ and $S^l(\cdot)$ for $r \neq l$ as we seek to rely on median of means estimation to establish  $p(q)$ is a good approximation of $16n|\supp(q)|^{-1}$ in the context of $+$ nodes.  Although samples in $S^r(\cdot)$ can only construct samples in other  $S^r(\cdot)$ sets, 
	the sampling probabilities of $S^r(q)$ and $S^l(q)$ for the same $q$ are identical, i.e., $p(q)$. Consequently, the random variables corresponding to the  two sets are not independent.

\end{enumerate}

While the high-level intuition remains the same, the technical analysis circumvents these issues through a carefully defined random process, as outlined in Section~\ref{section:randomProcess}. This approach enables us to work with variables that are {\em better behaved} than $ S^r(\cdot) $ and $ p(\cdot)$ for analyzing the algorithm. An interesting direction for future work would be to simplify the analysis, potentially avoiding reliance on the random process.

\paragraph{Organization.} The rest of the paper is organized as follows: We present, in Section~\ref{sec:derivation}, the notion of derivation trees, which is crucial for our analysis. We present the algorithm in Section~\ref{sec:algorithm} and  analyze the correctness and runtime complexity of the algorithm in Section~\ref{sec:analysis}.

\section{Derivation Trees}\label{sec:derivation}

Let $q$ be a node in a $(+,\times)$ program and let $\alpha \in \supp(q)$. Recall that the nodes of the program are totally ordered by $\prec$. We associate the pair $(\alpha,q)$ with a unique derivation tree rooted at q and denoted by $\tree(\alpha,q)$. The tree is defined recursively as follows: 
\begin{itemize}
	\item If $q$ is an input node, then $\tree(\alpha,q) = q$.
	\item If $q$ is a $+$ node, let $q'$ be the first child of $q$ with respect to $\prec$ such that $\alpha \in \supp(q')$. Then, $\tree(\alpha,q) = (q, \{\tree(\alpha,q')\})$
	\item If $q$ is a $\times$ node $q_1 \times q_2$ then we have  $\alpha = \alpha_1\alpha_2 $ such that $\alpha_1 \in \supp(q_1)$ and $\alpha_2 \in \supp(q_2)$. Then, $\tree(\alpha,q) = (q,\{\tree(\alpha_1,q_1), \tree(\alpha_2,q_2)\})$
\end{itemize}

\noindent Each derivation tree $\tree(\alpha,q)$ in a degree $n$ multilinear program has $|\alpha|$ leaves, so at most $n$ leaves. We claim that $\tree(\alpha,q)$ has at most $4n$ nodes in total. $\times$ nodes have fan-in of $2$ in $\tree(\alpha,q)$ and $+$ nodes have fan-in of $1$ in $\tree(\alpha,q)$ so, replacing $+$ nodes by edges in $\tree(\alpha,q)$ yields a binary tree $T$ with at most $n$ leaves and whose internal nodes are the $\times$ nodes. So at most $n$ nodes of the derivation tree are $\times$ nodes. By assumption on $\calP$, there can only be one $+$ node between two $\times$ nodes or between a $\times$ node and a leaf. So the number of $+$ nodes in $\tree(\alpha,q)$ is at most the number of edges in $T$, so at most $2n$. Thus, $\tree(\alpha,q)$ has at most $4n$ nodes in total.

\begin{example}
In the $(+,\times)$ programs represented in Figure~\ref{fig:red_blue_trees}, the monomial $\alpha = x_3x_5x_8x_9$ is in the support of the root $+$ node $q_0$ and can be derived in two ways highlighted in red and blue. If the children of $q$ are ordered as $q_1 \prec q_2$ then the derivation tree $\tree(\alpha,q)$ is the one shown in Figure~\ref{fig:red_derivation_tree} going through $q_1$.
\end{example}

Given two monomials $\alpha \in \supp(q)$ and $\alpha' \in \supp(q')$, for any nodes $q$ and $q'$, we call the \emph{last common subtree nodeset} of $\tree(\alpha,q)$ and $\tree(\alpha',q')$, denoted by $\lcsn{\tree(\alpha,q)}{\tree(\alpha',q')}$, the set of highest nodes such that $\hat{q} \in \lcsn{\tree(\alpha,q)}{\tree(\alpha',q')}$ if the following conditions hold: 
\begin{itemize}
	\item $\hat{q}$ is a node in $\tree(\alpha,q)$ and the subtrees rooted at $\hat q$ in $\tree(\alpha,q)$ and $\tree(\alpha',q')$ are equal.
	\item There is no other tree containing node $\hat{q}$ that is a subtree of both $\tree(\alpha,q)$ and $\tree(\alpha',q')$. 
\end{itemize}

\begin{example}
Figure~\ref{fig:real_derivation_tree} shows the derivation tree for $\tree(x_3x_5x_8x_9,q_0)$ in red. Figure~\ref{fig:other_derivation_tree} shows the derivation tree $\tree(x_1x_3x_8x_9,q_0)$ in blue. The two trees share eight nodes. The highest nodes up to which the two trees are identical are $q_6$ and $q_{13}$. So $\lcsn{\tree(x_3x_5x_8x_9,q_0)}{\tree(x_1x_3x_8x_9,q_0)} = \{q_6,q_{13}\}$.
\end{example}

It is worth observing that $\lcsn{\tree(\alpha,q)}{\tree(\alpha',q')}$ forms an antichain, i.e., it does not contain two nodes such that one is an ancestor of the other. Let $\tau$ be an antichain in the $(+,\times)$ program. We define 
\[
I(\alpha,q,\tau) := \{ \alpha' \in \supp(q) \mid \lcsn{\tree(\alpha,q)}{\tree(\alpha',q)} = \tau\}	.
\]

\begin{lemma}\label{lemma:first_antichain_lemma}
For every $\alpha \in \supp(q)$ and antichain $\tau$ over nodes of $\tree(\alpha,q)$, $| I(\alpha,q,\tau)| \leq \frac{|\supp(q)|}{\prod_{\hat{q} \in \tau} |\supp(\hat{q})|}$, where $\prod_{\hat{q} \in \tau} |\supp(\hat{q})| = 1$ when $\tau = \emptyset$.
\end{lemma}
\begin{proof}
The lemma is clear when $\tau = \emptyset$. Let $\var(\tau) = \bigcup_{\hat q \in \tau}  \var(\hat q)$. In a derivation trees, if two nodes $\hat q$ and $\hat q'$ are such that none is not the ancestor of the other, then $\var(\hat q) \cap \var(\hat q') = \emptyset$. So, since $\tau$ is an antichain over a derivation trees, it holds that $\var(\hat q) \cap \var(\hat q') = \emptyset$ for every $\hat q' \neq \hat q$ in~$\tau$. 

Let $I(\alpha,q,\tau) = \{\alpha_1,\alpha_2,\dots\}$. For every node $\hat q \in \tau$ let $\alpha_{\hat q}$ be the restriction of $\alpha$ to $\var(\hat q)$, let $\alpha_{i,\hat q}$ be the restriction of $\alpha_i$ to $\var(\hat q)$. Let also $\alpha'_i$ be the restriction of $\alpha_i$ to $\var(q) \setminus \var(\tau)$. Note that $\alpha_i = \alpha'_i\cdot \prod_{\hat q \in \tau} \alpha_{i,\hat q}$. From the definition of $I(\alpha,q,\tau)$, we have that, for every $\hat q \in \tau$ and every $i$, $\alpha_{\hat q, i} = \alpha_{\hat q}$. Thus, for every $i \neq j$, we have $\alpha'_i \neq \alpha'_j$. 

Now, for any $\hat q \in \tau$ and any $\beta_{\hat q} \in \supp(\hat q)$, if we replace $\alpha_{\hat q}$ by $\beta_{\hat q}$ in $\alpha_i$, then we still have a monomial in $\supp(q)$. This can be done for all $\hat q$, so each $\alpha_i$ is used to create monomials of the form $\alpha'_i \cdot \prod_{\hat q \in \tau} \beta_{\hat q}$ ($\alpha_i$ itself is of that form). If two  monomials  $\alpha'_i \cdot \prod_{\hat q \in \tau} \beta_{\hat q}$ and  $\alpha'_i \cdot \prod_{\hat q \in \tau} \beta'_{\hat q}$ are such that $\beta_{\hat q} \neq \beta'_{\hat q}$ for some $\hat q \in \tau$ then the two monomials are distinct since their restrictions to $\var(\hat q)$ are precisely $\beta_{\hat q}$ and $\beta'_{\hat q}$, respectively. So we create $\prod_{\hat q \in \tau} |\supp(\hat q)|$ different monomials from $\alpha_i$, all in $\supp(q)$. 

For $i \neq j$, the monomials $\alpha'_i\cdot \prod_{\hat{q} \in \tau} \beta_{\hat q}$ and $\alpha'_j\cdot \prod_{\hat{q} \in \tau} \beta'_{\hat q}$ are distinct since their restrictions to $\var(q) \setminus \var(\tau)$ are $\alpha'_i$ and $\alpha'_j$, respectively. So in total can create $|I(\alpha,q,\tau)|\cdot\prod_{\hat{q} \in \tau} |\supp(\hat{q})|$ distinct monomials in $\supp(q)$. Hence $|I(\alpha,q,\tau)| \leq |\supp(q)|/\prod_{\hat{q} \in \tau} |\supp(\hat{q})|$.
\end{proof}

We will need a stronger version of Lemma~\ref{lemma:first_antichain_lemma}. For $V \subseteq \calP$, we write 
$$
\lcsnvar{\tree(\alpha,q)}{\tree(\alpha',q')}{V} = V \cap (\lcsn{\tree(\alpha,q)}{\tree(\alpha',q')}).
$$
One can show that Lemma~\ref{lemma:first_antichain_lemma} holds for $I_V(\alpha,q,\tau) = \{ \alpha' \in \supp(q) \mid \lcsnvar{\tree(\alpha,q)}{\tree(\alpha',q)}{V} = \tau\}$. More generally, we can make $V$ depends on $\alpha$, $\alpha'$ and $q$ and still retains the bound. Formally, let $f$ be a function that maps every triple $(\alpha,\alpha',q)$ to a subset of $\calP$, and let 
\[
I_f(\alpha,q,\tau) := \{ \alpha' \in \supp(q) \mid \lcsnvar{\tree(\alpha,q)}{\tree(\alpha',q)}{f(\alpha,\alpha',q)} = \tau\}	.
\]
\begin{lemma}\label{lemma:antichain_lemma}
For every $\alpha \in \supp(q)$ and antichain $\tau$ over nodes of $\tree(\alpha,q)$, $|I_f(\alpha,q,\tau)| \leq \frac{|\supp(q)|}{\prod_{\hat{q} \in \tau} |\supp(\hat{q})|}$, where $\prod_{\hat{q} \in \tau} |\supp(\hat{q})| = 1$ when $\tau = \emptyset$.
\end{lemma}
\noindent The proof is identical to the proof of Lemma~\ref{lemma:first_antichain_lemma}.

The reason behind Lemma~\ref{lemma:antichain_lemma} is that we will need to use restrictions of the derivation trees on various subsets of nodes. In particular, we will use a variant of $\tree(\alpha,q)$ denoted by $\tree^*(\alpha,q)$. It is constructed like $\tree(\alpha,q)$ except that, if $q$ is in $Q^0$ then $\tree^*(\alpha,q) = q$. Formally:
\begin{itemize}
	\item If $q$ is an node of effective height $0$, then $\tree^*(\alpha,q) = q$.
	\item If $q$ is a $+$ node of effective height at least $1$, let $q'$ be the first child of $q$ with respect to $\prec$ such that $\alpha \in \supp(q')$. Then, $\tree^*(\alpha,q) = (q, \{\tree^*(\alpha,q')\})$.
	\item If $q$ is a $\times$ node $q_1 \times q_2$ of effective height at least $1$ then we have  $\alpha = \alpha_1\alpha_2 $ such that $\alpha_1 \in \supp(q_1)$ and $\alpha_2 \in \supp(q_2)$. Then, $\tree^*(\alpha,q) = (q,\{\tree^*(\alpha_1,q_1), \tree^*(\alpha_2,q_2)\})$.
\end{itemize}
	So $\tree^*(\alpha,q)$ is exactly like $\tree(\alpha,q)$ except that each branch stops as soon as a node of effective height~$0$ is reached. The algorithm presented in the next section respects $\tree^*(\alpha,q)$ in the sense that a monomial $\alpha$ from the support of $q$ will be sampled only if it is constructed following $\tree^*(\alpha,q)$. 
	
	The last common subtree nodeset remains a key concept for the restricted trees $\tree^*(\alpha,q)$. Let $V^*$ be $nodes(\tree^*(\alpha,q))\,\cup\,nodes(\tree^*(\alpha',q'))$, then we define $\tree^*(\alpha,q) \land \tree^*(\alpha',q')$ as 
	$$
	\tree^*(\alpha,q) \land \tree^*(\alpha',q') = \lcsnvar{\tree(\alpha,q)}{\tree(\alpha',q')}{V^*}.
	$$ 
So this is just the last common subtree nodeset of $\tree(\alpha,q)$ and $\tree(\alpha',q')$ but restricted to the nodes kept in $\tree^*(\alpha,q)$ and $\tree^*(\alpha',q')$.

\section{Algorithm.}\label{sec:algorithm}

Our goal is to compute a good approximation of $|\supp(\calP)|$. Let $n$ be the degree of $\calP$. Our algorithm constructs in a bottom-up manner $p(q)$, which seeks to estimate $16n/|\supp(q)|$, and samples from $\supp(q)$ for every node $q$ of $\calP$. In fact, for every $q$ we sample several sample sets  $S^1(q),\dots,S^{n_s n_t}(q)$, all subsets of $\supp(q)$, where the numbers $n_s$ and $n_t$ are given in algorithm $\approxMCDNNF$. $S^r(q)$ and $p(q)$ are computed using $\reduce$, thus they are random variables. We estimate $16n/|\supp(q)|$ rather than $1/|\supp(q)|$ for technical reasons. The sets $S^r(q)$ are sampled in such a way that, ideally, $\Pr[\alpha \in S^r(q)] = p(q)$ holds for every $\alpha \in \supp(q)$ once the value of $p(q)$ is known. Thus, if $p(q)$ is a good estimate of $16n/|\supp(q)|$, then each set $S^r(q)$ should contain around $16n$ monomials in expectation.

The main part of the algorithm is $\approxMCDNNFcore$. $\approxMCDNNFcore$ has access to all estimates $p(q)$ and all sample sets $S^r(q)$ and processes the nodes of $\calP$ in height-first order (so a node is processed after its children). Processing a node $q$ means computing a value for $p(q)$ and constructing all sets $S^1(q),\dots,S^{n_sn_t}(q)$. 

The first nodes verify $|\supp(q)| \leq 16n|\calP|^2$, i.e., their effective height is $0$. They are handled in a particular way. The algorithm starts by listing a polynomial amount of monomials of $\supp(q)$. If the enumeration lists fewer than $16n|\calP|^2 + 1$ monomials, then we know that $q \in Q^0$ and we have found $\supp(q)$. So we set $p(q)$ to $\min(1,16n/|\supp(q)|)$ and $S^r(q)$ to $\reduce(\supp(q),p(q))$; this ensures that $\Pr[\alpha \in S^r(q)] = p(q)$(lines~\ref{line:not_enough_monomials} to~\ref{line:S(q)_height_zero} in $\approxMCDNNFcore$). Note that when $|\supp(q)| < 16n$, we can not see $p(q)$ as both an estimate of $16n/|\supp(q)|$ and a probability. Hence our choice to set $p(q) = \min(1,16n/|\supp(q)|)$.

When the enumeration of the support finds that $|\supp(q)| > 16n|\calP|^2$, it is stopped and a specific procedure is called depending on the type of the node. When $q$ is a $\times$ node $q_1 \times q_2$, the procedure $\AndEstimateSample$ computes $p(q)$ from $p(q_1)$ and $p(q_2)$ deterministically in way ensuring that, if $p(q_1)$ and $p(q_2)$ are good estimates of $16n/|\supp(q_1)|$ and $16n/|\supp(q_2)|$, respectively, then $p(q)$ is a good estimate of $16n/|supp(q)|$. Next, each set $S^r(q)$ is constructed by reducing the cross product of $S^r(q_1)$ with $S^r(q_2)$. The multilinearity of $\calP$ guarantees that the cross product creates no duplicate. Without reduction, the size of the sample sets would grow too large and we would drift away from the ideal behavior for the sample sets, i.e., we would violate $\Pr[\alpha \in S^r(q)] = p(q)$.

$+$ nodes are handled by $\OrEstimateSample$. Let $q = q_1 + \dots + q_k$ with $q_1 \prec \dots \prec q_k$. $p(q)$ is computed from the values $p(q_i)$ and the sets $S^r(q_i)$. The idea is to see each $S^r(q_i)$ as if its elements were taken uniformly at random from $\supp(q_i)$. By sampling from the sets $S^r(q_i)$ with  $\reduce$ and rejecting samples with $\union$, we obtain $n_sn_t$ sample sets $\hat{S}^1(q),\dots,\hat{S}^{n_sn_t}(q)$ for $q$. $\union$ keeps a  monomial from $S^r(q_i)$ in $\hat S^r(q_i)$ only if this monomial does not belong to any $\supp(q_j)$ for $q_j \prec q_i$. This sampling method $\reduce$+$\union$ for $+$ nodes is inspired from~\cite{MVC21}. Note that determining whether a monomial $\alpha$ is in $\supp(q_j)$ is feasible in polynomial time. The size of the sets $\hat{S}^r(q)$ is used in a median of means estimator to compute $p(q)$. $\hat{S}^r(q)$ is also used to prepare the sample sets $S^r(q)$ that will later be used to process the parents of $q$. We obtain $S^r (q)$ by simply reducing $\hat{S}^r(q)$ to ensure that $\Pr[\alpha \in S^r(q)] = p(q)$.

\begin{algorithm}[H]
	\begin{algorithmic}[1]
		\STATE Let $\prev(q) = (q_1,\dots,q_k)$
		\FOR{$1 \leq i \leq k$}
		\STATE $S'_i = S_i \setminus (\supp(q_1) \cup \dots \cup \supp(q_{i-1}))$
		\ENDFOR
		\STATE return $S'_1 \cup \dots \cup S'_k$.
	\end{algorithmic}
	\caption{$\union(q,S_1,\dots,S_k)$}
\end{algorithm}

Even though we try to control the size of sample sets, with poor luck some sets might grow large and slow down the algorithm. To ensure a polynomial running time, the algorithm terminates whenever the size of some $S^r(q)$ grows beyond a threshold value $\theta$, in which case the value $0$ is returned as an estimate for $|\supp(\calP)|$. This is  Lines~\ref{line:interrupt1} and~\ref{line:interrupt2} in $\approxMCDNNFcore$. We will later prove that Line~\ref{line:interrupt2} is triggered with low probability.

\begin{algorithm}[H]
	\begin{algorithmic}[1]
		\STATE Let $\prev(q) = (q_1,\dots,q_k)$
		\STATE $\rho(q) = \min(p(q_1),\dots,p(q_k))$
		\FOR{$1 \leq r \leq n_s  n_t$}
		\STATE $\hat{S}^r(q) = \union\left(q,\reduce\left(S^r(q_1),\frac{\rho(q)}{p(q_1)}\right),\dots,\reduce\left(S^r(q_k),\frac{\rho(q)}{p(q_k)}\right)\right)$
		
		\ENDFOR
		\STATE $\hat\rho(q) = 16 n\cdot\left(\underset{0 \leq j < n_t}{\median}\frac{1}{\rho(q)\cdot n_s}\sum\limits_{r = j\cdot n_s+1}^{(j+1)n_s}|\hat{S}^r(q)|\right)^{-1}$ \hfill // Median of means estimator\label{line:median_of_means}
		\STATE $p(q) = \round(q,\min(\rho(q),\hat\rho(q)))$ 
		\FOR{$1 \leq r \leq n_s  n_t$}
		\STATE $S^r(q) = \reduce\left(\hat{S}^r(q),\frac{p(q)}{\rho(q)}\right)$
		\ENDFOR

	\end{algorithmic}
	\caption{$\OrEstimateSample(q,\calP)$}
\end{algorithm}

\begin{algorithm}[H]
	\begin{algorithmic}[1]
	\STATE Let  $\prev(q) = (q_1,q_2)$
	\IF{$p(q_1) = 1$}
		\STATE $p(q) = \round(q,\frac{p(q_2)}{|S^1(q_1)|})$ \hfill // Case when the effective height of $q_1$ is $0$
	\ELSIF{$p(q_2) = 1$}
		\STATE $p(q) = \round(q,\frac{p(q_1)}{|S^1(q_2)|})$ \hfill // Case when the effective height of $q_2$ is $0$
	\ELSE
		\STATE $p(q) = \round(q,\frac{p(q_1)  p(q_2)}{16n})$ \hfill // General case
	\ENDIF	
	\FOR{$1 \leq r \leq n_s  n_t$}
		\STATE $S^r(q) = \reduce\left(S^r(q_1)\otimes S^r(q_2),\frac{p(q)}{p(q_1)  p(q_2)}\right)$
	\ENDFOR
	
	\end{algorithmic}
	\caption{$\AndEstimateSample(q,\calP)$}
\end{algorithm}

For the analysis, we need to restrict the number of possible values for $p(q)$. The acceptable values for $q \in Q^i$, $i > 0$, are $1$ and all $16n e^{-\kappa \cdot 2^i}/\ell$ and $16n e^{+\kappa \cdot 2^i}/\ell$ for integers $\ell$ between $1$ and $2^n$. The function $\round(q,v)$ takes in $q$ and a value $v$ and returns the hightest acceptable value smaller than $v$. Note that the value $\min(1,16n/|\supp(q)|)$ given $p(q)$ to nodes $q$ with $|\supp(q)| \leq 16n|\calP|^2$ is acceptable and thus needs not be rounded.

\begin{algorithm}[h]
	\begin{algorithmic}[1]
		\STATE $o = $ output node of $\calP$
		\IF{$|\supp(o)| \leq 16n|\calP|^2$}
		\RETURN $|\supp(o)|$
		\ENDIF
		\FOR{$q \in \mathsf{sorted}(Q,\mathtt{height})$} 
		\IF{$|\supp(q)| \leq 16n|\calP|^2$}\label{line:not_enough_monomials}
			\STATE $p(q) = \min(1,\frac{16n}{|\supp(q)|})$ 
			\hfill
			// Process nodes with ``few'' monomials \label{line:p(q)_height_zero}
				\FOR{$1 \leq r \leq n_s n_t$}
				\STATE $S^r(q) = \reduce(\supp(q),p(q))$ \label{line:S(q)_height_zero}
				\ENDFOR	
		\ELSE
		\STATE \textbf{if} $q$ is $\times$ node \textbf{then}	$\AndEstimateSample(q,\calP)$
		\hfill // Process $\times$ node
		\STATE \textbf{if} $q$ is $+$ node \textbf{then}	$\OrEstimateSample(q,\calP)$
		\hfill //  Process $+$ node
		\ENDIF
		\FOR{$1 \leq r \leq n_s n_t$}
		\IF{$|S^r(q)| \geq \threshold$}  \label{line:interrupt1} 
		\STATE return $0$
		\hfill // Terminate when sample sets grow large \label{line:interrupt2}
		\ENDIF
		\ENDFOR
   \ENDFOR
	\RETURN $\frac{16n}{p(o)}$
	\end{algorithmic}
	\caption{$\approxMCDNNFcore(\calP,n_s, n_t, \theta)$}
\end{algorithm}

For well-chosen parameters $n_s$, $n_t$ and $\theta$, $\approxMCDNNFcore$ returns an estimate of $|\supp(q)|$ that respects some guarantees with  probability (strictly) greater than $1/2$. A standard median amplification increases the confidence to $1 - \delta$ in algorithm $\approxMCDNNF$.

\begin{algorithm}[h]
	\begin{algorithmic}[1]
		\STATE $n = \deg(\calP)$
		\STATE Transform $\calP$ into an equivalent homogeneous program of depth at most $3\lceil \log(n) \rceil$
		\STATE $\varepsilon' = \min(\varepsilon,\frac{4}{\log(e)})$,  $\kappa = \frac{\varepsilon'}{4(n+1)^3}$, $n_s = \lceil \frac{12}{\kappa^2}  \rceil$,  $n_t = 8n|\calP|$, $\theta = 512n_sn_tn|\calP|$
		\STATE $m = \lceil 16\log(\frac{1}{\delta}) \rceil$
		\FOR{$1 \leq j \leq m$}
		\STATE $\mathsf{est}_j = \approxMCDNNFcore(\calP,n_s,n_t,\theta)$
		\ENDFOR
		\STATE return $\median(\mathsf{est}_1,\dots,\mathsf{est}_m)$
	\end{algorithmic}
	\caption{$\approxMCDNNF(\calP,\varepsilon,\delta)$}
\end{algorithm}

\section{Analysis of \approxMCDNNF.}\label{sec:analysis}

We now focus on the analysis of \approxMCDNNF. The hardest part to analyze is the core algorithm \approxMCDNNFcore, for which we will prove the following.

\begin{lemma}\label{lemma:main_result_core}
	Let $\calP$ be an homogeneous and multilinear $(+,\times)$ program of depth $d \leq 3 \lceil \log n \rceil$ and degree $n$.
	Let $0 < \varepsilon < 4/\log(e)$, $\kappa = \varepsilon/(4n^{3})$, $n_s = \lceil 12/\kappa^2 \rceil$, $n_t = 8n|\calP|$ and $\theta =  512n_s n_t n |\calP|$. The algorithm $\approxMCDNNFcore(\calP,n_s,n_t,\theta)$ runs in time $O(|\calP|^5n^2_sn^2_tn)$ and returns $\mathsf{est}$ with the guarantee
	$
		\Pr\left[\mathsf{est}  \notin (1 \pm \varepsilon) |\supp(\calP)| \right] \leq 1/4
	$.
\end{lemma}

The probability $1/4$ is decreased down to any $\delta > 0$ with the median technique. Thus giving our main result. 

\mainResult*
\begin{proof}
First, $\calP$ is transformed into an equivalent program $\calP'$ using Lemma~\ref{lemma:depth_reduction}. Let $\mathsf{est}_1,\dots,\mathsf{est}_m$ be the estimates returned by $m = \lceil 16\log(1/\delta) \rceil$ independent calls to $\approxMCDNNFcore(\calP',n_s,n_t,\theta)$, with $\kappa$, $n_s$, $n_t$ and $\theta$ set as in Lemma~\ref{lemma:main_result_core}. Let $X_i$ be the indicator variable that takes value $1$ if and only if $\mathsf{est}_i \not\in (1\pm \varepsilon)|\supp(\calP)|$ and let $\bar X = X_1 + \dots + X_m$. By Lemma~\ref{lemma:main_result_core}, $\Ex[\bar X] \leq m/4$. Using Hoeffding bound, we have 
	\[
	\Pr\left[\underset{1 \leq i \leq m}{\median(X_i)} \not\in (1 \pm \varepsilon)|\supp(\calP)|\right] = \Pr\left[\bar X > \frac{m}{2}\right] \leq \Pr\left[\bar X - \Ex[\bar X] > \frac{m}{4}\right] \leq \exp\left(-\frac{m}{8}\right) \leq \delta.
	\]
	Using Lemma~\ref{lemma:main_result_core}, we find that computing a single $\mathsf{est}_j$ takes time $O(n^{15}|\calP'|^7\varepsilon^{-4}) = O(n^{15}|\calP|^{14}\varepsilon^{-4})$. The total running is $m$ times that, plus the time needed by the transformation to $\calP'$.
\end{proof}

In the rest of the section we prove Lemma~\ref{lemma:main_result_core}.

\subsection{Allowing the sample sets to grow large.}

$\approxMCDNNFcore$ has a feature that is not convenient for the analysis, namely, it stops whenever a sample set $S^r(q)$ has size greater than $\theta$. This corresponds to Lines~\ref{line:interrupt1} and~\ref{line:interrupt2} in \approxMCDNNFcore. Recall that  \approxMCDNNFcore$^*$ is the same algorithm $\approxMCDNNFcore$ as without these two lines. So \approxMCDNNFcore$^*$ does not stop when $S^r(q)$ grows big. We argue that working with \approxMCDNNFcore$^*$ is enough to prove Lemma~\ref{lemma:main_result_core} without the running time requirement. In particular, it will be enough to prove the following two lemmas. For these lemmas and after, $\calP$ is an homogeneous and multilinear $(+,\times)$ program of degree $n$ and depth $3\lceil  \log(n) \rceil$. 
 Furthermore, $0 < \varepsilon \leq 4/\log(e)$, $\kappa = \varepsilon/(4(n+1)^3)$, $n_s = \lceil 12/\kappa^2 \rceil$, $n_t = 8n|\calP|$ and $\theta =  512n_s n_tn|\calP|$. We also take $n$ large enough, namely $n \geq 16$.

\begin{lemma}\label{lemma:proba_p(q)}
The probability that $\approxMCDNNFcore^*(\calP,n_s,n_t,\theta)$ computes $p(q) \not\in 16ne^{\pm \kappa 2^i}|\supp(q)|^{-1}$ for some $q \in Q^{> 0}$ of $\calP$ is at most $1/16$, where $i$ is the degree of the node $q$.
\end{lemma}

\begin{lemma}\label{lemma:proba_S(q)}
The probability that $\approxMCDNNFcore^*(\calP,n_s,n_t,\theta)$ constructs a set $S^r(q)$ with $|S^r(q)| \geq \theta$ for some $q \in \calP$ and some $r \in [n_sn_t]$ is at most $1/8$.
\end{lemma}

With these lemmas we can prove Lemma~\ref{lemma:main_result_core}. 

\begin{proof}[Proof of Lemma~\ref{lemma:main_result_core}]
	Let $\calA = \approxMCDNNFcore(\calP,n_s,n_t,\theta)$ and $\calA^* = \approxMCDNNFcore^*(\calP,n_s,n_t,\theta)$. Let $o$ be the output node of $\calP$ and $d$ be the height of $\calP$. If $o \in Q^0$ then $|\supp(\calP)| \leq 16n|\calP|^2$ and the algorithm returns the exact count. Suppose $o \not\in Q^0$. $\calA$'s output is outside of $(1 \pm \varepsilon)|\supp(\calP)|$ whenever some $S^r(q)$ grows larger than $\theta$.
	\[
	\begin{aligned}
		\Pr_\calA\left[\mathsf{est} \not\in (1 \pm \varepsilon)|\supp(\calP)|\right] 
		&= \Pr_{\calA^*}\left[
		\bigcup_{r}
		\bigcup_{q}
		|S^r(q)| \geq \theta
		\right] + \Pr_{\calA^*}\left[
		\bigcap_{r}
		\bigcap_{q}
		|S^r(q)| < \theta\text{ and } \mathsf{est} \not\in (1 \pm \varepsilon)|\supp(\calP)|
		\right] 
		\\
		&\leq \Pr_{\calA^*}\left[
		\bigcup_{r}
		\bigcup_{q}
		|S^r(q)| \geq \theta
		\right] + \Pr_{\calA^*}\left[
		\mathsf{est} \not\in (1 \pm \varepsilon)|\supp(\calP)|
		\right]  
	\end{aligned}
	\]
	where, by default, $q$ ranges over all nodes in $\calP$ and $r$ ranges in $[n_sn_t]$. The event $\mathsf{est} \not\in (1 \pm \varepsilon)|\supp(\calP)|$ occurs when $p(o) \not\in \frac{16n}{(1 \pm \varepsilon)|\supp(o)|}$. The value $\kappa$ has been set so that $p(o) \not\in \frac{16n \cdot e^{\pm \kappa 2^d}}{|\supp(o)|}$ implies $\frac{p(q)}{16n} \not\in \frac{1}{(1 \pm \varepsilon)|\supp(o)|}$ so
		\[
	\begin{aligned}
		\Pr_\calA\left[\mathsf{est} \not\in (1 \pm \varepsilon)|\supp(\calP)|\right] 
		&\leq \Pr_{\calA^*}\left[
		\bigcup_{r}
		\bigcup_{q}
		|S^r(q)| \geq \theta
		\right] + \Pr_{\calA^*}\left[
		p(o) \not\in \frac{16n\cdot e^{\pm \kappa \cdot 2^d}}{|\supp(o)|}
		\right] 
		\\
		&\leq \Pr_{\calA^*}\left[
		\bigcup_{r}
		\bigcup_{q}
		|S^r(q)| \geq \theta
		\right] + \Pr_{\calA^*}\left[
			\bigcup\limits_{i > 0}
		\bigcup\limits_{q \in Q^i}
		p(q) \not\in \frac{16n\cdot e^{\pm \kappa \cdot 2^i}}{|\supp(q)|}
		\right]
	\end{aligned}
	\]	
	The conclusion $\Pr\limits_\calA\left[\mathsf{est} \not\in (1 \pm \varepsilon)|\supp(\calP)|\right]\leq \frac{1}{4}$ directly follows from Lemmas~\ref{lemma:proba_p(q)} and~\ref{lemma:proba_S(q)}.
	
	As for the running time of $\calA$, this algorithm stops as soon as a sample set grows in size beyond~$\theta$ so the worst case running time can be computed assuming all $S^r(q)$ have size $\theta$. When processing $q \in \calP$ with $\AndEstimateSample$ or $\OrEstimateSample$, at most $|\children(q)|\cdot \theta\cdot n_s\cdot n_t = O(|\children(q)|\cdot |\calP| \cdot n_s^2n_t^2n)$ samples are processed. For each sample, it is decided whether it survives a $\reduce$ and whether it survives a $\union$. In $\union$, up to $|\children(q)|$ queries of the form $\alpha \in_? \supp(q')$ for $q' \in \children(q)$ are done to determine the first child of $q$ that contains $\alpha$ in its support. Each query takes time at most $O(|\calP|)$ to anwser. So the total running time for processing $q$ is at most $O(|children(q)|^2 |\calP|^2 n^2_sn^2_tn) = O(|\calP|^4 n^2_sn^2_tn)$. So the running time for all $|\calP|$ nodes is at most $O(|\calP|^5 n^2_sn^2_tn)$
\end{proof}

It remains to prove Lemmas~\ref{lemma:proba_p(q)} and~\ref{lemma:proba_S(q)}.

\subsection[Probability of $p(q)$ being a good approximation (Lemma~\ref{lemma:proba_p(q)})]{Proof of Lemma~\ref{lemma:proba_p(q)}: All Estimates are Tight.}

We start with an easy observation.

\begin{lemma}\label{lemma:p(q)_non_decreasing_and_bounded_by_one}
For all $q \in \calP$, $\approxMCDNNFcore^*$ computes $p(q) \leq 1$ and, for every child $q'$ of $q$, $p(q) \leq p(q')$ holds.
\end{lemma}
\begin{proof}
For nodes with $|\supp(q)| \leq 16n|\calP|^2$, $p(q)$ is deterministically set to $\min(1,16n/|\supp(q)|)$ so $p(q) \leq 1$ holds and $p(q) \leq p(q')$ follows from $|\supp(q')| \leq |\supp(q)|$. For $\times$ nodes, it is clear in all cases of $\AndEstimateSample$ that, since $p(q_1)$ and $p(q_2)$ are less than $1$, $p(q) \leq \min(p(q_1),p(q_2))$. For $+$ nodes, we have that $p(q)$ is $\round(q,\min(\rho(q),\hat \rho(q))) \leq \min(\rho(q),\hat \rho(q)) \leq \rho(q) = \min(p(q') \mid q' \in \children(q))$. So the claim holds.
\end{proof}

For $q \in Q^i$, we denote the interval $\frac{16n e^{\pm \kappa 2^i}}{|\supp(q)|}$ by $\Delta(q)$ (making $i$ precise is not necessary as it is defined by $q$). Given an interval $[a,b]$, we write $\min(1,[a,b])$ instead of $[\min(1,a),\min(1,b)]$. 

\begin{lemma}\label{lemma:introduce_descendants}
	The event 
	$
	\bigcup_{q \in Q^{>0}}
	p(q) \not\in \Delta(q)
	$
	occurs if and only if the following event occurs
	\[
	\bigcup_{q \in Q^{>0}}\left(
	p(q) \not\in \Delta(q) \text{ and for all } q' \in \desc(q),\, p(q') \in \min\left(1,\Delta(q')\right)\right).
	\]
\end{lemma}
\begin{proof}
	The ``if'' direction is immediate. For the ``only if'' direction, let $i$ be the smallest integer strictly greater than $0$ such that there exits $q \in Q^i$ with $p(q) \not\in \Delta(q)$. By construction, for all $q' \in Q^0$ we have $p(q') = \min\left(1,\frac{16n}{|\supp(q')|}\right) \in \min\left(1,\frac{16n e^{\pm \kappa 2^0}}{|\supp(q')|}\right) = \min(1,\Delta(q'))$. By minimality of~$i$, for all $0 < j \leq i$ we have that all $q' \in \desc(q) \cap Q^j$  verify $p(q') \in \Delta(q')$, so we just have show that $\Delta(q') = \min(1,\Delta(q'))$. First, $j$ is at most $3\log(n)$ so $e^{\pm \kappa 2^j} \subseteq e^{\pm \kappa n^3} = e^{\pm \varepsilon/4}$. Second, $j > 0$ implies that $|\supp(q')| \geq 16n|\calP|^2 \geq 32n \geq 16n  e^{\varepsilon / 4}$ (because $\varepsilon \leq 4/\log(e)$). So $\frac{16n  e^{\kappa 2^j}}{|\supp(q')|} \leq \frac{16n  e^{\varepsilon/4}}{16n e^{\varepsilon/4}} = 1$ and therefore $\min(1,\Delta(q')) = \Delta(q')$.
\end{proof}

\noindent We apply Lemma~\ref{lemma:introduce_descendants} followed by a union bound to the target probability of Lemma~\ref{lemma:proba_p(q)}.
\begin{equation}\label{eq:probability_real_bad_event}
	\begin{aligned}
		\Pr\left[
		\bigcup_{q \in Q^{>0}}
		p(q) \not\in \Delta(q)\right] &= \Pr\left[
		\bigcup_{q \in Q^{>0}}
		p(q) \not\in \Delta(q)
	 \text{ and }\forall q' \in \desc(q),\, p(q') \in \min\left(1, \Delta(q')\right)
		\right]
		\\
		&\leq \sum_{q \in Q^{>0}}\underbrace{\Pr\left[
		p(q) \not\in \Delta(q)
 \text{ and }\forall q' \in \desc(q),\, p(q') \in \min\left(1,\Delta(q')\right)
		\right]}_{P(q)}.
	\end{aligned}
\end{equation}
We bound $P(q)$ separately for $+$ nodes and $\times$ nodes in the next sections and get the following.
\begin{lemma}\label{lemma:P(q)_for_times_nodes}
Let $q \in Q^{> 0}$ be a $\times$ node. Then $P(q) = 0$.
\end{lemma}

\begin{lemma}\label{lemma:P(q)_for_plus_nodes}
Let $q \in Q^{> 0}$ be a $+$ node. Then $P(q) \leq \frac{1}{16|\calP|}$.
\end{lemma} 
With these lemmas, we easily finish the proof of Lemma~\ref{lemma:proba_p(q)}. 
\begin{align*}
	\Pr\left[\bigcup_{q \in Q^{> 0}} p(q) \not\in \Delta(q) \right] \leq \sum_{q \in Q^{> 0}} P(q) \leq \sum_{q \in Q^{> 0}} \frac{1}{16|\calP|} \leq \frac{1}{16}
\end{align*}

Before proving Lemmas~\ref{lemma:P(q)_for_times_nodes} and~\ref{lemma:P(q)_for_plus_nodes}, it is  good to discuss the effect of rounding $p(q)$. A value is \emph{acceptable for $q \in Q^i$} if it is $1$ or if it is less than $1$ and of the form
$
16 n e^{-\kappa 2^{i}}/\ell 
$ or
$
16 n e^{+\kappa 2^{i}}/\ell 
$  
for some integer $\ell$ between $1$ and $2^n$. The function $\round(q,v)$ takes in $q$ and a value $v$ and returns the highest acceptable value less than or equal to $v$. It is fairly immediate that if $v$ is in $\min(1,\Delta(q))$, then so is $\round(q,v)$.

\begin{lemma}\label{lemma:rounding_of_good_is_still_good}
If $v \in \min(1,\Delta(q))$ then $\round(q,v) \in \min(1,\Delta(q))$.
\end{lemma}
\begin{proof}
Since $|\supp(q)| \leq 2^n$, this follows from the lower and upper limits of $\min(1,\Delta(q))$ being acceptable values for $q$.
\end{proof}

\subsubsection{Proof of Lemma~\ref{lemma:P(q)_for_times_nodes}: the First Loose Estimate is not Witnessed by a $\times$ Node.} First we show that only nodes $q$ of effective height $0$ can have $p(q) = 1$.

\begin{lemma}\label{lemma:case_p(q)_equals_1}
If $p(q) = 1$ then $q \in Q^0$.
\end{lemma}
\begin{proof}
By Lemma~\ref{lemma:p(q)_non_decreasing_and_bounded_by_one}, it suffices to show that, for all $q \in Q^{1}$, $p(q) < 1$ holds. When $q$ has effective height $1$, the effective height of all its children is $0$. So $p(q') = \min(1,16n/|\supp(q')|)$ for every child $q'$ of $q$ (Line~\ref{line:p(q)_height_zero} of $\approxMCDNNFcore^*$).

Suppose $q = q_1 \times q_2$ and that $\AndEstimateSample$ sets $p(q)$ to $1$. Then $p(q_1) = 1$ and $p(q_2)/|S^1(q_1)| = 1$, or  $p(q_2) = 1$ and $p(q_1)/|S^1(q_2)| = 1$. In both cases, we have $p(q_1) = p(q_2) = 1$ and therefore $\max(|\supp(q_1)|,|\supp(q_2)|) \leq 16n$. By multilinearity of $\calP$, it follows that $|\supp(q)| = |\supp(q_1)|\cdot |\supp(q_2)| \leq 256n^2$. Since $|\calP| \geq n$, we have $|\calP| \geq 4\sqrt{n}$ for $n \geq 16$ and thus $16n|\calP|^2 \geq 256n^2$. But then the effective height of $q$ is~$0$. This is contradictory, so $p(q) < 1$.

Suppose $q = q_1 + \dots + q_k$ and that $\OrEstimateSample$ sets $p(q)$ to $1$. Since $p(q)  = \min(p(q_1),\dots,$ $p(q_k))$, we have that $p(q) = 1$ if and only if $p(q_1) = \dots = p(q_k) = 1$. But then $\max(|\supp(q_1)|,\dots,$ $|\supp(q_k)|) \leq 16n$. It follows that $|\supp(q)| \leq 16nk \leq 16n|\calP|$ and thus the effective height of $q$ is~$0$. This is contradictory, so $p(q) < 1$.
\end{proof}

\begin{proof}[Proof of Lemma~\ref{lemma:P(q)_for_times_nodes}]
Let $q_1$ and $q_2$ be the children of $q$. We suppose $p(q_1) \in \min(1,\Delta(q_1))$ and $p(q_2) \in \min(1,\Delta(q_2))$. Let $i$ be the effective height of $q$, $j$ be the effective height of $q_1$ and $k$ be the effective height of $q_2$. By assumption $i > 0$ and thus $\max(j,k) = i - 1$. 
We show that 
$$
i < 0 \text{ and } p(q_1) \in \min(1,\Delta(q_1)) \text{ and } p(q_2) \in \min(1,\Delta(q_2)) \quad \Rightarrow \quad p(q) \in \Delta(q) 
$$
Suppose $p(q_1) = p(q_2) = 1$. Lemma~\ref{lemma:case_p(q)_equals_1} implies that $q_1,q_2 \in Q^0$. Then the same argument as in the second paragraph of the proof of Lemma~\ref{lemma:case_p(q)_equals_1} shows that $q \in Q^0$, which is a contradiction. For the rest of the proof suppose, without loss of generality, that $p(q_1) < 1$ holds. Then $p(q_1) \in \min(1,\Delta(q_1))$ implies $p(q_1) \in \Delta(q_1)$. 

If $p(q_2) = 1$ then, by Lemma~\ref{lemma:case_p(q)_equals_1}, we have that $q_2 \in Q^0$ and thus $|S^1(q_2)| = |\reduce(\supp(q_2),1)| = |\supp(q_2)|$. In this case, $\AndEstimateSample$ computes $p(q) = \round(q,p(q_1)/|S^1(q_2)|)$. By Lemma~\ref{lemma:rounding_of_good_is_still_good}, it is sufficient to show that $p(q_1)/|S^1(q_2)| \in \Delta(q)$.
$$
\frac{p(q_1)}{|S^1(q_2)|} = \frac{p(q_1)}{|\supp(q_2)|} \in  \frac{\Delta(q_1)}{|\supp(q_2)|} = \frac{16n\cdot e^{\pm \kappa 2^j}}{|\supp(q_1)|\cdot |\supp(q_2)|}
$$ 
By multilinearity of $\calP$, $|\supp(q_1)|\cdot |\supp(q_2)| = |\supp(q)|$, so $p(q) \in \frac{16n\cdot e^{\pm \kappa 2^j}}{|\supp(q)|} \subseteq \frac{16n\cdot e^{\pm \kappa 2^i}}{|\supp(q)|} = \Delta(q)$. 

If $p(q_2) < 1$, then $p(q_2) \in \min(1,\Delta(q_2))$ implies $p(q_2) \in \Delta(q_2)$. In this case, $\AndEstimateSample$ computes $p(q) = \round(q,p(q_1)p(q_2)/16n)$. By  Lemma~\ref{lemma:rounding_of_good_is_still_good}, it is sufficient to show that $p(q_1)p(q_2)/16n \in \Delta(q)$.
$$
\frac{p(q_1)p(q_2)}{16n} \in \frac{16n\cdot e^{\pm \kappa(2^j + 2^k)}}{|\supp(q_1)|\cdot |\supp(q_2)|} = \frac{16n\cdot e^{\pm \kappa(2^j + 2^k)}}{|\supp(q)|}
$$ Since $\max(j,k) = i - 1$, we have that $2^j + 2^k \leq 2^i$ and therefore $p(q) \in \Delta(q)$.
\end{proof}

\subsubsection{Interlude: the $\mathfrak{S}$-process.}\label{section:randomProcess}

To prove Lemma~\ref{lemma:P(q)_for_plus_nodes}, we introduce the $\mathfrak{S}$-process, a random process efined with respect to the $(+,\times)$ program $\calP$ that simulates $\approxMCDNNFcore^*$. Our intuition is that, for every $\alpha \in \supp(q)$, we should have something similar to ``$\Pr[\alpha \in S^r(q)] = p(q)$''. The problem is that this equality makes no sense because $\Pr[\alpha \in S^r(q)]$ is a real value between $0$ and $1$ whereas $p(q)$ is a random variable whose value may change with every run of $\approxMCDNNFcore^*$. The $\mathfrak{S}$-process allows us to circumvent this issue by replacing $S^r(q)$ with a variable easier to analyze. $S^r(q)$ has the same distribution as variables from the $\mathfrak{S}$-process given some knowledge on the prior computation made by $\approxMCDNNFcore^*$. This knowledge is encoded in the notion of \emph{history}. 

\begin{definition} A \emph{history} $h$ for a set of nodes $Q$ of $\calP$ is a mapping
$
h : Q \rightarrow \mathbbm{Q} \cap (0,1]
$. 
The history is \emph{realizable} when  there exists a run of $\approxMCDNNFcore^*$ that gives the value $h(q)$ to $p(q)$ for every $q \in Q$. When this happens, we say that the run of  $\approxMCDNNFcore^*$ is \emph{compatible} with $h$. Two histories $h$ for $Q$ and $h'$ for $Q'$ are \emph{compatible} when $h(q) = h'(q)$ for all $q \in Q \cap Q'$; we then denote by $h \cup h'$ the merged history to $Q \cup Q'$. For $q \in Q$ and $t \in \mathbbm{Q} \cap (0,1]$, we use the notation $h \cup (q \mapsto t)$ to augment $h$ with $h(q) = t$. 
\end{definition}

For $q$ a node in $\calP$, we denote by $\desc(q)$ its set of descendants in $\calP$. That is, $\desc(q) = \emptyset$ when $q$ is a variable and $\desc(q) = \children(q) \cup \bigcup_{q' \in \children(q)} \desc(q')$ otherwise. Note that $q \not\in \desc(q)$. We only study histories that are realizable for sets $Q$  closed under progeny, that is, if $q \in Q$ and $q'$ is a child of $q$, then $q' \in Q$. Thus, for simplicity, we write ``history for $q$'' instead of ``history for $\desc(q)$''.  The only possible history for an input node of $\calP$ is the empty set (because no descendant). 

\begin{definition}
The \emph{$\mathfrak{S}$-process} comprises $n_sn_t$ independent copies identified by the superscript $r$. Random variables $\vY{r}{h}{t}{q}$ with domain the subsets of $\supp(q)$ are defined for $q \in \calP$, $t \in \mathbbm{Q} \cap (0,1]$ and $h$ a realizable history for $q$. $\vY{r}{h}{t}{q}$ is meant to simulate $S^{r}(q)$ when the history of $q$ is $h$ and the value $t$ is assigned to $p(q)$. %
\begin{enumerate}[leftmargin=*]
\item[•] If $|\supp(q)| \leq 16n|\calP|^2$, i.e., if $q \in Q^0$, then $\vY{r}{h}{t}{q}$ is only defined for $t = \min(1,\frac{16n}{|\supp(q)|})$ by $\vY{r}{h}{t}{q} = \reduce(\supp(q),t)$, with $r \in [n_sn_t]$ and $h$ any realizable history for $q$.
\item[•] If $q$ is a $\times$ node $q = q_1 \times q_2$ in $Q^{> 0}$ then for every $\vY{r}{h_1}{t_1}{q_1}$ and $\vY{r}{h_2}{t_2}{q_2}$ such that $h_1$ and $h_2$ are realizable and compatible histories for $q_1$ and $q_2$, we define $h = h_1 \cup h_2 \cup (q_1 \mapsto t_1,q_2 \mapsto t_2)$ and, for every $t \leq t_1t_2$, we define $\vY{r}{h}{t}{q}$ as 
$$
\vY{r}{h}{t}{q} = \reduce(\vY{r}{h_1}{t_1}{q_1} \otimes \vY{r}{h_2}{t_2}{q_2}, {\textstyle \frac{t}{t_1t_2}})
$$
\item[•]  If $q$ is a $+$ node $q = q_1 + \dots + q_k$ in $Q^{> 0}$ with $\children(q) = (q_1,\dots,q_k)$ then, for every $\vY{r}{h_1}{t_1}{q_1},\dots,\vY{r}{h_k}{t_k}{q_k}$ where the histories $h_i$ for $q_i$ are realizable and pairwise compatible, we define $h = h_1 \cup \dots \cup h_k \cup (q_1 \mapsto t_1,\dots,q_k \mapsto t_k)$. Next we define $t_{\min} = \min(t_1,\dots,t_k)$ and the random variable $\vZ{r}{h}{q}$ as
$$
\vZ{r}{h}{q} = \union(q,\reduce(\vY{r}{h_1}{t_1}{q_1},{\textstyle \frac{t_{\min}}{t_1}}) , \dots ,\reduce(\vY{r}{h_k}{t_k}{q_k},{\textstyle \frac{t_{\min}}{t_k}}))
$$
and, for every $t \leq t_{\min}$ we define $\vY{r}{h}{t}{q}$ as $\vY{r}{h}{t}{q} = \reduce(\vZ{r}{h}{q},\frac{t}{t_{\min}})$. 
\end{enumerate}
\end{definition}

The $\mathfrak{S}$-process simultaneously plays runs of $\approxMCDNNFcore^*$ for all possible histories. So we can analyze $\approxMCDNNFcore^*$ from the $\mathfrak{S}$-process, provided we know the history. Let $H(q)$ be the random variable on the history for $q$  for a run of $\approxMCDNNFcore^*$.  

\begin{restatable}{lemma}{notSoBlackMagic}\label{lemma:not_so_black_magic}
Let $H(q)$ be the random variable recording the history for $q$ in a run of $\approxMCDNNFcore^*$. Let $h$ be a realizable history for $q$. Let $e(\hat S^1(q),\dots,\hat S^\nsnt(q))$ be an event that is function of $\hat S^1(q),\dots,\hat S^\nsnt(q)$ and let $e(\vZ{1}{h}{q},\dots,\vZ{\nsnt}{h}{q})$ be the same event where each $\hat S^r(q)$ is replaced by $\vZ{r}{h}{q}$. Then 
$$
\Pr[H(q) = h \text { and } e(\hat S^1(q),\dots,\hat S^\nsnt(q))] \leq \Pr[e(\vZ{1}{h}{q},\dots,\vZ{\nsnt}{h}{q})]
$$
Similarly, let $e(S^1(q),\dots,S^\nsnt(q))$ be an event that is function of $S^1(q),\dots,S^\nsnt(q)$ and let $e(\vY{1}{h}{t}{q},\dots,\vY{\nsnt}{h}{t}{q})$ be the same event where each $S^r(q)$ is replaced by $\vY{r}{h}{t}{q}$. Then 
$$
\Pr[H(q) = h \text { and } p(q) = t \text { and } e(S^1(q),\dots,S^\nsnt(q))] \leq \Pr[e(\vY{1}{h}{t}{q},\dots,\vY{\nsnt}{h}{t}{q})]
$$
\end{restatable}
We defer the proof of Lemma~\ref{lemma:not_so_black_magic} to Appendix~\ref{sec:proof-blackmagic}. To illustrate why the lemma holds, we provide a simple example in the same veins. Consider a random variable $X$ that takes value in $\{1,\dots,n\}$ for some fixed $n > 0$, and two random variables $Y$ and $Z$ following a binomial distribution whose first parameter is the value taken by $X$ and whose second parameter is a fixed probability $p_1$ for $Y$ and $p_2$ for $Z$. Formally, $Y$ follows the distribution $\mathcal{B}(X,p_1)$ and $Z$ follows $\mathcal{B}(X,p_2)$. Now, let $Y_1,\dots,Y_n,Z_1,\dots,Z_n$ be random variables that are mutually independent and independent of $X$, $Y$ and $Z$, such that for all $i \in \{1,\dots,n\}$, $Y_i$ follows $\mathcal{B}(i,p_1)$ and $Z_i$ follows $\mathcal{B}(i,p_2)$. We have that $(X,Y,Z)$ and $(X,Y_X,Z_X)$ have identical distribution, i.e., $\Pr[X = i \text{ and } Y = j \text{ and } Z = k] = \Pr[X = i \text{ and } Y_i = j  \text{ and } Z_i = k] \leq \Pr[Y_i = j  \text{ and } Z_i = k]$ holds for all $i$, $j$ and $k$. This is similar to the connection between the $\mathfrak{S}$-process with $\approxMCDNNFcore^*$: the random variables $H(q)$ and $p(q)$ act as parameters in the distributions of $S^r(q)$ and $\hat S^r(q)$ for all $r \in [n_sn_t]$, and the random sets $\vY{r}{h}{t}{q}$ and $\vZ{r}{h}{q}$ simulates $S^r(q)$ and $\hat S^r(q)$ when the random parameters  $H(q)$ and $p(q)$ are fixed to $h$ and $t$, respectively.

 Whereas in $\approxMCDNNFcore^*$ there may be dependence between two sets $S^r(q)$ and $S^{r'}(q')$, in the $\mathfrak{S}$-process $\vY{r}{h}{t}{q}$ and $\vY{r'}{h'}{t'}{q'}$ are independent when $r \neq r'$. Random sets $\vY{r}{h}{t}{q}$ and $\vY{r}{h'}{t'}{q'}$ for the same $r$ are also independent when $var(q) \cap var(q') = \emptyset$ because, although they both belong to the $r$th subprocess, they are constructed by disjoint procedures.

\begin{fact}\label{fact:independence_in_random_process}
Let $m \in \mathbb{N}$. If the random sets $\{\vY{r_j}{h_j}{t_j}{q_j}\}_{j \in [m]}$ (resp. $\{\vZ{r_j}{h_j}{q_j}\}_{j \in [m]}$)  verify that $r_i \neq r_j$ for every $i \neq j$, then they are mutually independent.
\end{fact}

\begin{restatable}{fact}{decomposeIndependence}\label{fact:independence}
Let $r \in [n_sn_t]$ be fixed and let $q_1,\dots,q_k$ be nodes in $\calP$ such that $\var(q_i) \cap \var(q_j) = \emptyset$ for all $i\neq j$. Then the events $\alpha_1 \in \vY{r}{h_1}{t_1}{q_1},\dots,\alpha_k \in \vY{r}{h_k}{t_k}{q_k}$ are mutually independent for all histories $h_1,\dots,h_k$ for $q_1,\dots,q_k$, respectively.
\end{restatable}

In the $\mathfrak{S}$-process, we can find the exact probability that a monomial $\approxMCDNNFcore^*$ is sampled in $\vZ{r}{h}{q}$ or in $\vY{r}{h}{t}{q}$. The proof of the following lemmas appear in appendix.

\begin{restatable}{lemma}{probaFirstOrder}\label{lemma:proba_first_order}
For every $\vY{}{h}{t}{q}$ and $\alpha \in \supp(q)$ we have 
$
\Pr[\alpha \in \vY{}{h}{t}{q}] = t.
$
And if $q$ is a $+$ node $q = q_1 + \dots + q_k$ of effective weight $> 0$, then $\Pr[\alpha \in \vZ{}{h}{q}] = \min(h(q_1),\dots,h(q_k))$. 
\end{restatable}

\begin{restatable}{lemma}{probaSecondOrder}\label{lemma:proba_second_order}
For every $\vY{r}{h}{t}{q}$ with $q \in Q^{> 0}$ and $\alpha,\alpha' \in \supp(q)$ with $\alpha \neq \alpha'$ we have,
\[
\Pr\left[\alpha \in \vY{r}{h}{t}{q} \text{ and } \alpha' \in \vY{r}{h}{t}{q}\right] = \frac{t^2}{\prod_{\hat q \in \tau} h(\hat q)}
\]
where $\tau = \tree^*(\alpha,q) \wedge \tree^*(\alpha',q)$. Moreover, if $q$ is a $+$ node, then 
\[
\Pr\,[\alpha \in \vZ{r}{h}{q} \text{ and } \alpha' \in\vZ{r}{h}{q}] = \frac{\min(h(q_1),\dots,h(q_k))^2}{\prod_{\hat q \in \tau} h(\hat q)}.
\]
\end{restatable}

\subsubsection{Proof of Lemma~\ref{lemma:P(q)_for_plus_nodes}: Bounding the Probability of a Loose Estimate for $+$ Nodes.}

Let $q = q_1 + \dots + q_k$ be a $+$ node in $Q^{> 0}$. First, we introduce histories in the expression of $P(q)$. We consider all realizable histories for $q$ and denote by $H(q) = h$ the event that the history for $q$ is $h$, that is, the event that the algorithm sets $p(q')$ to $h(q')$ for all descendants $q'$ of $q$. Let $\calH_q$ be the set of realizable histories for $q$.
\begin{align*}
	P(q) = \sum_{h \in \calH_q} \Pr\left[
	 H(q)=h \text{ and }
	p(q) \not\in \Delta(q)
	\text{ and for all } q' \in \desc(q),\, p(q') \in \min(1,\Delta(q'))
	\right].
\end{align*}
For every descendant $q'$ of $q$, $H(q) = h$ implies that $p(q') = h(q')$. So, a summand probability is zero when any $h(q')$ is not in $\min(1,\Delta(q'))$. Let $\calH^*_q$ be the set of realizable histories for $q$ such that $h(q') \in \min(1,\Delta(q'))$ holds for every descendant $q'$ of $q$. Then
\begin{align*}
		P(q) = \sum_{h \in \calH^*_q}
		\Pr\left[ H(q)=h \text{ and }
		p(q) \not\in \Delta(q)
		\right].
\end{align*}

\begin{lemma}\label{lemma:number_of_acceptable_good_histories}
Let $q$ be a node in the $(+,\times)$ program $\calP$. There are at most $2^{(n+1)|\calP|}$ histories for $q$ such that, for all descendants $q'$ of $q$, $h(q')$ is in $\min(1,\Delta(q'))$ and is an acceptable value for $q'$. Thus, if $\calH^*_q$ is the set of realizable histories for $q$, then $|\calH^*_q| \leq 2^{(n+1)|\calP|}$.
\end{lemma}
\begin{proof}
There are between $1$ and $2^{n+1}$ acceptable values for $q$ in $\min(1,\Delta(q))$. Indeed, there is at least one acceptable value for $q$ in $\min(1,\Delta(q))$ because the lower limit of $\min(1,\Delta(q))$ is acceptable. The $2^{n+1}$ upper bound follows from the definition. The Lemma~\ref{lemma:number_of_acceptable_good_histories} follows from the fact that there are less that $|\calP|$ descendants of $q$.
\end{proof}

\noindent Recall that $\rho(q) = \min(p(q_1),\dots,p(q_k))$. Let $M_j(q)$ be the mean $\frac{1}{\rho(q)n_s}\sum_{r = j\cdot n_s+1}^{(j+1)n_s}|\hat{S}^{r}(q)|$. The algorithm computes $p(q)$ as $\min(\rho(q),\hat \rho(q))$ with $\hat \rho(q) = 16 n\cdot \median(M_1(q),\dots,M_{n_t}(q))^{-1}$. Now let $\rho_h(q) = \min(h(q_1),\dots,h(q_k))$, $M_{j,h}(q) = \frac{1}{\rho_h(q) n_s}\sum_{r = j\cdot n_s+1}^{(j+1)n_s}|\hat{S}^{r}(q)|$, $\hat \rho_h = 16 n\cdot \median(M_{1,h}(q),\dots,M_{n_t,h}(q))^{-1}$ and $p_h(q) = \min(\rho_h(q),\hat \rho_h(q))$. Then 
\begin{align*}
		\Pr\left[ H(q)=h \text{ and }
		p(q) \not\in \Delta(q)
		\right] =  
		\Pr\left[ H(q)=h \text{ and }
		p_h(q) \not\in \Delta(q)
		\right].
\end{align*}

\begin{lemma}\label{lemma:use_rho_hat}
	If $h \in \calH^*_q$ and $\hat \rho_h(q) \in \Delta(q)$ then $p_h(q) \in \Delta(q)$.
\end{lemma}
\begin{proof}
Suppose $q \in Q^i$, $i > 0$. If $p_h(q) = \hat \rho_h(q)$ then the claim is trivial. Otherwise, if $p_h(q) = \rho_h(q)$, then $\hat \rho_h(q) \in \Delta(q)$ implies $p_h(q) \leq \hat \rho_h(q) \leq \frac{16n e^{ \kappa 2^i}}{|\supp(q)|}$. It remains to prove that $\rho_h(q) \geq \frac{16n e^{-\kappa 2^i}}{|\supp(q)|}$. Since $i > 0$, by Lemma~\ref{lemma:case_p(q)_equals_1}, there is a $q_j \in \children(q)$ such that $p_h(q_j) < 1$ and, since $h \in \calH^*_q$, it follows that $p_h(q_j) \in \Delta(q_j)$. Thus there is a $j \in [k]$ such that $\rho_h(q) = p_h(q_j) \in \Delta(q_j)$. Let $i'$ be the effective height of $q_j$. Because $i' \leq i$ and $|\supp(q_j)| \leq |\supp(q)|$, we have $\rho_h(q)\geq \frac{16n e^{- \kappa 2^{i'}}}{|\supp(q_j)|} \geq \frac{16n e^{- \kappa 2^i}}{|\supp(q)|}$.
\end{proof}
Let $M_j(q) = \frac{1}{\rho(q)n_s}\sum_{r = j\cdot n_s+1}^{(j+1)n_s}|\hat S^r(q)|$ and $\mathfrak{M}_{j,h}(q) = \frac{1}{\rho_h(q)n_s}\sum_{r = j\cdot n_s+1}^{(j+1)n_s}|\vZ{r}{h}{q}|$. We denote the interval  $\frac{|\supp(q)|}{16n e^{\pm \kappa 2^i}} = \frac{|\supp(q)| e^{\pm \kappa 2^i}}{16n}$ by $\nabla(q)$. 
\begin{align*} 
		\Pr\left[ H(q)=h \text{ and }
		p_h(q) \not\in \Delta(q)
		\right] &\leq  
		\Pr\left[ H(q)=h \text{ and }
		\hat \rho_h(q) \not\in \Delta(q)
		\right] \tag{Lemma~\ref{lemma:use_rho_hat}}
		\\
		&= \Pr\left[ H(q)=h \text{ and }
		16 n\cdot\median(M_1(q),\dots,M_{n_t}(q))^{-1} \not\in \Delta(q)\right] 
		\\
		&\leq \Pr\left[
		16 n\cdot\median(\mathfrak{M}_{1,h}(q),\dots,\mathfrak{M}_{n_t,h}(q))^{-1} \not\in \delta(q)\right] \tag{Lemma~\ref{lemma:not_so_black_magic}}
		\\
		&\leq \Pr\left[(16 n)^{-1}\median(\mathfrak{M}_{1,h}(q),\dots,\mathfrak{M}_{n_t,h}(q)) \not\in \nabla(q)\right]
		\\		
		&\leq \Pr\left[\median(\mathfrak{M}_{1,h}(q),\dots,\mathfrak{M}_{n_t,h}(q)) \not\in |\supp(q)|\cdot e^{\pm \kappa 2^i}\right] 
\end{align*}

\noindent 
We bound 
$\Pr[\median_j(\mathfrak{M}_{j,h}(q)) \not\in |\supp(q)|\cdot e^{\pm \kappa 2^i}]$
using Chebyshev's inequality  followed by Chernoff-Hoeffding bound. We aim to prove the following.

\begin{lemma}\label{lemma:median_bound}
When $h \in \calH^*_q$, $\Pr[\median(\mathfrak{M}_{1,h}(q),\dots,\mathfrak{M}_{n_t,h}(q)) \not\in |\supp(q)|\cdot e^{\pm \kappa 2^i}] \leq \exp(- n|\calP|)$.
\end{lemma}

\noindent Equipped with Lemma~\ref{lemma:median_bound}, we have 
\begin{align*}
P(q) &\leq \sum_{h \in \calH^*_q} \Pr[\median(\mathfrak{M}_{1,h}(q),\dots,\mathfrak{M}_{n_t,h}(q)) \not\in |\supp(q)|\cdot e^{\pm \kappa 2^i}]
\\
& \leq \sum_{h \in \calH^*_q} \exp(-n|\calP|) \tag{Lemma~\ref{lemma:median_bound}}
\\
& \leq \frac{1}{16|\calP|}  \tag{Lemma~\ref{lemma:number_of_acceptable_good_histories} and $n \geq 16$} 
\\
\end{align*}
It remains to prove Lemma~\ref{lemma:median_bound} using Chebyshev and Chernoff bound. This requires a variance computation. From Lemma~\ref{lemma:proba_first_order}, we have that 
$$
\mu := \Ex\left[|\vZ{r}{h}{q}|\right] = \rho_h(q)\cdot|\supp(q)|.
$$ The computation of the variance is arguably the core technical point of this paper.

\begin{lemma}\label{lem:variance-bound}
	When $h \in \calH^*_q$, $\Va\left[|\vZ{r}{h}{q}|\right] \leq \mu + 2 \mu^2$.
\begin{proof} 
\begin{align*}
	\Va\left[ |\vZ{r}{h}{q}| \right] &\leq \Ex\left[|\vZ{r}{h}{q}|^2 \right] = \mu + \sum_{\alpha \in \supp(q)}\sum_{\substack{\alpha' \in \supp(q) \\ \alpha \neq \alpha'}} \Pr\left[\alpha \in \vZ{r}{h}{q} \text{ and } \alpha' \in \vZ{r}{h}{q}\right]
	\\
	&\leq \mu + \sum_{\alpha \in \supp(q)}\sum_{\substack{\alpha' \in \supp(q) \\ \alpha \neq \alpha'}} \frac{\rho_h(q)^2}{
		\prod\limits_{\hat{q} \in \tau(\alpha,\alpha',q)} h(\hat{q})
		} \tag{Lemma~\ref{lemma:proba_second_order}} 
\end{align*}
where $\tau(\alpha,\alpha',q) = \tree^*(\alpha,q) \wedge \tree^*(\alpha',q)$. Let $V = \{\hat q : |\supp(\hat q)| > 16n\}$. We have $h(\hat q) = 1$ for all nodes not in $V$ for otherwise $h$ would not be realizable. Thus
\begin{align*}
	\Va\left[ |\vZ{r}{h}{q}| \right] - \mu &\leq  \sum_{\alpha \in \supp(q)}\sum_{\substack{\alpha' \in \supp(q) \\ \alpha \neq \alpha'}} \frac{\rho_h(q)^2}{
		\prod\limits_{\hat{q} \in V \cap \tau(\alpha,\alpha',q)} h(\hat{q})
		} 
\end{align*}
 $V \cap \tau(\alpha,\alpha',q)$ can be written as $\lcsnvar{\tree(\alpha,q)}{\tree(\alpha',q)}{f(\alpha,\alpha',q)}$, where $f(\alpha,\alpha',q)$ is the intersection of $V$ with the nodes used in $\tree^*(\alpha,q)$ and $\tree^*(\alpha',q)$. Recall that for $\sigma$ an antichain over the nodes of $\tree(\alpha,q)$, $I_f(\alpha,q,\sigma)$ is the set of $\alpha' \in \supp(q)$ such that $\lcsnvar{\tree(\alpha,q)}{\tree(\alpha',q)}{f(\alpha,\alpha',q)} = \sigma$. Let $A_\alpha$ be the set of antichains of $\tree^*(\alpha,q)$ excluding $\{q\}$. The sets $I_f(\alpha,q,\sigma)$ for all $\sigma \in A_\alpha$ partitions $\supp(q) \setminus \{\alpha\}$. Using Lemma~\ref{lemma:antichain_lemma}:

\begin{align*}
	\Va\left[ |\vZ{r}{h}{q}| \right] - \mu \leq \sum_{\alpha \in \supp(q)} \sum_{\sigma \subseteq A_\alpha} |I_f(\alpha,q,\sigma)|\frac{\rho_h(q)^2}{\prod\limits_{\hat{q} \in \sigma} h(\hat{q})} 
	\leq \sum_{\alpha \in \supp(q)} \sum_{\sigma \subseteq A_\alpha} \frac{ \rho_h(q)^2 \cdot |\supp(q)|}{\left(\prod\limits_{\hat{q} \in \sigma} |\supp(\hat{q})|\right)\left(\prod\limits_{\hat{q} \in \sigma}h(\hat{q})\right)}
\end{align*}
When $h(\hat q) < 1$ holds we have $h(\hat q) \in \Delta(\hat q)$ because $h \in \calH^*_q$. Thus $|\supp(\hat q)|\cdot h(\hat q) \geq 16n\cdot e^{-\kappa 2^i}$. The depth of $\calP$ is at most $3\log(n)$ so, since $\varepsilon \leq 4/\log(e)$,
$$
|\supp(\hat q)|\cdot h(\hat q) \geq 16n\cdot e^{-\kappa 2^i} \geq 16n\cdot e^{-\kappa \cdot n} = 16n\cdot e^{-\varepsilon/4} \geq 8n.
$$ 
Recall that a derivation tree of $\calP$ has $n$ leaves, so the maximum size of an antichain is $n$.
\begin{align*}
	\Va\left[ |\vZ{r}{h}{q}| \right] &\leq \mu + \rho_h(q)^2  \sum_{\alpha \in \supp(q)} \sum_{\sigma \in A_\alpha} \frac{|\supp(q)|}{(8n)^{|\sigma|}}
	= \mu + \rho_h(q)^2  \sum_{\alpha \in \supp(q)} \sum\limits_{\ell=0}^{n} \;  
	\sum_{\substack{\sigma \in A_\alpha \\ |\sigma| = \ell} } \frac{|\supp(q)|}{(8n)^{\ell}}
\end{align*}
A derivation tree has at most $4n$ nodes. So there are at most ${4n \choose \ell}$ antichains of size $\ell$.
\begin{align*}
	\Va\left[ |\vZ{r}{h}{q}| \right] &\leq \mu + \rho_h(q)^2  \sum_{\alpha \in \supp(q)} \sum\limits_{\ell=0}^{n} \;  {4n \choose \ell} \frac{|\supp(q)|}{(8n)^{\ell}} 
\\
&\leq \mu + \rho_h(q)^2 \sum_{\alpha \in \supp(q)} |\supp(q)|\left(1+\frac{1}{8n}\right)^{4n}
 \tag{binomial identity}\\ 
&\leq \mu + \rho_h(q)^2 \sum_{\alpha \in \supp(q)} |\supp(q)|e^{4/8}
 \tag{$1+x \leq e^x$}\\ 
&\leq \mu + 2\rho_h(q)^2 |\supp(q)|^2 = \mu + 2\mu^2
\end{align*}
\end{proof}
\end{lemma}
Now we can prove Lemma~\ref{lemma:median_bound}.
\begin{proof}[Proof of Lemma~\ref{lemma:median_bound}]
Recall that $\mathfrak{M}_{j,h} = \frac{1}{\rho_h(q) n_s} \sum_{r = j\cdot n_s+1}^{(j+1)n_s}|\vZ{r}{h}{q}|$. Its expected value is 
$$
\Ex[\mathfrak{M}_{j,h}] = \mu\cdot \rho_h(q)^{-1} = |\supp(q)|.
$$ 
The random sets $\vZ{r}{h}{q}$ are independent for different $r$ so, using Lemma~\ref{lem:variance-bound}, we find
$$
\Va[\mathfrak{M}_{j,h}] = \frac{1}{\rho_h^2(q)n_s^2}\sum_{r = 1}^{n_s} \Va[|\vZ{r}{h}{q}|]\leq \frac{1}{\rho_h(q)^2n_s} \left(  \mu + 2 \mu^2 \right) =   \frac{1}{n_s}\left(\frac{|\supp(q)|}{\rho_h(q)} + 2|\supp(q)|^2\right).
$$
Next, $\mathfrak{M}_{j,h}  \notin |\supp(q)|\cdot e^{\pm \kappa 2^i}$  is subsumed by $|\mathfrak{M}_{j,h}  - |\supp(q)|| >  \kappa 2^{i+1} |\supp(q)|$ because $1- \kappa 2^i  \leq e^{-\kappa 2^i}$ and $ 1 + \kappa 2^{i+1}  \geq e^{\kappa 2^i}$ (because $1+2x \geq e^x$ when $x \leq 1$, and $\kappa 2^i \leq \kappa (n+1)^3 = \varepsilon/4 \leq 1$).
\begin{align*}
	\Pr\left[ \mathfrak{M}_{j,h}  \notin |\supp(q)|e^{\pm \kappa 2^i}\right]
	\leq \Pr\left[\big|\mathfrak{M}_{j,h}  - |\supp(q)|\big| > \kappa 2^{i+1} |\supp(q)| \right]. 
\end{align*}
Next, we use Chebyshev's inequality.
\begin{align*}
	\Pr &\left[\big|\mathfrak{M}_{j,h}  - |\supp(q)|\big| > \kappa 2^{i+1} |\supp(q)| \right] 
	\leq \frac{\Va[\mathfrak{M}_{j,h} ]}{(\kappa 2^{i+1})^2|\supp(q)|^2} \tag{Chebyshev's inequality}
	\\
	&
	\leq \frac{\Va[\mathfrak{M}_{j,h} ]}{\kappa^2 |\supp(q)|^2}  
	\leq \frac{1}{\kappa^2n_s}\left(\frac{1}{\rho_h(q) |\supp(q)|} + 2\right) 
	\\
	&\leq \frac{1}{\kappa^2n_s}\left(\frac{e^{\kappa 2^i}}{16n} + 2\right)
	\tag{$\rho_h \geq \frac{16n e^{-\kappa 2^i}}{\max_{j}|\supp(q_j)|} \geq \frac{16 n  e^{-\kappa 2^i}}{|\supp(q)|}$}
	\\
	&\leq \frac{1}{\kappa^2n_s}\left(\frac{1}{8n} + 2\right) & \tag{$e^{\kappa 2^i} \leq e^{\kappa (n+1)^{3}} = e^{\varepsilon/4} \leq 2$}
	\\
	&\leq \frac{3}{\kappa^2  n_s} 
	\\
	&\leq \frac{1}{4} 
	\tag{$n_s \geq 12/\kappa^2$}
	\\
\end{align*}
Let $X_j$ be the indicator variable taking value~$1$ if and only if $\mathfrak{M}_{j,h} \not\in |\supp(q)|\cdot e^{\pm \kappa 2^i}$ and let $\bar X = \sum_{j = 0}^{n_t - 1} X_j$. We have $\Ex[\bar X] \leq \frac{n_t}{4}$ and the $X_j$ are independent so, using Hoeffding bound,
\[
\Pr\left[\underset{0 \leq j < n_t}{\median}(\mathfrak{M}_{j,h}) \not\in |\supp(q)|\cdot e^{\pm \kappa 2^i}\right] 
= 
\Pr\left[ \bar X > \frac{n_t}{2}\right] 
\leq 
\Pr\left[ \bar X - \Ex[\bar X] \geq \frac{n_t}{4}\right] 
\leq \exp\left(-\frac{n_t}{8}\right)
\]
and $\exp(-n_t/8) \leq \exp(-n|\calP|)$ follows from $n_t = 8n|\calP|$. 
\end{proof} 

\subsection[Probability of $S^r(q)$ growing too large (Lemma~\ref{lemma:proba_S(q)})]{Proof of Lemma~\ref{lemma:proba_S(q)}: Few Samples are Needed.}
Here we focus on the event that some $S^r(q)$ contains more than $\theta$ elements. Observe that $|S^r(q)| \leq |\supp(q)| < \theta$ holds for nodes $q$ with $|\supp(q)| \leq 16n$. These nodes form part of the effective nodes of height $0$. Let $Q \subseteq \calP$ be the set of nodes with $|\supp(q)| \geq 16n$. 
\begin{align*}
\Pr\left[\bigcup_{r \in [n_sn_t]} \bigcup_{q \in \calP} |S^r(q)| \geq \theta\right] &= \Pr\left[\bigcup_{r \in [n_sn_t]} \bigcup_{q \in Q} |S^r(q)| \geq \theta\right] 
\\
&\leq \Pr\left[\bigcup_r \bigcup_{q \in Q} |S^r(q)| \geq \theta \text{ and } \bigcap_{q' \in Q} p(q') \in \Delta(q')\right] + \Pr\left[\bigcup_{q' \in Q} p(q') \not\in \Delta(q')\right]
\\
&\leq \sum_{r}\sum_{q \in Q}  \Pr\left[|S^r(q)| \geq \theta \text{ and } \bigcap_{q' \in Q} p(q') \in \Delta(q')\right] + \Pr\left[\bigcup_{q' \in Q} p(q') \not\in \Delta(q')\right]
\end{align*}
The nodes in $Q$ that have effective height $0$ are those for which $p(q)$ is set to $16n/|\supp(q)|$. So $p(q) \in \Delta(q)$ holds with probability $1$ for $q \in Q^0 \cap Q$, and therefore 
\begin{align*}
\Pr\left[\bigcup_{q' \in Q} p(q') \not\in \Delta(q')\right] =  \Pr\left[\bigcup_{q' \in Q^{> 0}} p(q') \not\in \Delta(q')\right] \leq \frac{1}{16} \tag{Lemma~\ref{lemma:proba_p(q)}}
\end{align*}
\begin{align*}
\Pr\left[|S^r(q)| \geq \theta \text{ and } \bigcap_{q' \in Q} p(q') \in \Delta(q')\right]  &\leq \Pr\left[|S^r(q)| \geq \theta \text{ and } p(q) \in \Delta(q)\right]  
\\ &\leq \Pr\left[\sum_{\alpha \in \supp(q)} \mathbbm{1}(\alpha \in S^r(q)) \cdot \mathbbm{1}(p(q) \in \Delta(q)) \geq \theta \right] \tag{$\theta > 0$}
\\ &\leq \frac{1}{\theta}\cdot \Ex\left[\sum_{\alpha \in \supp(q)} \mathbbm{1}(\alpha \in S^r(q)) \cdot \mathbbm{1}_{\{p(q) \in \Delta(q)\}}\right] \tag{Markov's inequality}
\\
 &\leq \frac{1}{\theta}\sum_{\alpha \in \supp(q)} \Ex\left[\mathbbm{1}(\alpha \in S^r(q)) \cdot \mathbbm{1}(\{p(q) \in \Delta(q))\right]
\\
 &\leq \frac{1}{\theta}\sum_{\alpha \in \supp(q)} \Pr\left[\alpha \in S^r(q)\text{ and } p(q) \in \Delta(q)\right]
\end{align*}
We will show the following.
\begin{lemma}\label{lemma:white_magic_lemma}
For every $q \in Q$, $\alpha \in \supp(q)$ and $r \in [n_sn_t]$, $\Ex\left[\frac{\mathbbm{1}(\alpha \in S^r(q))}{p(q)}\right] = 1$.
\end{lemma}
A consequence of Lemmma~\ref{lemma:white_magic_lemma} is the following.
\begin{lemma}\label{lemma:proba_s_p}
For every $q \in Q$, $\alpha \in \supp(q)$ and $r \in [n_sn_t]$, $\Pr\left[\alpha \in S^r(q)\text{ and } p(q) \in \Delta(q)\right] \leq \frac{32n}{|\supp(q)|}$.
\end{lemma}
\begin{proof}
\begin{align*}
\Ex\left[\mathbbm{1}(\alpha \in S^r(q)) \cdot \mathbbm{1}(p(q) \in \Delta(q))\right]
=  \Ex\left[\mathbbm{1}(\alpha \in S^r(q)) \cdot \mathbbm{1}(p(q) \in \Delta(q)) \cdot \frac{p(q)}{p(q)}\right] 
\end{align*}
Say the effective height of $q$ is $i$, then the random quantity $\mathbbm{1}(p(q) \in \Delta(q)) \cdot p(q)$ is less than $\frac{16\cdot n \cdot e^{\kappa\cdot 2^i}}{|\supp(q)|}$ with probability $1$. Using that $i \leq 3 \lceil \log(n) \rceil$ and that $\kappa \leq \frac{1}{\log(e)(n+1)^3}$, we find that $\mathbbm{1}(p(q) \in \Delta(q)) \cdot p(q)$ is at most $\frac{32\cdot n}{|\supp(q)|}$ with probability $1$. It follows that
\begin{align*}
\Ex\left[\mathbbm{1}(\alpha \in S^r(q)) \cdot \mathbbm{1}(p(q) \in \Delta(q))\right]
\leq  \Ex\left[\frac{\mathbbm{1}(\alpha \in S^r(q))}{p(q)}\right]  \cdot \frac{32\cdot n}{|\supp(q)|} 
\end{align*}
The claim of the lemma then follows from Lemma~\ref{lemma:white_magic_lemma}. 
\end{proof}
Equiped with Lemma~\ref{lemma:proba_s_p}, we use $\theta = 512n_sn_tn|\calP|$ to finish the proof of Lemma~\ref{lemma:proba_S(q)}.
\begin{align*}
\Pr\left[\bigcup_{r \in [n_sn_t]} \bigcup_{q \in \calP} |S^r(q)| \geq \theta\right] \leq \frac{1}{16} + \sum_{r}\sum_{q \in Q} \left(\frac{1}{\theta}\sum_{\alpha \in \supp(q)} \frac{32n}{|\supp(q)|}\right)
\leq \frac{1}{16} + \frac{32 n_sn_tn|\calP|}{\theta} \leq  \frac{1}{8}
\end{align*}

\subsubsection{Proof of Lemma~\ref{lemma:proba_s_p}.}

The proof is by induction on the effective height of $q$. 

\textbf{Base case.} Let $q$ be a node of effective height $0$. Then $S^r(q) = \reduce(\supp(q),p(q))$ and $p(q) = \min(1,16n/|\supp(q)|)$. Here $p(q)$ is a constant and not a random variable so $\Ex[\mathbbm{1}(\alpha \in S^r(q))\cdot p(q)^{-1}] = p(q)^{-1} \cdot \Ex[\mathbbm{1}(\alpha \in S^r(q))] = p(q)^{-1}\cdot\Pr[\alpha \in S^r(q)] = 1$.

\textbf{Inductive case for $\times$ nodes.} Suppose $q$ is a $\times$ node with children $q_1$ and $q_2$. Let $V_1$ and $V_2$ be the finite sets of all possible values for $p(q_1)$ and $p(q_2)$. Since the circuit below $q_1$ and $q_2$ are disjoint, the random variables $p(q_1)$ and $p(q_2)$ are independent and thus
\begin{align*}
\Ex\left[\frac{\mathbbm{1}(\alpha \in S^r(q))}{p(q)}\right] = \sum_{t_1 \in V_1} \sum_{t_2 \in V_2} \Ex\left[\frac{\mathbbm{1}(\alpha \in S^r(q))}{p(q)} \, \bigg| \, p(q_1) = t_1, p(q_2) = t_2\right] \cdot \Pr[p(q_1) = t_1]\cdot \Pr[p(q_2) = t_2]
\end{align*}
Since $p(q) = \round(q,\frac{p(q_1)p(q_2)}{16n})$, the value of $p(q)$ is fully determined by that of $p(q_1)$ and $p(q_2)$. For convenience, we let $f : \mathbbm{R}^2 \rightarrow \mathbbm{R}$ be $f(t_1,t_2) = \round(q,\frac{t_1 t_2}{16n})$. Then the event $[p(q_1) = t_1 \text{ and } p(q_2) =t_2]$ forces the event $p(q) = f(t_1,t_2)$.
\begin{align}\label{eq:white_magic_eq_1}
\Ex\left[\frac{\mathbbm{1}(\alpha \in S^r(q))}{p(q)}\right] 
= \Ex\left[\frac{\mathbbm{1}(\alpha \in S^r(q))}{f(t_1,t_2)} \, \bigg| \, p(q_1) = t_1, p(q_2) = t_2\right] \cdot \Pr[p(q_1) = t_1]\cdot \Pr[p(q_2) = t_2] \notag
\\
= \sum_{t_1 \in V_1} \sum_{t_2 \in V_2} \frac{1}{f(t_1,t_2)} \cdot \Pr[\alpha \in S^r(q) \mid p(q_1) = t_1, p(q_2) = t_2] \cdot \Pr[p(q_1) = t_1]\cdot \Pr[p(q_2) = t_2] 
\end{align}
We have $S^r(q) = \reduce(S^r(q_1) \otimes S^r(q_2), \frac{p(q)}{p(q_1)p(q_2)})$. Let $\alpha_1$ and $\alpha_2$ be the restriction of $\alpha$ to $\var(q_1)$ and $\var(q_2)$, respectively. $\alpha$ can be in $S^r(q)$ only if $\alpha_1$ and $\alpha_2$ are in $S^r(q_1)$ and $S^r(q_2)$, respectively.
\begin{align*}
\Pr[\alpha \in S^r(q) \mid  p(q_1) = t_1, p(q_2) = t_2]  = \Pr[\alpha \in S^r(q) \mid  \alpha_1 \in S^r(q_1), \alpha_2 \in S^r(q_2), p(q_1) = t_1, p(q_2) = t_2]
\\
\cdot \Pr[\alpha_1 \in S^r(q_1), \alpha_2 \in S^r(q_2) \mid p(q_1) = t_1, p(q_2) = t_2]
\\
= \Pr\left[\alpha \in \reduce\left(S^r(q_1) \otimes S^r(q_2), \frac{f(t_1,t_2)}{t_1t_2}\right) \,\Big|\,  \alpha_1 \in S^r(q_1), \alpha_2 \in S^r(q_2), p(q_1) = t_1, p(q_2) = t_2\right] 
\\
\cdot \Pr[\alpha_1 \in S^r(q_1), \alpha_2 \in S^r(q_2) \mid p(q_1) = t_1, p(q_2) = t_2]
\\
= \frac{f(t_1,t_2)}{t_1t_2} \cdot \Pr[\alpha_1 \in S^r(q_1), \alpha_2 \in S^r(q_2) \mid p(q_1) = t_1, p(q_2) = t_2]
\\
= \frac{f(t_1,t_2)}{t_1t_2} \cdot \Pr[\alpha_1 \in S^r(q_1) \mid p(q_1) = t_1] \cdot \Pr[\alpha_2 \in S^r(q_2) \mid p(q_2) = t_2]
\end{align*}
Plugging this in~(\ref{eq:white_magic_eq_1}) yields
\begin{align*}
\Ex\left[\frac{\mathbbm{1}(\alpha \in S^r(q))}{p(q)}\right] &= \prod_{i \in \{1,2\}}\sum_{t_i \in V_i} \Pr[\alpha_i \in S^r(q_i)  \mid p(q_i) = t_i] \cdot \frac{\Pr[p(q_i) = t_i]}{t_i}
\\
&= \prod_{i \in \{1,2\}}\sum_{t_i \in V_i} \Ex\left[\frac{\mathbbm{1}(\alpha_i \in S^r(q_i))}{t_i} \,\bigg| \, p(q_i) = t_i\right] \cdot \Pr[p(q_i) = t_i]
\\
&= \prod_{i \in \{1,2\}}\sum_{t_i \in V_i} \Ex\left[\frac{\mathbbm{1}(\alpha_i \in S^r(q_i))}{p(q_i)} \,\bigg| \, p(q_i) = t_i\right] \cdot \Pr[p(q_i) = t_i]
= \prod_{i \in \{1,2\}} \Ex\left[\frac{\mathbbm{1}(\alpha_i \in S^r(q_i))}{p(q_i)}\right]
\end{align*}
which is exactly $1$ by the induction hypothesis.

\textbf{Inductive case for $+$ nodes.}  Suppose $q$ is a $+$ node with children $q_1,\dots,q_k$. First we show that 
\begin{align}\label{eq:white_magic_eq_2}
\Ex\left[\frac{\mathbbm{1}(\alpha \in S^r(q))}{p(q)}\right] = \Ex\left[\frac{\mathbbm{1}(\alpha \in \hat S^r(q))}{\rho(q)}\right]
\end{align}
Let $V$ be the set of values that can be taken by $p(q)$ and $U$ be the set of values that can be taken by $\rho(q)$.
\begin{align*}
\Ex\left[\frac{\mathbbm{1}(\alpha \in S^r(q))}{p(q)}\right] &= \sum_{t \in V} \sum_{s \in U} \Ex\left[\frac{\mathbbm{1}(\alpha \in S^r(q))}{p(q)} \,\bigg|\, p(q) = t, \rho(q) = s\right]\cdot \Pr[p(q) = t, \rho(q) = s]
\\
&= \sum_{t \in V} \sum_{s \in U} \frac{1}{t}\cdot \Ex[\mathbbm{1}(\alpha \in S^r(q)) \mid p(q) = t, \rho(q) = s]\cdot \Pr[p(q) = t, \rho(q) = s]
\\
&= \sum_{t \in V} \sum_{s \in U} \frac{1}{t}\cdot \Pr[\alpha \in S^r(q) \mid p(q) = t, \rho(q) = s]\cdot \Pr[p(q) = t, \rho(q) = s]
\end{align*}
When $\Pr[p(q) = t, \rho(q) = s] = 0$ the quantity $\Pr[\alpha \in S^r(q)) \mid p(q) = t, \rho(q) = s]$ can be defined in any arbitrary way since they it is multiplied by $0$ anyway. $S^r(q)$ is $\reduce(\hat S^r(q), p(q)/\rho(q))$ and $\alpha \in \hat S^r(q)$ is a necessary condition for $\alpha \in S^r(q)$ to occur. Thus
\begin{align*}
\Pr[\alpha \in S^r(q) \mid p(q) = t, \hat\rho(q) = s] = \frac{t}{s}\cdot \Pr[\alpha \in \hat S^r(q) \mid p(q) = t, \rho(q) = s] 
\end{align*}
Thus
\begin{align*}
\Ex\left[\frac{\mathbbm{1}(\alpha \in S^r(q))}{p(q)}\right] 
&= \sum_{t \in V} \sum_{s \in U} \frac{1}{s}\cdot \Pr[\alpha \in \hat S^r(q) \mid p(q) = t, \rho(q) = s] \cdot \Pr[p(q) = t, \rho(q) = s]
\\
&= \sum_{s \in U} \frac{1}{s} \sum_{t \in V} \cdot \Pr[\alpha \in \hat S^r(q), p(q) = t, \rho(q) = s]
= \sum_{s \in U} \frac{1}{s}\cdot \Pr[\alpha \in \hat S^r(q), \rho(q) = s]
\\
&= \sum_{s \in U} \frac{1}{s}\cdot \Ex[\mathbbm{1}(\alpha \in \hat S^r(q))\mathbbm{1}(\rho(q) = s)]
= \Ex\left[\frac{\mathbbm{1}(\alpha \in \hat S^r(q))}{\rho(q)}\right]
\end{align*}
Now, let $q_j$ be the unique child of $q$ such that $\alpha \in \hat S^r(q)$ only if $\alpha \in S^r(q_j)$. We are done if we can prove
\begin{align}\label{eq:white_magic_eq_2}
\Ex\left[\frac{\mathbbm{1}(\alpha \in \hat S^r(q))}{\rho(q)}\right] = \Ex\left[\frac{\mathbbm{1}(\alpha \in S^r(q_j))}{p(q_j)}\right]
\end{align}
the induction hypothesis will then show that $
\Ex\left[\frac{\mathbbm{1}(\alpha \in S^r(q))}{p(q)}\right] = 1$, as desired. Let $W$ be the set of values that can be taken by $p(q_j)$.
\begin{align*}
\Ex\left[\frac{\mathbbm{1}(\alpha \in \hat S^r(q))}{\rho(q)}\right] 
= \sum_{t \in W} \sum_{s \in U} \frac{1}{s}\cdot \Pr[\alpha \in \hat S^r(q) \mid p(q_j) = t, \rho(q) = s]\cdot \Pr[p(q_j) = t, \rho(q) = s]
\end{align*}
$\alpha \in \hat S^r(q)$ occurs if and only if $\alpha \in S^r(q,q_j) = \reduce(S^r(q_j), \rho/p(q_j))$. Thus
\begin{align*}
\Pr[\alpha \in \hat S^r(q) \mid p(q_j) = t, \rho(q) = s] = 
\frac{s}{t}\cdot \Pr[\alpha \in S^r(q_j) \mid p(q_j) = t, \rho(q) = s]
\end{align*}
It follows that
\begin{align*}
\Ex\left[\frac{\mathbbm{1}(\alpha \in \hat S^r(q))}{\rho(q)}\right] 
= \sum_{t \in W} \frac{1}{t} \sum_{s \in U} \Pr[\alpha \in S^r(q_j) \mid p(q_j) = t, \rho(q) = s]\cdot \Pr[p(q_j) = t, \rho(q) = s]
\\
= \sum_{t \in W} \frac{1}{t} \cdot \Pr[\alpha \in \hat S^r(q_j) \mid p(q_j) = t]
= \sum_{t \in W} \frac{1}{t}\cdot \Ex[\mathbbm{1}(\alpha \in \hat S^r(q_j))\mathbbm{1}(p(q_j) = t)]
= \Ex\left[\frac{\mathbbm{1}(\alpha \in \hat S^r(q_j))}{p(q_j)}\right]
\end{align*}
\subsection*{Acknowledgements.}

Meel acknowledges the support of the Natural Sciences and Engineering Research Council of Canada (NSERC), [funding reference number RGPIN-2024-05956]; de Colnet is supported by the Austrian Science Fund (FWF),
ESPRIT project FWF ESP 235. This work was done in part while de Colnet was visiting the University of Toronto and Georgia Institute of Technology. This research was initiated at Dagstuhl Seminar 24171 on ``Automated Synthesis: Functional, Reactive and Beyond'' (\url{https://www.dagstuhl.de/24171}). We gratefully acknowledge the Schloss Dagstuhl - Leibniz Center for Informatics for providing an excellent environment and support for scientific collaboration.  We are grateful to Marcelo Arenas, Alberto Croquevielle, Umang Mathur, and Cristian Riveros  for many conversations that have shaped the final version of this work. 

\newpage

\clearpage 
\appendix
\section{Appendix}

\addtocontents{toc}{\protect\setcounter{tocdepth}{2}}

\subsection{Proof of Proposition~\ref{prop:cnf-to-union-concat}}\label{sec:cnf-to-union}

\cnfunion*

\begin{proof}
	The core idea is to construct a program that directly mirrors the structure of derivations in a Chomsky Normal Form (CNF)  grammar. The construction is a dynamic programming algorithm, building up the language of derivable strings of length $i$ from the languages of strings of lengths less than $i$.
	
	We construct a program node, let's call it $q_{A,i}$, for each non-terminal $A \in V$ and each length $i \in \{1, \dots, n\}$. The goal is for each node $q_{A,i}$ to represent the language $L_i(A)$, which is the set of all words of length exactly $i$ that can be derived from $A$.
	
	The construction proceeds by mapping the two types of CNF production rules to the operations of the $(\cup, \cdot)$ algebra:
	
	\begin{enumerate}
		\item \textbf{Terminal Productions ($A \to a$):} A string of length 1 can be derived from a non-terminal $A$ if there is a rule of the form $A \to a$. Since there may be multiple such rules for a given $A$ (e.g., $A \to a$ and $A \to b$), the total set of length-1 strings derivable from $A$ is the \emph{union} of the outcomes of all its applicable terminal productions. We therefore define the node $q_{A,1}$ as the union of all terminals $a$ for which the rule $A \to a$ exists.
		
		\item \textbf{Non-terminal Productions ($A \to BC$):} Deriving a string of length $i > 1$ from $A$ involves both concatenation and union operations, stemming from different sources of choice in the derivation.
		\begin{itemize}
			\item \textbf{Concatenation:} Any single rule $A \to BC$ implies that a word derivable from $A$ is formed by the \emph{concatenation} of a word from $B$ and a word from $C$. This maps directly to the program's concatenation ($\cdot$) operation.
			
			\item \textbf{Union:} The union operation ($\cup$) arises from two distinct choices in the derivation process:
			\begin{enumerate}
				\item \textbf{Choice of Length Split:} For a single rule $A \to BC$, a word of length $i$ can be formed by concatenating a word of length $k$ from $B$ and a word of length $i-k$ from $C$. Since this split can occur in multiple ways (i.e., for any $k$ from $1$ to $i-1$), the total language generated by this single rule is the \emph{union} over all possible length splits.
				\item \textbf{Choice of Production Rule:} A given non-terminal $A$ can typically be the left-hand side of multiple production rules (e.g., $A \to B_1C_1$, $A \to B_2C_2$, etc.). The complete language $L_i(A)$ is the set of words derivable from \emph{any} of these rules. Therefore, we must take the \emph{union} of the languages generated by each individual rule.
			\end{enumerate}
		\end{itemize}
	\end{enumerate}
	
	Combining these observations, our program constructs each node $q_{A,i}$ (for $i>1$) by implementing a nested union: an {\em outer} union over all production rules with $A$ on the left-hand side, and an {\em inner} union over all possible length splits for each rule's concatenated sub-problems.
	
	Since this bottom-up construction perfectly mimics the definition of a derivation in CNF for a fixed length, the node $q_{A,i}$ will represent the language $L_i(A)$. The final output of the program is the node $q_{S,n}$, which represents the set of words of length $n$ derivable from the start symbol $S$, which is exactly $L_n(G)$.
	
	(The special case for $n=0$ is handled separately: if $S \to \varepsilon$ is a rule, $L_0(G) = \{\varepsilon\}$; otherwise, it is $\emptyset$. The program can represent these trivial languages directly.)
\end{proof}

\begin{proposition}
Suppose there is an homogeneous $(\cup ,\cdot)$ program $\calP$ computing a language $L \subseteq \Sigma^n$. Then there is a multilinear homogeneous $(+,\times)$ program  $\calP'$ of size $n|\calP|$ and degree $n$ computing a polynomial $p$ over $n|\Sigma|$ variables such that $|\supp(p)| = |L|$. In addition, given $\calP$ the program $\calP'$ is constructed in time $O(n|\calP|)$.
\end{proposition}
\begin{proof}
We reproduce the proof of Lemma 3.5 from~\cite{GoreJKSM97}; this lemma partially subsumes our proposition, it only misses a proof that the $(+,\times)$ program is multilinear. Let $\calP =(u_0, \dots,u_{|\calP| - 1})$. We define $\deg(u_i)$ inductively: if $u_i$ is a primitive element, i.e., a letter in $\Sigma$ then $\deg(u_i) = 1$; if $u_i = u_j \cup u_k$ then $\deg(u_i) = \max(\deg(u_j),\deg(u_k))$; and if $u_i = u_j \cdot u_k$ then $\deg(u_i) = \deg(u_j) + \deg(u_k)$. Since $\calP$ is homogeneous, $\deg(u_i)$ is the unique length of the words in the language represented by $u_i$. We consider $n|\Sigma|$ variables $x_{\sigma,r}$ for $\sigma \in \Sigma$ and $0 \leq r \leq n-1$. For each $0 \leq i \leq |\calP| - 1$ and $0 \leq r \leq n - \deg(u_i)$, define
$v^{(r)}_i$ as follows:
\begin{itemize}
\item[•] if $u_i = \sigma \in \Sigma$ then $v^{(r)}_i = x_{\sigma,r}$;
\item[•] if $u_i = u_j \cup u_k$ then $v^{(r)}_i = v^{(r)}_j + v^{(r)}_k$; 
\item[•] if $u_i = u_j \cdot u_k$ then $v^{(r)}_i = v^{(r)}_j \times v^{(r + \deg(u_j))}_k$. 
\end{itemize}
$\calP'$ is the $(+,\times)$ program obtained by arranging the polynomials $v^{(r)}_i$ in lexicographic order of the pairs $(i,r)$. ~\cite[Lemma 3.5]{GoreJKSM97} shows that $\supp(v_i^{(r)}) = \supp(\pi^{(r)}(u_i))$ where $\pi^{(r)}(a_0,\dots,a_{\deg(u_i)-1}) = \prod_{t = 0}^{\deg(u_i)-1} x_{a_t,r+t}$ and $\pi^{(r)}(L') = \sum_{w \in L'} \pi^{(r)}(w)$ for any language $L'$. It follows that $|\supp(\calP')| = |\supp(\pi^{(0)}(L))| = |L|$ (where $L$ is the language computed by $\calP$).

We claim that $\calP'$ is multilinear and homogeneous. A simple induction shows that $\deg(v^{(r)}_i) = \deg(u_i)$ for every $i$ and $r$ so the homogeneity of the $\calP$ is passed on to $\calP'$. To prove multilinearity, it suffices to show the following.
\begin{itemize}
\item[] \textit{For all $i$ and $r$, $x_{\sigma,t} \in \var(v^{(r)}_i)$ only if $r \leq t \leq r +\deg(u_i) - 1$.}
\end{itemize} 
The proof is by induction on $i$. If $u_i = \sigma$ then $\var(v^{(r)}_i)  = \{x_{\sigma,r}\}$ and the claim is immediately true; this takes care of the base case $i = 1$. Next, if  $u_i = u_j \cup u_k$ then $x_{\sigma,t} \in \var(v^{(r)}_i)$ if and only if $x_{\sigma,t} \in \var(v^{(r)}_j) \cup \var( v^{(r)}_k)$ which, by induction, is only $r \leq t \leq r + \max(\deg(u_j),\deg(u_k)) - 1 = r + \deg(u_i) - 1$. If $u_i = u_j \cdot u_k$ then $x_{\sigma,t} \in \var(v^{(r)}_i)$ if and only if $x_{\sigma,t} \in \var(v^{(r)}_j) \cup \var(v^{(r + \deg(u_j))}_k)$ which, by induction, is only if $r \leq t \leq r + \max(\deg(u_j),\deg(u_j) + \deg(u_k)) - 1 = r + \deg(u_i) - 1$.
\end{proof}

\clearpage 

\subsection{Proof of Lemma~\ref{lemma:not_so_black_magic}}\label{sec:proof-blackmagic} 

\notSoBlackMagic*
\begin{proof}
For $A$ an algorithm and $X_1,\dots,X_k$ random variables in $A$, with $\Omega_i$ the universe of $X_i$, we denote by
$$
\Pr_A\left[\bigcap_{l = 1}^k X_l = \omega_l\right]
$$
the probability that, after running $A$ (which we assume terminates), we have $X_i = \omega_i$ for every $i$ (with $\omega_i \in \Omega_i$). We rename $\approxMCDNNFcore^*$ $A_1$. We make a sequence of modifications to $A_1$ obtaining algorithms $A_2,A_3,\dots$. For $X_1,\dots,X_k$ random variables in $A_i$ and $Y_1,\dots,Y_k$ are random variables in $A_j$ such that $X_i$ and $Y_i$ have the same universe $\Omega_i$, we write
$$
(X_1,\dots,X_k)_{A_i} \sim (Y_1,\dots,Y_k)_{A_j}
$$
to say that, for every $(\omega_1,\dots,\omega_k) \in \Omega_1 \times \dots \times \Omega_k$, we have
$$
\Pr_{A_i}\left[\bigcap_{l = 1}^k X_l = \omega_l\right] =  \Pr_{A_j}\left[\bigcap_{l = 1}^k Y_l = \omega_l\right]
$$
We may have variables that have the same name in $A_i$ and $A_j$, say $X_1,\dots,X_k$. Then $
(X_1,\dots,X_k)_{A_i} \sim (X_1,\dots,X_k)_{A_j}
$ means that the lefthand side variables are considered in $A_i$ and the righthand side variables are considered in $A_j$.

The algorithm modifications are all on done in $\OrEstimateSample$ and $\AndEstimateSample$ at the same time. Both $\OrEstimateSample$ and $\AndEstimateSample$ are as follow in $A_1$ (for $\AndEstimateSample$, $\hat{S}^r(q)$ is technically not defined and is not required in the computation of $p(q)$, so we can assume $\hat{S}^r(q) = \emptyset$).

\begin{algorithm}[H]\caption*{$\mathsf{estimateSample}$ in $A_1$}
compute $\hat{S}^r(q)$ for all $r$ using $\{ S^r(q') \mid q' \in \children(q)\}$
\\
compute $p(q)$ using $\{\hat{S}^r(q)\}_r$
\\
compute $S^r(q)$ for all $r$ using $p(q)$ and $\hat{S}^r(q)$
\end{algorithm}

For every state $q$, we keep the history for the descendants of $q$. Formally, we maintain a variable $H(q)$ for every state $q$. Initially $H(q)$ is empty for all $q$. After $p(q)$ is computed, $H(q')$ is updated for all descendants $q'$ of $q$ as follow: $H(q') = H(q') \cup (p \mapsto p(q))$. Thus, when we reach $\mathsf{estimateSample}(q')$, $H(q')$ is the history for $q'$. Clearly, the variables $H(q)$ are unused for the computation of the other random variables in $A_1$.
 
\begin{algorithm}[H]\caption*{$\mathsf{estimateSample}$ in $A_1$} 
compute $\hat{S}^r(q)$ for all $r$ using $\{ S^r(q') \mid q' \in \children(q)\}$
\\
compute $p(q)$ using $\{\hat{S}^r(q)\}_r$ \textcolor{red}{and update $H$}
\\
compute $S^r(q)$ for all $r$ using $p(q)$ and $\hat{S}^r(q)$
\end{algorithm}

In $A_2$, for every realizable history $h$ for $q$, and every $r \in [\nsnt]$ and for every $v$ that is a possible candidate for $p(q)$, we have sets $\hat{S}^r_h(q)$ and $S^r_{h,v}(q)$. Initially these new sets are empty. When the $A_2$ computes the value for $p(q)$, $H(q)$ has already been set. Once $\hat{S}^r(q)$ is computed, we copy its content in $\hat{S}^r_{H(q)}(q)$ and, once $S^r(q)$ is computed, we copy its content in $S^r_{H(q),p(q)}(q)$. $A_1$ and $A_2$ construct the random variables $p(q)$, $S^r(q)$, $\hat S^r(q)$ and $H(q)$ in the same way so
$$
(H(q),p(q),(S^r(q),\hat S^r(q))_{r \in [\alpha]})_{A_1} \sim 
(H(q),p(q),(S^r(q),\hat S^r(q))_{r \in [\alpha]})_{A_2}
$$
We also have that 
$$
(H(q),p(q),(S^r(q),\hat S^r(q))_{r \in [\alpha]})_{A_2}
\sim
(H(q),p(q),(S^r_{H(q),p(q)}(q),\hat S^r_{H(q)}(q))_{r \in [\alpha]})_{A_2}
$$
\begin{algorithm}[H]\caption*{$\mathsf{estimateSample}(q)$ in $A_2$}
compute $\hat{S}^r(q)$ for all $r$ using $\{ S^r(q') \mid q' \in \children(q)\}$
\\
\textcolor{red}{copy $\hat{S}^r(q)$ to $\hat{S}^r_{H(q)}(q)$ for all $r$}
\\
compute $p(q)$ using $\{\hat{S}^r(q)\}_r$ and update $H$
\\
compute $S^r(q)$ for all $r$ using $p(q)$ and $\hat{S}^r(q)$
\\
\textcolor{red}{copy $S^r(q)$ to $S^r_{H(q),p(q)}(q)$ for all $r$}
\end{algorithm}

We are going to dedfine $A_3,A_4,A_5,A_6$ and $A_7$ such that, for every $3 \leq i \leq 7$,
\begin{equation}\label{eq:same_distrib}
(H(q),p(q),(S^r_{H(q),p(q)}(q),\hat S^r_{H(q)}(q))_{r \in [\alpha]})_{A_{i-1}}
\sim
(H(q),p(q),(S^r_{H(q),p(q)}(q),\hat S^r_{H(q)}(q))_{r \in [\alpha]})_{A_i}
\end{equation}

For a fixed $q$, to compute $p(q)$, $\{\hat S^r(q)\}_r$, $A_2$ uses $\{p(q')\}_{q' \in \children(q)}$ and $\{\{S^r(q')\}_r\}_{q' \in \children(q)}$. But at that point, $S^r_{H(q'),p(q')}(q')$ is indentical to $S^r(q')$ so we can swap them. Similarly, to compute $\{S^r(q)\}_r$, the algorithm uses $p(q)$ and $\{\hat S^r(q)\}_r$. But at that point, $\hat S^r(q)$ and $\hat S^r_{H(q)}(q)$ are indentical.$A_3$ instead does the computation with the $S^r_{H(q'),p(q')}(q')$ and the $\hat S^r_{H(q)}(q)$
and then copies the result in $\hat S^r(q)$ and $S^r(q)$. Since the swapped sets are identical, Eq~(\ref{eq:same_distrib}) holds for $i = 3$.

\begin{algorithm}[H]
\caption*{$\mathsf{estimateSample}(q)$ in $A_3$}
compute $\hat{S}^r(q)$ for all $r$ using $\{\textcolor{red}{S^r_{H(q'),p(q')}(q')} \mid q' \in \children(q)\}$
\\
copy $\hat{S}^r(q)$ to $\hat{S}^r_{H(q)}(q)$ for all $r$
\\
compute $p(q)$ using $\{\textcolor{red}{\hat{S}^r_{H(q)}(q)}\}_r$ and update $H$
\\
compute $S^r(q)$ for all $r$ using $p(q)$ and \textcolor{red}{$\hat{S}_{H(q)}^r(q)$}
\\
copy $S^r(q)$ to $S^r_{H(q),p(q)}(q)$ for all $r$
\end{algorithm}

\noindent In $A_4$ we swap $\hat{S}^r_{H(q)}(q)$ with $\hat S^r(q)$ and $S^r_{H(q),p(q)}(q)$ and $S^r(q)$. That is, first compute $\hat{S}^r_{H(q)}(q)$ (resp. $S^r_{H(q),p(q)}(q)$) and then copy its content to $\hat{S}^r(q)$ (resp. $S^r(q)$). Again, compared to $A_3$ we are just swapping the roles of sets that are anyway indentical, so Eq~\ref{eq:same_distrib} holds for $i = 4$.

\begin{algorithm}[H]\caption*{$\mathsf{estimateSample}(q)$ in $A_4$}
compute \textcolor{red}{$\hat{S}^r_{H(q)}(q)$} for all $r$ using $\{S^r_{H(q'),p(q')}(q') \mid q' \in \children(q)\}$
\\
compute $p(q)$ using $\{\hat{S}^r_{H(q)}(q)\}_r$ and update $H$
\\
compute \textcolor{red}{$S^r_{H(q),p(q)}(q)$} for all $r$ using $\hat{S}^r_{H(q)}(q)$
\\
copy \textcolor{red}{$S^r_{H(q),p(q)}(q)$ to $S^r(q)$} for all $r$
\\
copy \textcolor{red}{$\hat{S}^r_{H(q)}(q)$ to $\hat{S}^r(q)$} for all $r$
\end{algorithm}

$A_4$ does not touch any sets $\hat{S}^r_{h}(q)$ or $S^r_{v,h}(q)$ when $h \neq H(q)$ and $v \neq p(q)$. $A_5$ computes $\hat{S}^r_{H(q)}(q)$ and $S^r_{p(q),H(q)}(q)$ like $A_4$ -- which ensures Eq~(\ref{eq:same_distrib}) holds for $i = 5$ -- but also computes all other $\hat{S}^r_{h}(q)$ and $S^r_{h,t}(q)$. Say $\children(q) = (q_1,\dots,q_k)$ then $A_5$ sets $m_h(q) = \max(h(q_1),\dots,h(q_k)$ and $\hat{S}^r_{h}(q) = \union(q,\reduce(S^r_{h_1,h(q_1)}(q_1),h(q_1)/m_h(q)),\dots,\reduce(S^r_{h_k,h(q_k)}(q_k),h(q_k)/m_h(q)))$ and $S^r_{h,v}(q) = \reduce(\hat{S}^r_{h}(q),m_h(q)/v)$, where $h_i$ denote the restriction of $h$ to the descendants of $q_i$.

\begin{algorithm}[H]\caption*{$\mathsf{estimateSample}(q)$ in $A_5$}
compute $\hat{S}^r_{H(q)}(q)$ for all $r$ using $\{ S^r_{H(q'),p(q')}(q') \mid q' \in \children(q)\}$
\\
\textcolor{red}{compute $\hat{S}^r_{h}(q)$ for all $r$ and $h$ using $\{ \{S^r_{h,v}(q') \}_{r,h,v} \mid q' \in \children(q)\}$}
\\
compute $p(q)$ using $\{\hat{S}^r_{H(q)}(q)\}_r$ and update $H$
\\
compute $S^r_{H(q),p(q)}(q)$ for all $r$ using $\hat{S}^r_{H(q)}(q)$
\\
\textcolor{red}{compute $S^r_{h,v}(q)$ for all $r$ and all $(h,v) \neq (H(q),p(q))$ using $\{\hat{S}^r_{h}(q)\}_h$}
\\
copy $S^r_{H(q),p(q)}(q)$ to $S^r(q)$ for all $r$
\\
copy $\hat{S}^r_{H(q)}(q)$ to $\hat{S}^r(q)$ for all $r$
\end{algorithm}

\noindent But then, in $A_5$, the variables $\hat{S}^r_{H(q)}(q)$ and $S^{r}_{H(q),p(q)}(q)$ are computed like any other $\hat S^r_h$ and $S^r_{h,v}$. So we define $A_6$ that directly computes all variables uniformly and have that Eq~(\ref{eq:same_distrib}) holds for $i = 6$

\begin{algorithm}[H]\caption*{$\mathsf{estimateSample}(q)$ in $A_6$}
compute $\hat{S}^r_{h}(q)$ for all $r$ \textcolor{red}{and all $h$} using $\{ S^r_{h',p(q')}(q')\}_{h',q' \in \children(q)}$
\\
compute $p(q)$ using $\{\hat{S}^r_{H(q)}(q)\}_r$ and update $H$
\\
compute $S^r_{h,v}(q)$ for all $r$, \textcolor{red}{and all $h$ and $v$} using $\{\hat{S}^r_{h}(q)\}_h$
\\
copy $S^r_{H(q),p(q)}(q)$ to $S^r(q)$ for all $r$
\\
copy $\hat{S}^r_{H(q)}(q)$ to $\hat{S}^r(q)$ for all $r$
\end{algorithm}

\noindent In $A_6$, line 3 does not depend on line 2 so we can swap them, thus obtaining $A_7$. Clearly, Eq~(\ref{eq:same_distrib}) holds for $i = 7$.

\begin{algorithm}[H]\caption*{$\mathsf{estimateSample}(q)$ in $A_7$}
compute $\hat{S}^r_{h}(q)$ for all $r$ and all $h$ using $\{ S^r_{h',p(q')}(q')\}_{h',q' \in \children(q)}$
\\
compute $S^r_{h,v}(q)$ for all $r$, $h$ and $v$ using $\{\hat{S}^r_{h}(q)\}_h$
\\
compute $p(q)$ using $\{\hat{S}^r_{H(q)}(q)\}_r$ and update $H$
\\
copy $S^r_{H(q),p(q)}(q)$ to $S^r(q)$ for all $r$
\\
copy $\hat{S}^r_{H(q)}(q)$ to $\hat{S}^r(q)$ for all $r$
\end{algorithm}

Finally we let $A_8$ be the same as $A_7$ without line 3, 4 and 5. Since lines 1 and 2 (for $q$) do not depend on the execution of lines 3,4,5 (for descendants of $q$) we have that
$$
((S^r(q),\hat S^r(q))_{r \in [\alpha]})_{A_7}
\sim
((S^r_{H(q),p(q)}(q),\hat S^r_{H(q)}(q))_{r \in [\alpha]})_{A_8}
$$
\begin{algorithm}[H]\caption*{$\mathsf{estimateSample}(q)$ in $A_8$}
compute $\hat{S}^r_{h}(q)$ for all $r$ and all $h$ using $\{ S^r_{h',p(q')}(q')\}_{h',q' \in \children(q)}$
\\
compute $S^r_{h,v}(q)$ for all $r$, $h$ and $v$ using $\{\hat{S}^r_{h}(q)\}_h$
\end{algorithm}

We see that $A_8$ is exactly the random process, so
$
\Pr_{A_8}[e(\hat S^1_h(q),\dots,\hat S^\nsnt_h(q))] = \Pr[e(\vZ{1}{h}{q},\dots,\vZ{\nsnt}{h}{q})]
$
and
$
\Pr_{A_8}[e(S^1_{h,v}(q),\dots,S^\nsnt_{h,v}(q))] = \Pr[e(\vY{1}{h}{v}{q},\dots,\vZ{\nsnt}{h}{v}{q})]
$.
Now, we have shown that 
$$
(H(q),p(q),(S^r(q),\hat S^r(q))_{r \in [\alpha]})_{A_1}
\sim
(H(q),p(q),(S^r_{H(q),p(q)}(q),\hat S^r_{H(q)}(q))_{r \in [\alpha]})_{A_7}
$$
It follows that 
\begin{align*}
(H(q),p(q),(S^r(q))_{r \in [\alpha]})_{A_1}
&\sim
(H(q),p(q),(S^r_{H(q),p(q)}(q))_{r \in [\alpha]})_{A_7}
\\
(H(q),(\hat S^r(q))_{r \in [\alpha]})_{A_1}
&\sim
(H(q),(\hat S^r_{H(q)}(q))_{r \in [\alpha]})_{A_7}
\end{align*}
Therefore 
\begin{align*}
\Pr_{A_1}[H(q) = h \text{ and } p(q) = v \text{ and } e(\hat S^1(q),\dots,\hat S^\nsnt(q))] 
&= 
\Pr_{A_7}[H(q) = h \text{ and } p(q) = v \text{ and } e(S^1_{h,v}(q),\dots,S^\nsnt_{h,v}(q))] 
\\
&\leq \Pr_{A_7}[e(S^1_{h,v}(q),\dots,S^\nsnt_{h,v}(q))] 
\\
&= \Pr_{A_8}[e(S^1_{h,v}(q),\dots,S^\nsnt_{h,v}(q))] 
\\
&= \Pr[e(\vY{1}{h}{v}{q},\dots,\vY{\nsnt}{h}{v}{q})]
\end{align*}
and
\begin{align*}
\Pr_{A_1}[H(q) = h \text{ and } e(\hat S^1(q),\dots,\hat S^\nsnt(q))] &= 
\Pr_{A_7}[H(q) = h \text{ and } e(\hat S^1_h(q),\dots,\hat S^\nsnt_h(q))] 
\\
&\leq \Pr_{A_7}[e(\hat S^1_h(q),\dots,\hat S^\nsnt_h(q))] 
\\
&= \Pr_{A_8}[e(\hat S^1_h(q),\dots,\hat S^\nsnt_h(q))] 
\\
&= \Pr[e(\vZ{1}{h}{q},\dots,\vZ{\nsnt}{h}{q})]
\end{align*}
\end{proof}

\subsection{Proof of Lemma~\ref{lemma:proba_first_order}}
{\probaFirstOrder*}
\begin{proof}
We proceed by induction. For nodes $q$ of effective weight $0$ the statement of the lemma is immediate by definition of $\vY{r}{h}{t}{q}$. Now consider $q$ of effective height strictly greater than $0$ and suppose the statement of the lemma holds for all $q'$ of smaller effective height.
\begin{enumerate}[leftmargin=*]
\item[•] first case: $q = q_1 \times q_2$. Let $h = h_1 \cup h_2 \cup (q_1 \mapsto t_1, q_2 \mapsto t_2)$, with $h_1$ and $h_2$ two compatible and realizable histories for $q_1$ and $q_2$, respectively, and $t \leq t_1t_2$. By definition, $\vY{r}{h}{t}{q} =
\reduce(\vY{r}{h_1}{t_1}{q_1} \otimes \vY{r}{h_2}{t_2}{q_2},\frac{t}{t_1t_2}).
$
Let $\alpha_1$ be the restriction of $\alpha$ to $\var(q_1)$ and $\alpha_2$ be the restriction of $\alpha$ to $\var(q_2)$.
\begin{align*}
\Pr&\left[\alpha \in \vY{r}{h}{t}{q}\right] 
\\
&= \Pr\left[\alpha \in \vY{r}{h}{t}{q} \mid \alpha_1 \in \vY{r}{h_1}{t_1}{q_1}, \alpha_2 \in \vY{r}{h_2}{t_2}{q_2}\right]\Pr\left[\alpha_1 \in \vY{r}{h_1}{t_1}{q_1} \text{ and } \alpha_2 \in \vY{r}{h_2}{t_2}{q_2}\right]
\\
&= \frac{t}{t_1t_2}\Pr\left[\alpha_1 \in \vY{r}{h_1}{t_1}{q_1} \text{ and } \alpha_2 \in \vY{r}{h_2}{t_2}{q_2}\right] \tag{$\reduce$}
\\
&= \frac{t}{t_1t_2}\Pr\left[\alpha_1 \in \vY{r}{h_1}{t_1}{q_1}\right]\Pr\left[\alpha_2 \in \vY{r}{h_2}{t_2}{q_2}\right] \tag{Fact~\ref{fact:independence}}
\\
&= t \tag{induction}
\end{align*}
\item[•] Second case: $q = q_1 + \dots + q_k$ with $\prev(q) = (q_1,\dots,q_k)$. Let $h = h_1 \cup \dots \cup h_k \cup (q_1 \mapsto t_1,\dots,q_k \mapsto t_k)$ where the histories $h_j$ for the $q_j$ are realizable and pairwise compatible. Let $t \leq t_{\min} = \min(t_1,\dots,t_k)$. By definition, $\vY{r}{h}{t}{q}
=
\reduce(\vZ{r}{h}{q},\frac{t}{t_{\min}})
$
where 
$$
\vZ{r}{h}{q} = 
	\union(
		\reduce(\vY{r}{h_1}{t_1}{q_1},\frac{t_{\min}}{t_1}),
		\reduce(\vY{r}{h_2}{t_2}{q_2},\frac{t_{\min}}{t_2}),
		\dots,
		\reduce(\vY{r}{h_k}{t_k}{q_k},\frac{t_{\min}}{t_k})
	).
$$
There is a unique smallest $j$ such that $\alpha \in \supp(q_j) \setminus (\supp(q_1) \cup \dots \cup \supp(q_{j-1}))$ holds. So 
$$
\Pr\left[\alpha\in \vZ{r}{h}{q}\right] 
= \Pr\left[\alpha \in \reduce\left(\vY{r}{h_j}{t_j}{q_j},\frac{t_{\min}}{t_j}\right)\right] 
= \frac{t_{\min}}{t_j}\Pr\left[\alpha \in \vY{r}{h_j}{t_j}{q_j}\right]
= t_{\min}
$$ 
where the last equality uses the induction hypothesis. It follows that 
$$
\Pr\left[\alpha\in \vY{r}{h}{t}{q}\right] 
= \Pr\left[\alpha \in \reduce\left(\vZ{r}{h}{q},\frac{t}{t_{\min}}\right)\right] 
= \frac{t}{t_{\min}}\Pr\left[\alpha\in \vZ{r}{h}{q}\right] 
= t.
$$
\end{enumerate}
\end{proof}

\subsection{Proof of Lemma~\ref{lemma:proba_second_order}}
\probaSecondOrder*

In the following we consider a single copy of the $n_sn_t$ random processes. So the superscript $r$ is dropped from notations. We call \emph{contain-event} for $q$ any event of the form $\alpha \in \vY{}{h}{t}{q}$. We call \emph{reduce-event} any event of the form $\alpha \in \reduce(\calZ,t)$ for $\calZ$ a (potentially random) set and $t$ a number between $0$ and $1$. All contain-events $E$ can be written as reduce-events, either of the form $\alpha \in \reduce(\supp(q),t)$ when $q$ is in $Q^0$, or of the form $\alpha \in \reduce(\vY{}{h_1}{t_1}{q_1} \otimes \vY{}{h_2}{t_2}{q_2},p)$ when $q = q_1 \times q_2$, or of the form $\alpha \in \reduce(\vY{}{h'}{t'}{q'},p')$ when $q$ is a $+$ node. The last claim takes some explaining. Suppose $q$ is a $+$ node with children $q_1,\dots,q_k$ and that $q \not\in Q^0$. The variable $\vY{}{h}{t}{q}$ is defined as $\reduce(\vZ{}{h}{q},p)$ for some probability $p$. Since $\alpha \in \vZ{}{h}{q}$ occurs if and only if $\alpha \in \reduce(\vY{}{h_i}{t_i}{q_i},p_i)$ for some child $q_i$ of $q$ and some unique probability $p_i$, the event $\alpha \in \vY{}{h}{t}{q}$ can be rewritten as the reduce event $\alpha \in \reduce(\reduce(\vY{}{h_i}{t_i}{q_i},p_i),p) = \reduce(\vY{}{h_i}{t_i}{q_i},p_ip)$. Now we define the \emph{premises} of a reduce-event.
\begin{definition}\text{}
\begin{itemize} 
\item The premises of $\alpha \in \reduce(\vY{}{h_1}{t_1}{q_1} \otimes \vY{}{h_2}{t_2}{q_2},p)$ are the two events $\alpha_1 \in \vY{}{h_1}{t_1}{q_1}$ and $\alpha_2 \in \vY{}{h_2}{t_2}{q_2}$, where $\alpha_i$ denotes the restriction of $\alpha$ to $var(q_i)$
\item The premises of $\alpha \in \reduce(\vY{}{h}{t}{q},p)$ consist only of the event $\alpha \in \vY{}{h}{t}{q}$. 
\item  The premises of $\alpha \in \reduce(\supp(q),p)$ consist only of the event $\alpha \in \supp(q)$.
\end{itemize}
\end{definition}
Note that the premises of $\alpha \in \reduce(\supp(q),p)$ occur with probability one as we assume that $\alpha$ is in $\supp(q)$. Since all contain-events are reduce-events, the \emph{premises} of a contain-event are well-defined. For $E$ a contain-event for $q$, $\premises(E)$ is the set of its premises. $E$ occurs only if its premises occur; formally,
$$
\Pr[E] = \Pr[E \mid \premises(E)]\Pr[\premises(E)].
$$
Next, let $\premises^*(E)$, the \emph{extended premises of $E$}, be defined as $\premises^*(E) = \premises(E)$ when $q \in Q^0$, and as $\premises^*(E) = \premises(E) \cup \bigcup_{E' \in \premises(E)} \premises^*(E')$ otherwise. In words, $\premises^*(E)$ contains the premises of $E$, and the premises of its premises, and so on. We have that 
$$
\Pr[\premises(E)] = \Pr[\premises^*(E)].
$$
Let $\premises^*[E] = \{E\} \cup \premises^*(E)$. Finally, for $\calE$ a set of contain-events, let $\premises^*(\calE) = \bigcup_{E \in \calE}\premises^*(\calE)$ and $\premises^*[\calE] = \bigcup_{E \in \calE} \premises^*[E]$. 

Say $E\,:\,\alpha \in \vY{}{h}{t}{q} = \reduce(\calZ,p)$ for some (random) set $\calZ$ and some probability $p$. If the premises of $E$ are known to hold then $\alpha \in \calZ$ is known and the occurence of $E$ boils down to a probability-$p$ coin toss to determine whether $\alpha$ is reduced away. No other contain-event can influence that particular coin toss except those whose (extended) premises force the result of that coin toss. Formally, let $\calE$ be a set of contain-events such that $E \not\in \premises^*(E')$ for any $E' \in \calE$, then $E$ is independent of $\calE$ when conditioned on $\premises(E)$:
$$
\Pr[E \text{ and } \calE \mid \premises(E)] = \Pr[E \mid \premises(E)]\Pr[\calE \mid \premises(E)]
$$
As a consequence, the event $\alpha \in \vY{}{h}{t}{q}$ is decomposable in a chain-rule manner as a product of probabilities of reduce-events conditioned on their premises, with one event per node of the derivation tree $\tree^*(\alpha,q)$. Informally, the following lemma says that $\alpha \in \vY{}{h}{t}{q}$ occurs if and only if all independent coin-tosses in the $\reduce$ procedures required to propagate the samples that compose $\alpha$ up in $\tree^*(\alpha,q)$ are successful.

\begin{restatable}{lemma}{decomposeFull}\label{lemma:decomposeFull}
Let $E_q\,:\, \alpha \in \vY{}{h}{t}{q}$. For $q' \in \tree^*(\alpha,q)$, let $\alpha_{q'}$ be the restriction of $\alpha$ to $var(q')$, let $h_{q'}$ be the restriction of $h$ to $\descendants(q')$ and let $E_{q'}\,:\,\alpha_{q'} \in \vY{}{h_{q'}}{t_{q'}}{q'}$. Then
$$
\Pr[E_q] = \prod_{q' \in \tree^*(\alpha,q)} \Pr[E_{q'} \mid \premises(E_{q'})] = \prod_{E' \in \premises^*[E_q]} \Pr[E' \mid \premises(E')].
$$
\end{restatable}
\begin{proof}
We proceed by induction on $q$'s effective height. If $q$'s effective height is $0$ then the premises of $E_q$ consist in the event $\alpha \in \supp(q)$ and thus occur with probability $1$. In this case $\tree^*(\alpha,q)$ consists in a single leaf $(q)$. So $\prod_{q' \in \tree^*(\alpha,q)} \Pr[E_{q'} \mid \premises(E_{q'})] = \Pr[E_q \mid \premises(q)] = \Pr[E_q \mid premises(q)]\Pr[\premises(q)]$. Now suppose $q$'s effective height is $d > 0$ and that the claim holds for all contain-events for nodes of smaller effective height.
\begin{align}\label{eq:decomposeFulleq}
\Pr[E_q] = \Pr[E_q \mid \premises(E_q)]\Pr[\premises(E_q)] 
\end{align}
\begin{itemize}
\item If $q$ is a $+$ node then $\premises(E_q)$ is a contain-event for a child $q'$ of $q$; call it $E_{q'}\,:\,\alpha \in \vY{}{h'}{t'}{q'}$. By induction 
$
\Pr[E_{q'}] = \prod_{q'' \in \tree^*(\alpha,q')} \Pr[E_{q''} \mid \premises(E_{q''})] 
$
so, since $\tree^*(\alpha,q) = (q,\tree^*(\alpha,q'))$ the claim follows from~(\ref{eq:decomposeFulleq}).
\item If $q$ is a $\times$ node $q_1 \times q_2$ then $\premises(E_q)$ is the conjunction of the two contain-events $E_{q_1}\,:\,\alpha_1 \in \vY{}{h_1}{t_1}{q_1}$ and $E_{q_2}\,:\,\alpha_2 \in \vY{}{h_2}{t_2}{q_2}$ By induction 
$
\Pr[E_{q_1}] = \prod_{q'' \in \tree^*(\alpha,q_1)} \Pr[E_{q''} \mid \premises(E_{q''})] 
$
and
$
\Pr[E_{q_2}] = \prod_{q'' \in \tree^*(\alpha,q_2)} \Pr[E_{q''} \mid \premises(E_{q''})] 
$
so, since $\tree^*(\alpha,q) = (q,\tree^*(\alpha_1,q_1),\tree^*(\alpha_2,q_2))$ and since $\Pr[\premises(E_q)] = \Pr[E_{q_1},\,E_{q_2}] = \Pr[E_{q_1}]\Pr[E_{q_2}]$, the claim follows from ~(\ref{eq:decomposeFulleq}).
\end{itemize}
\end{proof}

\noindent The decomposition of Lemma~\ref{lemma:decomposeFull} extends to conjunctions of contain-events.
\begin{restatable}{lemma}{decomposeMany}\label{lemma:decomposeMany}
For $\calE$ a set of contain-events we have
$
\Pr[\calE] = \prod\limits_{E \in \premises^*[\calE]} \Pr[E \mid \premises(E)].
$
\end{restatable}
\begin{proof}
Let $maxHeight(\calE)$ be the largest effective height of nodes targeted by contain-events in $\calE$.
We proceed by induction on $maxHeight(\calE)$. If $maxHeight(\calE) = 0$ then $\calE$ has only events of the form $\alpha \in \reduce(\supp(q),p)$ that are mutually independent and whose premises $\alpha \in \supp(q)$ occur with probability $1$. Formally, $\Pr[\alpha \in \reduce(\supp(q),p)] = \Pr[\alpha \in \reduce(\supp(q),p) \mid \alpha \in \supp(q)]\Pr[\alpha \in \supp(q)] = \Pr[\alpha \in \reduce(\supp(q),p) \mid \alpha \in \supp(q)]$ and
$$
\Pr[\calE] = \prod_{E \in \calE} \Pr[E \mid \premises(E)]
$$
and the claim follows from $\premises^*[E] = \{E\}$ for contain-events of effective height $0$. Now suppose the claim holds for all sets $\calE'$ with $maxHeight(\calE') < i$ (with $i > 0$) and suppose $maxHeight(\calE) = i$. Let $E_1,\dots,E_k$ be the events of effective height $i$ in $\calE$. $E_1$ occurs only if its premises occur:
\begin{align*}
\Pr[\calE] &= \Pr[\calE \cup \premises(E_1)] = \Pr[E_1,\, \calE \setminus \{E_1\} \mid \premises(E_1)]\cdot\Pr[\premises(E_1)]
\end{align*}
Since $E_1$ has maximum effective height in $\calE$, it cannot be an extended premises of any event in $\calE \setminus \{E_1\}$, hence
\begin{align*}
\Pr[\calE] &= \Pr[E_1 \mid \premises(E_1)]\cdot \Pr[\calE \setminus \{E_1\} \mid \premises(E_1)]\Pr[\premises(E_1)]
\\
& =\Pr[E_1 \mid \premises(E_1)]\cdot \Pr[\calE \setminus \{E_1\}\, \premises(E_1)].
\end{align*}
Repeating this argument with $E_2,E_3,\dots,E_k$ gives
\begin{align*}
\Pr[\calE] &= \Pr[E_1 \mid \premises(E_1)]\cdot\Pr[\calE \setminus \{E_1\},\, \premises(E_1)]
\\
&= \Pr[E_1 \mid \premises(E_1)]\Pr[E_2 \mid \premises(E_2)]\cdot\Pr[\calE \setminus \{E_1,\, E_2\},\,  \premises(E_1),\, \premises(E_2)]
\\
&= \Pr[\calE \setminus \{E_1,\dots,E_k\},\,\premises(E_1),\dots,\premises(E_k)]\cdot \prod_{j = 1}^k\Pr[E_j \mid \premises(E_j)]
\end{align*}
Let $\calE' =  \calE \setminus \{E_1,\dots,E_k\} \cup \bigcup_{j = 1}^k \premises(E_j)$. By induction, 
$
\Pr[\calE'] = \prod_{E \in \premises^*[\calE']} \Pr[E \mid \premises(E)]
$. The proof is finished by observing that 
\begin{align*}
\premises^*[\calE'] &= \premises^*[\calE \setminus \{E_1,\dots,E_k\}] \cup \bigcup_{j = 1}^k \premises^*[\premises(E_j)] 
\\
&=  \premises^*[\calE \setminus \{E_1,\dots,E_k\}] \cup \bigcup_{j = 1}^k \premises^*[E_j] \setminus \{E_j\} 
\\
&= \premises^*[\calE] \setminus \{E_1,\dots,E_k\}
\end{align*} 
\end{proof}

\begin{proof}[Proof of Lemma~\ref{lemma:proba_second_order}]
	Let $E \,:\,\alpha \in \vY{}{h}{t}{q}$ and $E'\,:\,\alpha' \in \vY{}{h}{t}{q}$ with $\alpha \neq \alpha'$. Let $T = \tree^*(\alpha,q)$ and $T' = \tree^*(\alpha',q)$. From Lemmas~\ref{lemma:decomposeFull} and~\ref{lemma:decomposeMany}, we have that 
	$$
	\Pr[E,\,E'] = \frac{\Pr[E] \cdot \Pr[E']}{\prod_{\substack{ X \in \premises^*[E] \\ \cap \premises^*[E']}} \Pr[X \mid \premises(X)]}.
	$$ 
	For $\hat q \in T$, let $\alpha_{\hat q}$ be the restriction of $\alpha$ to $var(\hat q)$ and let $h_{\hat q}$ be the restriction of $h$ to $\desc(\hat q)$ and let $E_{\hat q} : \alpha_{\hat q} \in \vY{}{h_{\hat q}}{h(\hat q)}{\hat q}$. We claim that $\premises^*[E] \cap \premises^*[E']$ is exactly $\bigcup_{\hat q \in T \wedge T'} \premises^*[E_{\hat q}]$. First, observe that $\premises^*[E] = \{E_{\hat q} \mid \hat q \in T\}$ and $\premises^*[E'] = \{E'_{\hat q} \mid \hat q \in T'\}$. Now, $E_{\hat q} = E'_{\hat q}$ when $\alpha_{\hat q} = \alpha'_{\hat q}$ for some node $\hat q$ in both $T$ and $T'$, we have that $T$ and $T'$ below $\hat q$ both equal $\tree^*(\alpha_{\hat q},\hat q)$. So $\hat q$ is either in $T \wedge T'$ or is a descendant of a node in $T \wedge T'$. So the events $X$ in $\premises^*[E] \cap \premises^*[E']$ are exactly the $E_{\hat q}$ for $\hat q$ in $T \wedge T'$ or below $T \wedge T'$. Hence $\premises^*[E] \cap \premises^*[E'] = \bigcup_{q'\in T \wedge T'} \premises^*[E_{q'}]$. 
	\begin{align*}
		\prod_{\substack{ X \in \premises^*[E] \\ \cap \premises^*[E']}} \Pr[X \mid \premises(X)] &= \prod_{q' \in T \wedge T'} \prod_{X \in \premises^*[E_{q'}]} \Pr[X \mid \premises(X)] 
		\\
		&= \prod_{q' \in T \wedge T'} \Pr[E_{q'}] \tag{Lemma~\ref{lemma:decomposeFull}}
		\\
		&= \prod_{q' \in T \wedge T'} h(q') \tag{Lemma~\ref{lemma:proba_first_order}}
	\end{align*}
	Thus $\Pr[\alpha \in \vY{}{h}{t}{q},\, \alpha' \in \vY{}{h}{t}{q}] = \Pr[E,\,E'] = \frac{\Pr[E]\Pr[E']}{\prod_{q' \in T \wedge T'} h(q')} = \frac{t^2}{\prod_{q' \in T \wedge T'} h(q')}$ and the first part of Lemma~\ref{lemma:proba_second_order} is shown. 
	
	For the second part, suppose $q$ is a $+$ node. there are two unique children $q_i$ and $q_j$ of $q$ such that
	\begin{align*}
		\Pr&[\alpha \in \vY{}{h}{t}{q},\, \alpha' \in \vY{}{h}{t}{q}]\\
		= &\Pr\left[\alpha \in \reduce\left(\vZ{}{h}{q},\frac{t}{t_{\min}}\right),\, \alpha' \in \reduce\left(\vZ{}{h}{q},\frac{t}{t_{\min}}\right) \, \bigg|\, \alpha \in \vZ{}{h}{q}\, \alpha' \in \vZ{}{h}{q}\right]
		\cdot \Pr[\alpha \in \vZ{}{h}{q},\, \alpha' \in \vZ{}{h}{q}]
	\end{align*}
	The two reduce events are conditioned on their premises and thus 
	\begin{align*}
		\Pr&[\alpha \in \vY{}{h}{t}{q},\, \alpha' \in \vY{}{h}{t}{q}]\\
		= &\Pr\left[\alpha \in \reduce\left(\vZ{}{h}{q},\frac{t}{t_{\min}}\right)  \bigg|\,  \alpha \in \vZ{}{h}{q}\right]\Pr\left[\alpha' \in \reduce\left(\vZ{}{h}{q},\frac{t}{t_{\min}}\right) \bigg|\,  \alpha' \in \vZ{}{h}{q}\right] \Pr[\alpha \in \vZ{}{h}{q},\, \alpha' \in \vZ{}{h}{q}]
		\\
		= &\frac{t^2}{t^2_{\min}}\Pr[\alpha \in \vZ{}{h}{q},\, \alpha' \in \vZ{}{h}{q}]
	\end{align*}
	It suffices to use that $\Pr[\alpha \in \vY{}{h}{t}{q},\, \alpha' \in \vY{}{h}{t}{q}] = \frac{t^2}{\prod_{q' \in T \wedge T'} h(q')}$ to conclude.
\end{proof}

\end{document}